\documentclass[paper]{article}

\usepackage{amsmath, amsfonts, amssymb, amsthm, latexsym}

\usepackage{authblk}

\usepackage{todonotes}

\usepackage{cancel}

\usepackage{color}

\usepackage[toc,page]{appendix}
\usepackage{braket}
\usepackage{bm}

\usepackage{graphicx}

\usepackage{tikz}

\usepackage{subfigure}

\usepackage{soul}

\usepackage{comment}
\usepackage{verbatim}
\usepackage{ulem}    % For crossing out sentences

\let\Re\undefined

\let\Im\undefined

\DeclareMathOperator{\Re}{Re}

\DeclareMathOperator{\Im}{Im}

\DeclareMathOperator{\supp}{supp}

\newtheorem{theorem}{Theorem}[section]

\newtheorem{lemma}[theorem]{Lemma}

\newtheorem{proposition}[theorem]{Proposition}

\newtheorem{corollary}[theorem]{Corollary}

\theoremstyle{definition}

\newtheorem{definition}[theorem]{Definition}

\newtheorem{remark}[theorem]{Remark}

\newtheorem{problem}{Riemann-Hilbert Problem}[section]

\numberwithin{equation}{section}

\definecolor{shadecolor}{rgb}{0.95, 0.95, 0.86}

\usetikzlibrary{decorations.markings}

\def\be{\begin{equation}}

\def\ee{\end{equation}}

\def\bi{\begin{itemize}}

\def\ei{\end{itemize}}

\def\bea{\begin{eqnarray}}

\def\eea{\end{eqnarray}}

\def\bl{\begin{lemma}}

\def\el{\end{lemma}}

\def\bd{\begin{definition}}

\def\ed{\end{definition}}

\def\bp{\begin{proposition}}

\def\ep{\end{proposition}}

\def\br{\begin{remark}}

\def\er{\end{remark}}

\def\bt{\begin{theorem}}

\def\et{\end{theorem}}

\def\bc{\begin{corollary}}

\def\ec{\end{corollary}}

\def\ra{\rightarrow}

\newcommand{\G}{\Gamma}

\renewcommand{\O}{\Omega}

\renewcommand{\k}{\varkappa}

\renewcommand{\d}{\delta}

\newcommand{\D}{\Delta}

\newcommand{\e}{\varepsilon}

\renewcommand{\o}{\omega}

\newcommand{\g}{\gamma}

\renewcommand{\part}{\partial}

\newcommand{\tB}{\tilde B}

\newcommand{\Res}{\text{Res} \,}

\def\Rscr{\mathcal R}

\def\le{\left}

\def\ri{\right}

\def\C{{\mathbb C}}

\def\R{{\mathbb R}}

\def\N{{\mathbb N}}

\def\a{\alpha}

\def\b{\beta}

\def\g{\gamma}

\def\gt{\hat\gamma}

\def\e{\varepsilon}

\def\m{\mu}

\def\l{\lambda}

\def\n{ {\nu}}

\def\1{{\bf 1}}

\def\r{\rho}

\def\s{ {\sigma}}

\def\t{ {\tau}}

\def\th{ {\theta}}

\def\x{\xi}

\def\z{\zeta}

\def\hf{\frac{1}{2}}

\newcommand{\Iscr}{\mathcal I}

\begin{document}

\title{Recent developments in spectral theory of the focusing NLS soliton and breather gases: the thermodynamic limit of  average densities, fluxes and certain meromorphic differentials; periodic gases}

\author[1]{Alexander Tovbis}

\author[2]{Fudong Wang}

\affil[1]{University of Central Florida, Orlando FL, U.S.A.,

Alexander.Tovbis@ucf.edu}

\affil[2]{University of Central Florida, Orlando FL, U.S.A., fudong@math.ucf.edu}

\footnotetext{The work of  is supported  by the  NSF grant DMS-2009647.}

\maketitle

\begin{abstract}

In this paper we consider soliton and breather gases for one dimensional integrable focusing Nonlinear 
Schr\"{o}dinger Equation (fNLS). We
derive  average densities and fluxes for such gases by studying the thermodynamic limit
of the fNLS finite gap solutions.  Thermodynamic limits of quasimomentum, quasienergy and their connections with the corresponding $g$-functions were also established.

We then introduce
the notion of periodic fNLS gases and calculate for them the  average densities, fluxes and  thermodynamic
limits  of meromorphic differentials. Certain accuracy estimates of the obtained results are  also included.  

Our results constitute another  step
towards the mathematical foundation for the spectral  theory of fNLS soliton and breather gases that
appeared in work of G. El and A. Tovbis, Phys. Rev. E, 2020.

\end{abstract}

\tableofcontents

\section{Introduction and statement of results}\label{sec-intro}

Solitons and breathers represent well known  localized solutions in many integrable
systems. Due to their ``elastic'' interaction, they can also be viewed as ``quasi-particles" of complex
statistical objects called soliton and breather gases. The nontrivial relation 
between the integrability and randomness in these gases falls within the framework
of “integrable turbulence”, introduced by V. Zakharov in \cite{Za09}. The latter was
motivated by the complexity of many nonlinear wave phenomena in physical
systems that can be modeled by integrable equations. In view of the growing
evidence of wide spread presence of the integrable gases (fluids, nonlinear optical media, etc.), see \cite{ElTovbis}, \cite{El-review} and references therein, they present a fundamental interest for nonlinear science.

In this paper we consider soliton and breather gases for the fNLS
\begin{equation}  \label{NLS}
i  \psi_t +  \psi_{xx} +2 |\psi|^2 \psi=0,
\end{equation}
where $x,t\in \R$ are the space-time variables and  $\psi:\R^2 \to \C$ is the unknown complex -valued function. 

One of the central objects in the spectral theory of soliton and breather gases is the nonlinear dispersion relations (NDR), defining the density of states (DOS) $u(z)$ and its temporal analog (density of fluxes) $v(z)$. The NDR 
for the fNLS breather  gas are defined by geometry (a Schwarz symmetric branchcut (band) $\g_0$ and a
compact  $\G^+\subset \C^+\setminus \g_0$) and a spectral scaling function $\s(z)\geq 0$, $z\in \G^+$. The compact $\G^+$ is the locus of accumulation of shrinking spectral bands in the {\it thermodynamic limit} of some finite gap solutions of \eqref{NLS}, whereas $\s(z)$ represents the ratio of scaled logarithmic bandwidth and the density of the bands.
%is often called relative scaled bandwidth.
Description of the thermodynamic limit as well as 
further discussion and some details about the derivation of the NDR can be found in Section \ref{sec-background}.

The NDR for the solitonic component of the fNLS breather  gas have the form (\cite{ElTovbis})
\begin{multline} \label{dr_breather_gas1}
\int_{\G^+}
\le[\log\le| \frac{w-\bar z}{w-z}\ri|+ \log\le|\frac{R_0(z)R_0(w)+z w -\d_0^2}
{R_0(\bar z)R_0(w)+\bar z w-\d_0^2}\ri|\ri]   u(w) d\lambda(w)
+\sigma(z)u(z) \\
= \Im R_0(z), 
\end{multline}
\begin{multline}  \label{dr_breather_gas2}
 \int_{\Gamma^+}
\left[\log\left| \frac{w-\bar z}{w-z}\right|+\log\left|\frac{R_0(z)R_0(w)+ z w-\delta_0^2}
{R_0(\bar z)R_0(w)+ \bar z w-\delta_0^2}\right|\right] v(w) d\lambda(w)
+\sigma(z)v(z) \\
= - 2\Im[z R_0(z)],
\end{multline}
 where  $R_0(z)=\sqrt{z^2-\delta_0^2}$ with the branchcut $\g_0=[-\d_0,\d_0]\subset i\R$ being the exceptional (permanent) band,
 %$\G^+\subset \C^+\setminus \g_0$ is a compact, 
 $z\in\Gamma^+$, $\s$ is a continuous non negative function on $\G^+$
and $\lambda$ is some reference measure on $\Gamma^+$ that
reflects the density of accumulating in the thermodynamic limit
bands. In the case of $\d_0=0$ the fNLS breather gas reduces to the 
fNLS soliton gas with the NDR
\begin{align} \label{dr_soliton_gas1}
 %\frac 1{\pi}
 \int _{\Gamma^+}\log \left|\frac{w-\bar z}{w-z}\right|
u(w)d\lambda(w)+\sigma(z)u(z)& = \Im z,\\
\label{dr_soliton_gas2}
 \int _{\Gamma^+}\log \left|\frac{w-\bar z}{w-z}\right| v(w) d\lambda(w)+  \sigma(z) v(z)& = -4 \Im z\Re z.
\end{align}
%for unknown functions $u$ and $v$ respectively, 
We routinely assume that $\lambda$ is the standard area measure if $\G^+$ (or its connected component) is a
 $2D$ region, or the arclength measure if $\Gamma^+$ (or its connected component) is a contour.

 It has to be noted that the properties of $u,v$ are completely defined by $\g_0$, the compact $\G^+$ and 
 the function $\s$.
 %which has the meaning of the (limiting) relative scaled bandwidth, see Section \ref{sec-background}. 
Rigorous mathematical analysis of equations  \eqref{dr_soliton_gas1}-\eqref{dr_soliton_gas2} was reported
in \cite{KT2021}. It was shown there (\cite{KT2021},  Corollary 1.7) that, subject to certain mild conditions, each of  these equations has
a (unique) solution and, moreover, the solution $u(z)$ of \eqref{dr_soliton_gas1} is non negative. Moreover, in the case of a 1D compact $\G^+$, it was  shown that $u,v$ inherit some smoothness from $\G^+, \s$.
Most of these results (but not $u\geq 0$) will hold if we replace the the right hand side of \eqref{dr_soliton_gas1} with any sufficiently smooth (at least two times continuously differentiable) function. Similar results are expected for more general
equations \eqref{dr_breather_gas1}-\eqref{dr_breather_gas2}
but this work have not been completed yet.

In the present paper we will assume the existence and uniqueness  of solutions $u,v$, where $u\geq 0$, to  \eqref{dr_breather_gas1}-\eqref{dr_breather_gas2}. 
{\it To be more precise,  we will assume that the solutions $u,v$ of  \eqref{dr_breather_gas1}-\eqref{dr_breather_gas2}, as well as that
of \eqref{dr_soliton_gas1}-\eqref{dr_soliton_gas2}, belong to $L^1(\G^+)$.} This assumption is not too restrictive: for example, 
it was proved in \cite{KT2021} that $\s u \in L^1(\G^+)$ (with respect to the reference measure $\l$) for any $u$ satisfying 
\eqref{dr_soliton_gas1}; moreover,  the requirement $\s>0$ on $\G^+$ implies the continuity of $u$. 
Sometimes, 
the solutions $u,v$ will be assumed to have certain smoothness provided $\G^+$ and $\s$ possess the appropriate smoothness.

Equations  \eqref{dr_breather_gas1}-\eqref{dr_breather_gas2} were derived in  \cite{ElTovbis} as thermodynamic limits of the
two $N\times N$ linear systems, see \eqref{WFR}, 
of equations satisfied (respectively) by  the solitonic  wavenumbers  $k_m$ and the frequences $\o_m$,
$|m|=1,\dots, N$, of a finite gap (nonlinear multi phase) solution $\psi_N$ to \eqref{NLS}. A finite gap  solution $\psi_N$ is defined by a
hyperelliptic Schwarz symmetrical Riemann surface $\mathfrak R_{N}$  of the genus $2N$ together with a collection of $2N$ initial phases,
%in the (thermodynamic) limit $N\ra\infty$ 
see Section \ref{sec-background} for some details. As it is known (\cite{ForLee}, \cite{BBIM}),  $k_m$ and $\o_m$ are defined as 
periods of certain 2nd kind meromorphic differentials $dp_N,dq_N$ on  $\mathfrak R_{N}$, called quasimomentum and quasienergy
respectively. In fact, the  $N\times N$ systems  of linear equations defining $k_m, \o_m$ are imaginary parts of Riemann bilinear identities for the  normalized holomorphic differentials on  $\mathfrak R_{N}$ and  $dp_N,dq_N$  respectively.

The  differentials $dp_N,dq_N$ and their higher analogs are interesting and important objects,
%not only because of their periods but also 
since their Laurent expansions at $z=\infty$ contain valuable information about  $\psi_N$. For example,
\be
\frac{dp_N}{dz}=1-\sum_{m=1}^\infty \frac{I_{m,N}}{z^{m+1}},
\ee
where $I_{m,N}$, $m\in\N$,  represent average densities of $\psi_N$  that are conserved with the $t$ evolution (\cite{ForLee}).
As it is well known, the coefficients  $I_{m,N}\in \R$. Similar  Laurent coefficients for $\frac{dq_N}{dz}$ and for the corresponding  higher meromorphic   differentials  on $\mathfrak R_{N}$ represent the average fluxes of the fNLS and higher NLS hierarchy flows.

One of the main subjects of this paper is the derivation of the average densities    $I_{m,N}$ as well as of  the corresponding  average fluxes in 
the thermodynamic limit $N\ra\infty$, followed by the calculation of the thermodynamic limits of $\frac{dp_N}{dz}$ and other meromorphic differentials.
We also provide the corresponding error estimates.

Let
$d_m\d_0^{m+1}$, $m\in\N$,
denote the coefficient of $z^{-m}$ in the Laurent expansion of $R_0(z)$ at infinity, where $R_0(z)$ is defined below  \eqref{dr_breather_gas2}. That is,  
 $d_m=0$ when $m$ is even and 
\be\label{dm-int}
d_m= -\frac{1}{m}\frac{m!!}{(m+1)!!}
\ee 
when $m$ is odd. Note if we set $m=2n+1$, then $d_{2n+1}=-2^{-2n}C_n$, where $C_n$ is the Catalan number (see DLMF 26.5.1). We start with the following theorem for the fNLS breather gas.

\bt\label{theo-ImN}
(i) Fix some $m\in\N$. Then  for a sufficiently large $N$ in the thermodynamic limit of the 
breather gases we have 

\begin{align}\label{br-avg-den}
% &&I_{m,N}=\le (\frac{1}{2\pi i}\sum_{|j|=1}^N \red{U_{j,N}}\oint_{B_j}\z^{m-1}d\z \ri )(1+O(N^\hf\delta)), \cr
     &&I_{m,N}=\le[\frac{m}{2\pi i}\sum_{|j|=1}^N U_{j,N}\oint_{B_j}\frac{[\z^{m-1}R_0(\z)]_+d\z}{R_0(\z)}+md_m\d_0^{m+1}\ri](1+O(N^\hf\delta)),
\end{align}
where:  the polynomial $[f(\z)]_+$ consists of the non-negative powers   of the Laurent series expansion of $f(\z)$ at infinity; 
$\d>0$ denotes the largest length among all  the shrinking bands;
$B_j$ are the cycles on $\mathfrak R_{N}$ shown at Figure \ref{Fig:Cont}, and; $U_{j,N}=\hf k_j$ with the  wavenumbers $k_j$ described in Section \ref{sec-background}.

(ii) If the measures $\l_N:=\sum_{j=1}^N\frac{U_{j,N}}{2\pi}\d(z-z_{j,N})$, where $z_{j,N}\in\G^+$ denotes the center of the $j$th bands of $\mathfrak R_{N}$, $j=1,\dots,N$, and $\d(z)$ denotes the delta-function,
%centered at $z\in\G^+$,
are weakly* convergent to the measure $u(z)d\l(z)$ on $\G^+$, where $u(z)$ solves \eqref{dr_breather_gas1}, then the thermodynamic limit of   $I_{m,N}$ is given  by
\begin{align}\label{inv-breather-t1}
I_m:= \lim_{N\ra\infty}I_{m,N}=2m\int_{\G^+}u(\z)\Im{F_m(\z)}d\l(\z)+md_m\d_0^{m+1},
 %\text{ as } N\rightarrow\infty.
 \end{align}
where 
\begin{align}\label{def-F-m}
   F_m(z)=\int_0^z \frac{[\z^{m-1}R_0(\z)]_+}{R_0(\z)}d\z. 
\end{align}

\et

The proof of Theorem \ref{theo-ImN} can be found in Section \ref{sec-cons-dens}.
As it was mentioned above, soliton gas can be obtained from breather gas by shrinking the exceptional band $\g_0$, that is, by taking 
limit $\d_0\ra 0$. In this case the statements of Theorem \ref{theo-ImN} are given in the following corollary.

\bc\label{cor-sol-inv}
In the case of the fNLS
soliton gas, i.e. when $\g_0$ is one of the shrinking bands, the equations \eqref{br-avg-den} and \eqref{inv-breather-t1} become
%\begin{align}
\bea\label{inv-soliton-t1}
 &&I_{m,N}=\le (\frac{m}{2\pi i}\sum_{|j|=1}^N {U_{j,N}}\oint_{B_j}\z^{m-1}d\z \ri )(1+O(N^\hf\delta)), \cr
&&I_m =2\int_{\G^+}u(\z)\Im \z^m d\l(\z)
 %\text{ as } N\rightarrow\infty.
 \eea
 %\end{align}
 respectively, 
 where $\d$ denotes the length of the largest band.
\ec

%\red{[Perhaps, the remark should be converted into a theorem. Also, it's a good idea to put $I_0$.]} 

Similar results for averaged fluxes and their higher time analogs are discussed in Subsection \ref{sect-higher}.

To estimate  the rate of convergence of $I_{m,N}\ra I_m$
for a fixed $m\in\N$ one needs to know the rate of convergence of  the measures $\l_N$ in Theorem \ref{theo-ImN}, 
to $u d\l$ on $\G^+$. This question was considered in Section \ref{sect-approx-u} for the soliton gas with a 1D compact (contour) $\G^+$ under   the  following additional assumptions: (i) $\s>0$ on $\G^+$; (ii) the solution $u(z)$ of \eqref{dr_soliton_gas1} is $\a$-H\"older continuous on $\G^+$ with some $\a\in(0,1]$, and; (iii) the density $\varphi$ of the centers of the shrinking bands provides  $O(N^{-\beta})$, $\beta>0$, approximation for the arclength measure $\l$, see Section \ref{sec-background}, including \eqref{varphi-def}, for details.

Under these assumptions the fact that  the measures $\l_N$ provide $O(N^{-\varrho})$ approximations  to $u d\l$ on $\G^+$, where $\varrho=\min\{\a,\beta\}$, follows directly from Theorem \ref{th-v>0}.

\bc
 Under the conditions of Theorem \ref{th-v>0}, 
 %if $\G^+$ is a 1D compact and $u(z)\in C(\G^+)$ then 
 \be\label{I_m-est}
 |I_m-I_{m,N}|=O(N^{-\varrho})
 \ee 
 for any $m\in\N$ provided $\varrho<1$. The accuracy in \eqref{I_m-est} is $O(\frac{\ln N}N)$ if $\varrho=1$.
 \ec

Theorem \ref{th-v>0} also implies that, in the thermodynamic limit,  $NU_{j,N}$ is approximated by $\frac{1}{\pi}\varphi(z_{j,N})u(z_{j,N})$ with the accuracy $O(N^{-\varrho})$ as $N\ra \infty$, thus providing the error estimate of the transition  from the discrete NDR \eqref{WFR} (systems of linear equations) to its continuous counterpart \eqref{dr_soliton_gas1}-\eqref{dr_soliton_gas2} (integral equations). The accuracy is $O(\frac{\ln N}N)$ in the case $\varrho=1$.

Next,  we consider the thermodynamic   $\lim_{N\ra\infty}\frac{dp_N}{dz}$ of the density of the quasimomentum 
meromorphic differential, which we define as 
\be\label{def-dpdz}
\frac{dp}{dz}:=1{-}\sum_{m=1}^\infty \frac{I_{m}}{z^{m+1}} 
 \ee
 in a neighborhood of $z=\infty$ on the main sheet. It is clear that a requirement such as  $u\in L^1(\G^+)$ is sufficient 
 for the series in \eqref{def-dpdz} to be convergent in a neighborhood of $z=\infty$.
In the {\it remaining part of the Introduction we  restrict ourselves to the case when the compact  $\G^+$ is a contour (or a collection of contours).} 
Moreover, we: call the connected parts of $\G=\G^+\cup \bar\G^+$ ``superbands";   define $\mathfrak R_{\infty}$ as the limiting hyperelliptic Riemann surface  to the sequence $\mathfrak R_{N}$ , $N\in\N$,
	where the    branchcuts of $\mathfrak R_{\infty}$ coincide with  the superbands of  $\G$.
	% In the case of  only one.
 
 The general complex nonlinear dispersion relations for the fNLS breather gas were obtained  in 
 \cite{ElTovbis}.  The imaginary parts of the general complex NDR form the NDR \eqref{dr_breather_gas1}-\eqref{dr_breather_gas2} for the solitonic components of the fNLS breather gas, whereas their real parts form  the NDR for the corresponding carrier components.
 The first general complex NDR 
equation has the form:
\begin{multline} \label{dr_breather_gas1-gen}
i\int_{\G^+}
\le[\log \frac{\bar w-z}{w-z}+ \log\frac{R_0(z)R_0(w)+z w -\d_0^2}
{R_0( z)R_0(\bar w)+ z \bar w-\d_0^2}+ i\pi \chi_z(w)\ri]   u(w) |d w| \\
+i\sigma(z)u(z) 
= R_0(z)+\tilde u(z) 
\end{multline}
(\cite{ElTovbis}, equation (25)),
where: (i) $\tilde u(z)$  is the ``carrier density of states" function that is
defined  as a smooth 
%Schwarz symmetrical 
interpolation of the carrier wavenumbers (see Section \ref{sec-background}), i.e, $\tilde u(z)$ interpolates the values $\hf\tilde k_j$ at $z_j\in\G^+$, $j=1,\dots,N$, and;
(ii)
$\chi_z(\m)$ is the indicator function of the arc $(z_\infty,z)$ of $\G^+$. Here $z_\infty$ denotes the beginning of the oriented curve $\G^+$.

Let us consider 
first  the soliton gas. 
To further simplify the situation, take $\G^+\subset i\R^+$. Then we prove that
\be\label{dp/dz-vert}
\frac{dp}{dz}=1{-}2\pi iC_\G[ u] \quad\text{ on $\bar\C\setminus \G$ },
	\ee
where: $\G=\G^+\cup \G^-$, $\G^-=\overline{\G^+}$; the density of states $u$ (see \eqref{dr_soliton_gas1}) has  odd (anti-Schwarz symmetrical) continuation to $\G^-$, and; $C_{\G}$ denotes the Cauchy transform on $\G$. Thus, $\frac{dp}{dz}$ is analytic in $\bar\C\setminus \G$ and has a jump $-2\pi i u(z)$ on $\G$.

 It is an interesting observation that  $ \frac{dp}{dz}$ defined by \eqref{dp/dz-vert} in the case of a condensate (i.e., $\s\equiv 0$)   coincides with the denisty of the quasimomentum differential  $ \frac{dp_\infty}{dz}$ of $\mathfrak R_{\infty}$. Indeed, it was proved in \cite{KT2021}, Theorem 6.1, that in the case of a bound state (i.e.,  $\G\subset i\R$) soliton condensate  the DOS
$u(z)=\frac{i}{\pi}\left(\frac{dp_\infty}{dz}\right)_+$,  where 
$\left(\frac{dp_\infty}{dz}\right)_+$ denotes the boundary value on the positive side of $\G$.  Then it is easy to check that 
$\frac{dp}{dz}$ and $\frac{dp_\infty}{dz}$ have the same jump $-2\pi i u(z)$ on $\G$ and the same asymptotics at $z=\infty$ on $\bar \C$. Since $u(z)$ has at most $z^{-\hf}$ singularities on $\G$ (at the endpoints), we proved that for a bound state condensate  
\be\label{bstate-cond}
\frac{dp}{dz}\equiv \frac{dp_\infty}{dz}.
\ee
As discussed in Section \ref{sec-background}, this result is also valid for the KdV soliton condensates.

In general, the requirement $\G^+\subset i\R^+$ from  \eqref{dp/dz-vert} can be removed if we introduce
 $\breve u(z): = u(z) e^{-i\th(z)}$, 
	where $\th(z)=\arg dz$ along  $\G^+$ at $z\in\G^+$ traverses it in the positive direction. Then (Theorem \ref{the-lim-dpdz}) %Subsection \ref{sec-dens-of-diff}, we proved that 
 equation \eqref{dp/dz-vert} becomes
 \be\label{dp/dz}
\frac{dp}{dz}=1{+}2\pi C_\G[\breve u] \quad\text{ on $\bar\C\setminus \G$ },
	\ee
	so that \eqref{dp/dz-vert}  is a particular case of \eqref{dp/dz} .
 
Let  $\G^+$ be a given contour.
 Then equations \eqref{dr_soliton_gas1} and \eqref{dp/dz} allow one to express any two of the functions $\frac{dp}{dz}$, $\s(z)$, $u(z)$
 in terms of the remaining one. Similar results can be obtained for the quasienergy density $\frac{dq}{dz}$  as well as for the thermodynamic limits 
 of the higher meromorphic differentials on $\mathfrak R_{N}$.

Let $\g\in\C$ be a simple Schwarz symmetric curve that consisits of the superbands and connecting them gaps.
Together with the density $\frac{dp}{dz}$ we consider  its ``antiderivative" $2g_x(z)$, satisfying $\frac d{dz} [2g_x(z)]=1-\frac{dp}{dz}$. The notation $2g_x$ emphasizes connection with the $g$-functions associated with the Riemann-Hilbert Problems (RHP) for the finite gap solutions $\psi_N$ of the fNLS \eqref{NLS}, see Section \ref{sec-background}. Indeed, it follows from \eqref{dp/dz} that
 \begin{align}
\label{2gx-lim-int}
2g_x(z):=z-\int_0^zdp=
-\int_\G u(\m)\arg\m|d\m|-i\int_\G\ln (\m-z)u(\m)|d\m|,
   \end{align}
so that 
\be\label{2gx-infty}
2g_x(\infty)=-\int_\G u(\m)\arg\m|d\m|.
\ee
If $0\in\G$, the integral in \eqref{2gx-lim-int} start from the positive side of $\G$, i.e., from $0\in\G_+$.
Note that  $2g_x$ is analytic in $\bar\C\setminus \G$ but,  
	if $\G$ consists of several superbands, not necessarily  single valued there. However, similarly to the g-function from Section \ref{sec-background}, $2g_x(z)$ is analytic and single valued in $\bar\C\setminus \g$.

Next, in Theorem \ref{the-lim-dpdz-breath}, we obtain a similar to
\eqref{2gx-lim-int} expression for
the thermodynamic limit of $2g_x(z)$ for the breather gas. As a consequence, combining it with \eqref{dr_breather_gas1-gen},
%equation (25) from \cite{ElTovbis}, which defines both the solitonic and the carrier DOS $u$ and $\tilde u$ respectively, 
we derive the following corollary for both soliton and breather gases.

\bc\label{cor-gxpm-int}
For any $z\in\G$ we have
\be\label{gxpm-uu-int}
g_{x+}(z)+g_{x-}(z)=\tilde u(z)-i\s(z)u(z)+z, 
%\quad z\in \G^+,
\ee
where $\tilde u(z)$ is defined in \eqref{dr_breather_gas1-gen}. That is,
\be\label{su-sol}
\s(z)u(z)=\Im (z- 2g_x(z)),\quad \tilde u(z)=\Re[g_{x+}(z)+g_{x-}(z)-z].
\ee
%for any $z\in\G$. Here we used the fact that $\Im g_x(z)$ is a continuous function in $\C$.
\ec

It is shown in Section \ref{sec-cons-dens} that the above mentioned results regarding the DOS $u(z)$ and the quasimomentum differential $dp$ after appropriate reformulation will be valid for the density of fluxes $v(z)$, see \eqref{dr_breather_gas2}, and the quasienergy
differential $dq$. They will also be valid for higher meromorphic differentials associated with the flows of the fNLS hierarchy.
Moreover, the corresponding results will still be valid if we replace the right hand side in \eqref{dr_soliton_gas1}
by ${xz+2tz^2+f_0(z)}$, where $f_0(z)$ is a sufficiently smooth Schwarz symmetrical function, defining the initial phases of a particular family of finite gap solutions $\psi_N$, $N\ra\infty$,
to the fNLS \eqref{NLS}, that is, defining a particular realization of a soliton gas.
Various sets of sufficient conditions on $f_0(z)$ (or even $f_0(z,N)$) will be studied elsewhere.
It is expected that similar results will be also valid for breather gases.

 In Section \ref{sec-per} we introduce a new class of {\it periodic} soliton and breather fNLS  gases. These are the gases whose density of bands $\varphi(z)$ and the scaled bandwidth $\nu(z)$, see Section \ref{sec-background} for details, are determined through  the semiclassical  ($\e\ra 0$) limit of the fNLS spectral data  with some  periodic potential $q(x)$.
 
 In this paper we assume that $q(x)$ is an even, nonnegative, continuous and  single humped $2L$-periodic, $L>0$,  function, such
 that $M=q(0)$ and $m=q(L)\geq 0$ are the maximum and the minimum of $q$  respectively. In particular, $q(x)$ must be strictly monotonically decreasing on $[0,L]$.
 It follows from results of \cite{BioOre2020} and \cite{BOT2021} that: 
 \begin{itemize}
     \item the bands of the Lax spectrum of the corresponding Zakharov-Shabat operator
 are  confined to the ``cross'' $\R\cup [-iM,iM]$ in the limit $\e\ra 0$;
 \item $\R$ is a single band and if  $m>0$ then $[-im,im]$ is also (asymptotically) a single band, and;
 \item   
 the leading $\e\ra 0$ order of the Floquet discriminant (the trace of the monodromy
 matrix for Zakharov-Shabat ) on $z\in[im,iM]$ is
 \be\label{w-lam}
 w(\l)=2\cos\frac{S_1(\l)}{\e}\cosh \frac{S_2(\l)}{\e},
 \ee
  where $\l=-z^2$ 
  %$m\leq \l\leq M$ 
  and
  \be\label{S12I1}
S_1(\l)=\int_{-q^{-1}(|z|)}^{q^{-1}(|z|)}\sqrt{[q^2(x)-\l]}dx,~~~~~~\qquad S_2(\l)=\int_{q^{-1}(|z|)}^{L}\sqrt{|q^2(x)-\l|}dx.
\ee
 
 \end{itemize}
 %(i)  (ii) 

% \red{[I suggest to remove $\l$ variable.]}
The latter  statement was obtained in \cite{BioOre2020} through formal WKB arguments.  Some rigorous results about the localization of the spectral
bands can be found in \cite{BOT2021}. 

%The results of \cite{BioOre2020} indicate that 
Equations \eqref{w-lam}, \eqref{S12I1} show that if we take
$\G^+=[im,iM]$, the number of spectral bands will grow like $O(1/\e)$ whereas the size of these bands will decay exponentially in $1/\e$.
Thus, the semiclassical fNLS evolution of the periodic potential $q(x)$ should provide a realization of a soliton ($m=0$) or a breather ($m>0$) fNLS gas. Spectral characteristics of such gases, namely,
the density of bands $\varphi(z)$ and the scaled logarithmic bandwidth $\nu(z)$, $z\in \G^+$, see Section \ref{sec-background}, given by
  \be\label{per-nu-eta-int}
 \varphi(z)=\frac{-iz\int_{0}^{q^{-1}(|z|)}\frac{dx} {\sqrt{q^2(x)+z^2}}}{\int_0^L\sqrt{q^2(x)-m^2}dx }, \qquad
\nu(z)=
\frac{\pi \int_{q^{-1}(|z|)}^{L}\sqrt{|q^2(x)+z^2|}dx}{2\int_0^L\sqrt{q^2(x)-m^2}dx}
\ee
respectively, are calculated 
 in Section \ref{sec-per-sol}. 
 
The above calculations serve as  a motivation   to call fNLS soliton gases with $\G^+=[0,iM]$  (the case of $m=0$) and $\varphi, \nu$ given by \eqref{per-nu-eta-int} as  {\it periodic}
fNLS soliton gases.  The gases corresponding  to $m>0$ with the same $\varphi, \nu$  and $\G^+=[im,iM]$ will be called {\it periodic} fNLS  breather gases with the permanent 
band $\g_0=[-im,im]$. 

The density of states $u(z)$ of a periodic gas must be proportional 
to its density of bands $\varphi(z)$, a fact that is intuitively clear  since the fNLS evolution of a periodic potential remains periodic for all times. 
%because in the periodic environment the number of waves leaving a period is equal to the number of 
In Section \ref{sec-per}, see Theorems \ref{the-per-sol-gas}, \ref{the-per-breath-gas}, we prove this fact and calculate the coefficient of proportionality. Namely, the density of states $u(z)$, $z\in \G^+$, of a periodic gas with 
a potential $q(x)$ is given by
  \be\label{uz-per}
  u(z)= \frac{1}{\pi L}\int_0^L\sqrt{q^2(x)-m^2}dx \cdot\varphi(z)=
  \frac{-iz}{\pi L}\int_{0}^{q^{-1}(|z|)}\frac{dx} {\sqrt{q^2(x)+z^2}}.
  \ee

It is worth noting that for the considered periodic gases the density of fluxes $v(z)\equiv 0$, as the right hand sides
of \eqref{dr_soliton_gas2}, \eqref{dr_breather_gas2} are identical zeros.

As an example, consider the potential $q(x)=Q$ on $[0,L]$ and assume ${q^{-1}(|z|)}=L$ for any $z\in \G^+$, where $\G^+=[0,iQ]$. Then \eqref{uz-per} immediately yields 
\be\label{bstate-con}
u(z)=\frac{-iz}{\pi  \sqrt{Q^2+z^2}}, \quad z\in [0,iQ],
  \ee
which is the well known DOS of the genus zero  bound state soliton condensate, see \cite{ElTovbis}, Example 1. Of course, $q(x)=Q$ does not satisfy the monotonicity assumption we put on $q(x)$, but that can be mitigated by modifying  $q(L)=0$ just at one point $x=L$ and then approximating (with respect to any integral norm) the ``modified" discontinuous $q(x)$ by smooth and strictly monotonic functions. That will justify $u(z)$ given by \eqref{bstate-con}. Note that the corresponding $\nu(z)\equiv 0$.

On the other hand, we can set $q(L)=m$ for any $m\in(0,Q)$ and repeat the previous reasoning. We will still have the same $u(z)$ given by \eqref{bstate-con} but only for $z\in[im,iQ]$. Thus, the same DOS $u$ simultaneously satisfies the NDR \eqref{dr_soliton_gas1} for soliton gas on $\G^+=[0,iQ]$ and 
the NDR \eqref{dr_breather_gas1} for breather gas on $\G^+=[im,iQ]$.
In both cases $\n\equiv 0$ and, thus, $\s\equiv 0$.
Theorems \ref{the-per-sol-gas}, \ref{the-per-breath-gas} imply that the above mentioned property of DOS $u(z)$ is valid for any  $q(x)$ with $m>0$.

We then calculate (Theorem \ref{the-per-I-gx}) the average densities $I_k$ for the  periodic soliton and breather gases, which turned out to be zero for any even $k$ and proportional to the $(k+1)$th moment of $q(x)$ for any odd $k$:

\begin{align}
    I_k=\frac{(-1)^{\frac{k+1}{2}}kd_k}{L}\int_0^L q^{k+1}(x)dx,
    \qquad k=1,3,5,\dots.
\end{align}
Moreover,
\begin{align}
    2g_x(z)=\frac{1}{L}\int_0^L\le({z-\sqrt{z^2+q^2(x)}}\ri)dx+\frac{1}{L}\int_0^L\sqrt{q^2(x)-m^2}dx,\quad z\in \bar\C\backslash \G.
\end{align}

According to Corollary \ref{cor-gxpm-int}, that implies
\be\label{comp-u}
{\s(z)u(z)=\frac 1{ L}\int_{q^{-1}(|z|)}^L\sqrt{|z^2+q^2(x)|}dx, \quad
\tilde u(z)=\frac{1}{L}\int_0^L\sqrt{q^2(x)-m^2}dx}
\ee
on $\G^+$ for both soliton and breather periodic gases.

 In addition to the bound state condensate \eqref{bstate-con}, a few more particular examples of  periodic soliton gases were considered in Section \ref{sec-per}.

 Finally, in Appendix \ref{sec-err}, Lemma \ref{lem-est-evenB}, we find the  thermodynamic limit of the  asymptotic behavior of the coefficients of the systems of linear equations \eqref{WFR}  for $k_j,\o_j$ together with error estimates. This result is later used in Theorem \ref{th-v>0},
 where we show that DOS $u(z)$ in the thermodynamic limit can be approximated by a piece-wise constant function $\hat u(z)$  that is defined through the wavenumbers $k_j$.

\section{Background}\label{sec-background}

The  simplest solution of  equation \eqref{NLS} 
is a plane wave 
 \begin{equation}\label{pw}
\psi= q e^{2iq^2  t}, 
%\equiv \psi_0.
\end{equation}
where $q>0$ is the amplitude of the wave.

It is well known that the fNLS is an integrable equation \cite{ZS}; the Cauchy (initial value) problem for \eqref{NLS}  can
be solved using the inverse scattering transform (IST) method for different classes of initial data (potentials).
The scattering transform connects a given potential with its scattering data expressed in terms of the spectral variable $z\in\C$.
In particular, the scattering data consisting of one pair of spectral points $z=a\pm ib$, where $b>0$, and a (norming) constant $c\in \C$,
defines the famous  soliton solution 
\begin{equation}\label{nls_soliton}
\psi_S (x,t)= 2ib \, \hbox{sech}[2b(x+4at-x_0)]e^{-2i(ax + 2(a^2-b^2)t)+i\phi_0},
   \end{equation}
to the fNLS. This solution represents a solitary traveling wave  (pulse  on a zero background) with $c$ defining the initial position $x_0$ of its center and the initial phase $\phi_0$.
%The soliton \eqref{nls_soliton} 
%represents a spatially localized traveling wave  (pulse)  on a zero background. 
It is characterized by two independent papameters:  $b=\Im z$ determines the soliton amplitude $2b$ and
$a=\Re z$ determines its velocity $s=-4 a$. Scattering data that consists of several points
$z_j \in \C^+ $ (and their complex conjugates $\bar z_j$), $j\in \N$, together with their norming constants corresponds to  multi-soliton
solutions. Assuming that at $t=0$ the centers of individual solitons are far from each other, 
we can represent the fNLS time evolution of a multi-soliton solution as propagation and interaction of the 
individual solitons.

It is well known 
that the interaction of solitons in multi-soliton fNLS solutions reduces to only two-soliton elastic
collisions, where
the faster soliton (corresponding to $z_m$) gets a forward shift \cite{ZS}
$$
\Delta_{mj}=\frac1{\Im(z_m)}\log\left|\frac{z_m-\overline{z}_j}{z_m-z_j}
\right|
%,\quad \Re(z_m)>\Re(z_j),
$$
and the slower ``$z_j$-soliton'' is shifted backwards by $-\Delta_{mj}$.

Suppose now we have a ``large ensemble''  (a ``gas'') of solitons \eqref{nls_soliton} whose  spectral characteristics $z$ are 
distributed over a compact set $\G^+\subset \C^+$ according to some non negative measure $\m$. Assume also that 
the locations (centers) of these solitons are distributed uniformly on $\R$ and that $\mu(\G^+)$ is small, i.e, the gas 
is dilute. Let us consider the speed of the trial $z$ - soliton in the gas. Since it undergoes rare but
sustained collisions with other solitons, the speed $s_0(z)=-4\Re z$ of a free solution must be modified as
\begin{equation} \label{mod-speed}
s(z)=s_0(z) + \frac{1}{\Im z} \int_{\G^+}\log \le|\frac{w-\bar z}{w-z}\ri|[s_0(z)-s_0(w)]d\mu(w).
\end{equation}
Similar modified speed formula was first obtain by V. Zakharov \cite{Zakharov} in the context of the Korteweg-de Vries
(KdV) equation. Without the ``dilute'' assumption, i.e, with $\mu(\G^+)=O(1)$, equation \eqref{mod-speed} for $ s(z)$ turns 
into the 

integral equation
\begin{equation} \label{eq-state}
s(z)=s_0(z) + \frac{1}{\Im z} \int_{\Gamma^+}\log \left|\frac{w-\bar z}{w-z}\right|[s(z)-s(w)]d\mu(w)
\end{equation}
 known as  {the} {\it equation of state} for the fNLS soliton gas, an analog of which was first obtained in \cite{ElKamch} using 
purely physical reasoning. Here  $s(z)$ has the meaning of the speed of the ``element of the gas'' associated with the spectral
parameter $z$ (note that when $\mu(\G^+)=O(1)$ we cannot distinguish individual solitons).

{A} similar equation in the KdV context was obtained earlier in \cite{El2003}. 
If we now assume some  dependence of  $s$ and $u$ on space time parameters $x,t$  
(here $d\mu=u d\lambda$  with $\lambda$ being the Lebesgue measure) that  
occurs on very large spatiotemporal scales, 
%than the typical scales of the oscillations corresponding to individual solitons,  
then we complement the equation of state 
\eqref{eq-state} by the continuity equation for the density of states
\begin{equation} \label{kinet}
\part_t u+\part_x (su)=0, 
\end{equation}
which was first suggested in  \cite{ElKamch} and derived in \cite{ElTovbis}. Equations \eqref{eq-state},
\eqref{kinet} form the kinetic equation for  {a} dynamic (non-equilibrium) fNLS soliton gas. 
The kinetic equation for the KdV soliton gas was derived in  \cite{El2003}.
It is remarkable that 
recently the kinetic equation having
similar structure 
was derived in the framework of the “generalized hydrodynamics”
for quantum many-body integrable systems, see, for example, \cite{DYC, DSY, VY}. 

It is easy to observe that \eqref{eq-state} is a direct consequence of
\eqref{dr_breather_gas1}-\eqref{dr_breather_gas2}, where $s(z)= \frac{v(z)}{u(z)}$.
%as speed of the ``component'' of the soliton gas parametrized by $z\in \G^+$. 
Indeed, after multiplying  \eqref{dr_soliton_gas1} by $s(z)$,
substituting $v(z)=s(z)u(z)$ into  \eqref{dr_soliton_gas2}, subtracting the second equation from the first one
and dividing both parts by $\Im z$ we obtain exactly \eqref{eq-state}. In this paper we mostly consider the NDR \eqref{dr_breather_gas1}-\eqref{dr_breather_gas2} and \eqref{dr_soliton_gas1}-\eqref{dr_soliton_gas2} for equilibrium soliton gases, that is,
we do not assume any dependence of $u,v$ on the space-time variables $x,t$.

A mathematical albeit formal (i.e., without, for example, error estimates)  derivation of the equation of state \eqref{eq-state} 
was presented in the recent
paper \cite{ElTovbis}. The first step in this process is derivation of equations \eqref{dr_soliton_gas1}-\eqref{dr_soliton_gas2}, which describe
the density of states $u$ and its temporal analog $v$.
The derivation  is based on the idea of  thermodynamic limit for a family of
finite gap solutions of the fNLS, which was originally developed for the KdV equation in \cite{El2003}. 
Finite-gap solutions are quasi-periodic functions in $x,t$ that  can be spectrally represented by a finite
number of Schwarz symmetrical arcs (bands) on the complex $z$ plane. Here Schwarz symmetry means that either
a band $\g$ coincides with its Schwarz symmetrical image $\bar \g$ or if $\g$ is a band then  $\bar\g$
is another band. Assume additionally that there is a complex 
constant (initial phase) associated with each band that also respects the Schwarz symmetry, i.e.,  Schwarz symmetrical bands 
have Schwarz symmetrical phases. Given a finite set of Schwarz symmetrical bands with the corresponding phases,
a finite-gap solution to the fNLS can be written explicitly in terms of the Riemann theta functions
on the hyperelliptic Riemann surface $\mathfrak R$, where  the bands are the branchcuts of $\mathfrak R$, see, for example, \cite{BBIM} and references therein.

For convenience of the further brief exposition, we will consider  the hyperelliptic Riemann surface $\mathfrak R=\mathfrak R_{N}$ to be of genus $2N$, which  
equals the number of bands minus one. The one exceptional band $\g_0$ will be crossing $\R$, whereas the remaining $N$
bands $\g_j\subset \C^+$, $j=1,\dots,N$, and their Schwarz symmetrical $\g_{-j}:=\bar\g_j\subset \C^-$. 
Definition of a particular finite-gap solution of the fNLS   starts with  
a smooth Schwarz symmetrical 
function $f_0=f_0(z;N)$ that is  defined on $\g_j$, $j=0,\pm 1\cdot, \pm N$. Given $f_0$, the
corresponding finite gap solution of the  fNLS can be defined through the solution of the following matrix RHP (see, for examlpe,
\cite{DIZ}, \cite{TVZ1}).

\begin{problem}\label{RHPY} Find a matrix-valued function $Y(z)$, such that $Y(z)$ is: (i) analytic together with its inverse  $Y^{-1}(z)$ on 
 {$\bar\C\setminus \cup_{|j|=0}^{N} \g_j$}; (ii) satisfies the jump condition
\be\label{RHPY-jump}
Y_+(z)=Y_-(z) i\s_2e^{2 if(z)\s_3}~~~{\rm on}~~~\g_j,~~~~j=0,\pm 1\dots,\pm N,
\ee
{where $\s_j$, $j=1,2,3$,  denote standard Pauli matrices, $ f(z)=f_0(z)+xz+2tz^2$} and the orientation of the bands $\g_j$ is shown on Figure \ref{Fig:Cont} below;
(iii) $\lim_{z\ra\infty}Y(z)=\1$, and; (iv) the boundary values $Y_\pm(z)$ on the positive and negative sides respectively are locally $L^2$ on all the bands $\g_j$.  
\end{problem}

It is well known that the RHP \ref{RHPY} has a unique solution and that the solution to the fNLS \eqref{NLS} is given by
\be\label{NLS-RHPY}
\psi(x,t)=-2 (Y_1)_{1,2},~~\text{where} ~~~Y(z)=\1+ Y_1z^{-1}+\cdots 
\ee
is the expansion of $Y(z)$ at infinity (\cite{Deift}). Note that $Y_1$ depends on $x,t$.

The next step in finding $Y(z)$ is to reduce the jump matrix on each $\g_j$ to a constant (in $z$, but not in $x,t$) matrix.
Assume that  there exists a simple piecewise smooth  symmetrical contour $\G$ such that 
$\G= \cup_{|j|=0}^N \g_j  \bigcup \cup_{|j|=1}^N c_j$, where the arcs $c_j$ will be called ``gaps'', connecting  the consecutive bands,
see Figure \ref{Fig:Cont}.
Then the reduction to the ``constant jumps" RHP 
can be done with the help of the so-called $g$-function, that is, by the transformation 
\be\label{Y2Z}
Y(z) = e^{-2ig(\infty)\s_3}Z(z)e^{2ig(z)\s_3}, 
\ee
where $g=g(z;x,t,N)$ is an unknown function, analytic at  $\bar\C\setminus \G$.
%$\bar\C\setminus \cup_{|j|=0}^{N} \g_j$. 
Then $Z(z)$ satisfies jump conditions 
\be\label{Z-jump}
Z_+(z)=Z_-(z)e^{2ig_-(z)\s_3} i\s_2e^{2 if(z)\s_3}e^{-2ig_+(z)\s_3} =Z_-(z)i\s_2e^{2 i(f(z)-g_+(z)-g_-(z))\s_3}
\ee
% \end{multiline}
on each $\g_j$. Similarly, we have 
$ Z_+(z)=Z_-(z)e^{-2i(g_+(z)-g_-(z))\s_3}$ on each $c_j$.

The reduction to a piece-wise constant jump matrix for $Z(z)$ will be achieved if we define $g(z)$ as a solution of 
the following scalar RHP.  
\begin{problem}\label{RHPg}
Find  a function $g$ satisfying the following requirements: 1) $g$ is   analytic in $\bar \C\setminus \G$;
2) it satisfies the jump conditions
\begin{align}
\label{g-jump}
g_+(z)+g_-(z)&=f(z) +W_j~~~{\rm on}~~~\g_j, ~~|j|=0,1,\dots, N,  \nonumber\\
g_+(z)-g_-(z)&=\O_j~~~{\rm on}~~~c_j, ~~|j|=1,\dots,N,
\end{align} 
where $f(z)$ is given in RHP \ref{RHPY}, $W_0=0$ and the real constants $W_j, \O_j$, $|j|=1,\dots, N$, subject to the symmetries
$W_{-j}=W_j$, $\O_{-j}=\O_j$,  are to be defined,  and; 3) the boundary values $g_\pm(z)$ are locally $L^2$ functions. 
\end{problem}

According to Sokhotsky-Plemelj formula, the solution of the RHP \eqref{g-jump}, if exists, is given by
\begin{equation}\label{gform2}
g(z)={\frac{R(z)}{2\pi i}}\le[
\sum_{|j|=0}^{N} \int_{\g_j}{\frac{[f(\z)+W_j]d\z}{(\z-z)R_+(\z)}}+\sum_{|j|=1}^{N} \int_{c_j}{\frac{\O_jd\z}{(\z-z)R_+(\z)}}\ri].
\end{equation}
Here
\be\label{R_N}
R(z)
=\sqrt{\prod_{j=1}^{2N+1}(z-\a_{j})(z-\bar\a_{j})} ,
\ee
where $\a_{2j},\a_{2j+1}$ are the respective endpoints of the oriented main arcs  $\g_j$, $j=1,\dots,N$, $\a_1$ and $\bar\a_1$ are the 
endpoints of $\g_0$
and $R_+$ denotes the value
of $R$ on the the positive (left) side of the each $\g_j$. 
To show that \eqref{gform2} satisfies the RHP \ref{RHPg}, one has to show
that the right hand side of \eqref{gform2} is analytic at infinity, that is, the system
\be\label{WO}
\begin{split}
\sum_{|j|=1}^N W_j\oint_{A_j}{\frac{\z^m d\z}{R(\z)}} +\sum_{|j|=1}^N \O_j\oint_{C_j}{\frac{\z^md\z}{R(\z)}} 
=-\sum_{|j|=0}^N \oint_{A_j}{\frac{f(\z)\z^m d\z}{R(\z)}}
\end{split}
\ee
for unknown real constants $W_j, \O_j$,
where $m=0,1\cdots, 2N-1$, 
obtained from \eqref{gform2} by expanding $\frac{1}{\z-z}$ in powers of $\frac 1z$ as $z\ra\infty$, has a solution.
Here $A_j$,  $C_j$ are  negatively oriented loop  contours on $\mathfrak R_{N}$ containing the arcs $\g_j,c_j$ respectively and no other arcs; 
to allow loop contour integration, we have to assume that  
$f_0(z;N)$ has an analytic extension around each band $\g_j$. We note that \eqref{gform2} can be rewritten as 
\begin{equation}\label{gform3}
2g(z)={\frac{R(z)}{2\pi i}}\le[
\sum_{|j|=0}^{N} \oint_{A_j}{\frac{[f(\z)+W_j]d\z}{(\z-z)R_+(\z)}}+\sum_{|j|=1}^{N} \oint_{C_j}{\frac{\O_jd\z}{(\z-z)R_+(\z)}}\ri],
\end{equation}
where $z$ is outside of each loop.

Since the bands $\g_j$ and the gaps $c_j$ do not depend on $x,t$, differentiating the RHP \ref{RHPg} in $x,t$ would transform it
to similar RHPs for $g_x,g_t$ with constants  $W_{xj},\O_{xj}$ or $W_{tj},\O_{tj}$, where $f(z)$ should be replaced by $f_x(z)=z$ or $f_t(z)=2z^2$ respectively. Thus, we obtain the corresponding expressions \eqref{gform3} for $2g_x$ and $2g_t$. 

Let us define  the quasimomentum $dp_N$ and quasienergy $dq_N$ differentials on $\mathfrak R_{N}$ as real normalized (all the periods of $dp_N,dq_N$ are real)
meromorphic differentials of the second kind  with the only poles at $z=\infty$ on both sheets. These differentials 
are  (uniquely) 
defined 
(see e.g. \cite{ForLee},
\cite{Krichever}, \cite{bertolatovbis2015})
by local expansions 
\begin{equation}\label{pq}
dp_N \sim [\pm 1 + \mathcal{O}(z^{-2})]dz , \qquad dq_N \sim [\pm 4z+\mathcal{O}(z^{-2})]dz 
\end{equation}
near $z=\infty$ on the main and second sheet respectively.
One can observe  that $dp_N, dq_N$ - the quasimomentum  and the quasienergy differentials - can be expressed as
\be\label{g-diff}
{dp_N=  [1 - 2g_{xz}(z)]dz, \qquad dq_N=  [4z - 2g_{tz}(z)]dz.}
\ee
Indeed, the right hand sides of \eqref{g-diff} have the required behavior at infinity and, as one can easily see  (\cite{TVZ1}), all their the periods
are real. 

Introduce 
\bea\label{k-om}
k_j= 2(\O_{x,j}-\O_{x,j-1}),~~~\tilde k_j= 2W_{x,j},~~~{|j|}=1,\cdots, N, \cr
\o_j= 2(\O_{t,j}-\O_{t,j-1}),~~~\tilde \o_j= 2W_{t,j},~~~{|j|}=1,\cdots, N.
\eea
Then it can be easily shown  (see also  \cite{ElTovbis}) that 
the wave numbers $k_j$, $\tilde k_j$ and the frequencies $\o_j$, $\tilde \o_j$ of a quasi-periodic
finite gap solution $\psi_{N}$ determined by $\mathfrak R_{N}$, can be expressed as
\begin{align}
k_j &= -\oint_{\mathrm{A}_j}dp_N, \quad \o_j = -\oint_{\mathrm{A}_j}dq_N, \quad j=1, \dots, N, \label{kdp}\\
\tilde k_j &= \oint_{\mathrm{B}_j}dp_N, \quad \tilde \o_j = \oint_{\mathrm{B}_j}dq_N, \quad j = 1, \dots, N, \label{omdq}
\end{align}
where the cycles $A_j,B_j$ are shown on Figure \ref{Fig:Cont}.

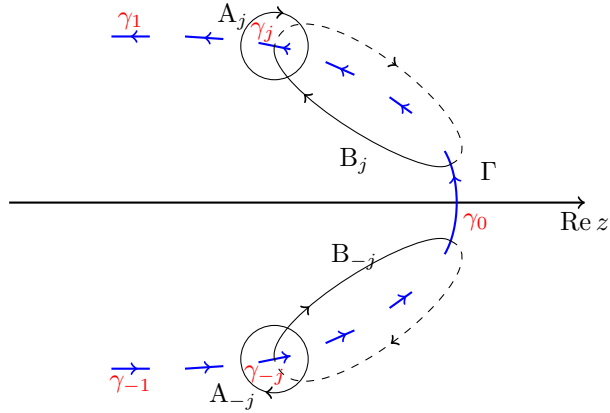
\begin{figure}[h]
\begin{center}
\begin{tikzpicture}[scale=.85]

%real axis
\draw[black,thick,postaction = {decorate, decoration = {markings, mark = at position 1 with {\arrow[black,thick]{>};}}}] (-3,0) -- (6,0);
\node[below] at (6,0) {$\Re z$};

%band intersecting real axis
\draw[blue,thick,postaction = decorate, decoration = {markings, mark = at position .5 with {\arrow[blue,thick]{>};}}] (4,0) arc [start angle=0, end angle=50, x radius=15pt, y radius=30pt];
\draw[blue,thick] (4,0) arc [start angle=0, end angle=-50, x radius=15pt, y radius=30pt];

%bands above the real axis
\draw[blue,thick,,postaction = decorate, decoration = {markings, mark = at position .5 with {\arrow[blue,thick]{>};}}] (3.3,1.4) -- (2.95,1.65);
\draw[blue,thick,,postaction = decorate, decoration = {markings, mark = at position .5 with {\arrow[blue,thick]{>};}}] (2.4,2) -- (1.95,2.2);
\draw[blue,thick,postaction = decorate, decoration = {markings, mark = at position .25 with {\arrow[blue,thick]{>};}}] (1.4,2.4) -- (0.9,2.5);
\draw[blue,thick,postaction = decorate, decoration = {markings, mark = at position .5 with {\arrow[blue,thick]{>};}}] (0.35,2.56) -- (-0.25,2.6);
\draw[blue,thick,postaction = decorate, decoration = {markings, mark = at position .5 with {\arrow[blue,thick]{>};}}] (-0.8,2.6) -- (-1.4,2.6);

%bands below the real axis
\draw[blue,thick,postaction = decorate, decoration = {markings, mark = at position .5 with {\arrow[blue,thick]{<};}}] (3.3,-1.4) -- (2.95,-1.65);
\draw[blue,thick,postaction = decorate, decoration = {markings, mark = at position .5 with {\arrow[blue,thick]{<};}}] (2.4,-2) -- (1.95,-2.2);
\draw[blue,thick,postaction = decorate, decoration = {markings, mark = at position .25 with {\arrow[blue,thick]{<};}}] (1.4,-2.4) -- (0.9,-2.5);
\draw[blue,thick,postaction = decorate, decoration = {markings, mark = at position .5 with {\arrow[blue,thick]{<};}}] (0.35,-2.56) -- (-0.25,-2.6);
\draw[blue,thick,postaction = decorate, decoration = {markings, mark = at position .5 with {\arrow[blue,thick]{<};}}] (-0.8,-2.6) -- (-1.4,-2.6);

% \gamma_j contour
\draw[black,postaction = decorate, decoration = {markings, mark = at position .24 with {\arrow[black,thick]{<};}}] (1.15,2.45) circle [radius=15pt];
% \gamma_j label
\node[below] at (0.5,3.25) {$\mathrm{A}_j$};
% W_j label
\node at (0.97,2.67) {$\textcolor{red}{\g_j}$};

% \gamma_0 contour
%\draw[black,postaction = decorate, decoration = {markings, mark = at position .4 with {\arrow[black,thick]{<};}}] (4.45,0) arc [start angle=0, end angle=360, x radius=15pt, y radius=30pt];
%% \gamma_0 label
\node at (4.3,-0.3) {$\textcolor{red}{\g_0}$};

% B_j cycle non-dashed
\draw[black,postaction = decorate, decoration = {markings, mark = at position .74 with {\arrow[black,thick]{>};}}] (3.9,0.6) .. controls (3.5,0.3) and (1.0,1.8) .. (1.15,2.45);
% B_j cycle dashed
\draw[black,dashed,postaction = decorate, decoration = {markings, mark = at position .5 with {\arrow[black,thick]{<};}}] (3.9,0.6) .. controls (4.9,1) and (1.5,3.75) .. (1.15,2.45);
% B_j cycle label
\node[below] at (2.4,1) {$\mathrm{B}_j$};
% U_j label
\node at (3,1) {$\textcolor{red}{}$};

% B_{-j} cycle non-dashed
\draw[black,postaction = decorate, decoration = {markings, mark = at position .74 with {\arrow[black,thick]{<};}}] (3.9,-0.6) .. controls (3.5,-0.3) and (1.0,-1.8) .. (1.15,-2.45);
% B_{-j} cycle dashed
\draw[black,dashed,postaction = decorate, decoration = {markings, mark = at position .5 with {\arrow[black,thick]{>};}}] (3.9,-0.6) .. controls (4.9,-1) and (1.5,-3.75) .. (1.15,-2.45);
% B_{-j} cycle label
\node[below] at (2.4,-0.5) {$\mathrm{B}_{-j}$};
% U_{-j} label
\node at (3,-1.1) {$\textcolor{red}{}$};

% \gamma_{-j} contour
\draw[black,postaction = decorate, decoration = {markings, mark = at position .74 with {\arrow[black,thick]{<};}}] (1.15,-2.45) circle [radius=15pt];
% \gamma_{-j} label
\node[below] at (0.5,-2.7) {$\mathrm{A}_{-j}$};
% W_{-j} label
\node at (1,-2.7) {$\textcolor{red}{\g_{-j}}$};

% W_N label
\node at (-1.1,2.85) {$\textcolor{red}{\g_1}$};

\node at (4.5, 0.5) {$\textcolor{black}{ \G}$};

% W_{-N} label
\node at (-1.1,-2.85) {$\textcolor{red}{\g_{-1}}$};

\end{tikzpicture}
\end{center}
\vskip -.6cm
\vspace{-10pt}
\caption{The spectral bands $\g_{\pm j}$ and the cycles $\mathrm{A}_{\pm j}, \mathrm{B}_{\pm j}.$
 The 1D Schwarz symmetrical curve $\G$ consists of the bands $\g_{\pm j}$, $j=0,\dots,N$, and the gaps $c_{\pm j}$ between the bands (the gaps are  not shown on this figure).}
\label{Fig:Cont}
\vspace{-5pt}
\end{figure}

%We shall call the special set of wavenumbers and frequencies defined by  \eqref{kdp}, \eqref{omdq} {\it the  fundamental wavenumber-frequency set}. 
 Note that the wavenumbers and frequencies defined by  \eqref{kdp} and those defined by \eqref{omdq} are of essentially different nature:   in the limit of $\g_j$  shrinking to a point,  we have
 \begin{equation}
  k_j, \o_j \to 0,\quad\text{whereas}  \quad \tilde k_j, \tilde \o_j = \mathcal{O}(1),  \qquad j=1, \dots, N,\label{breather_lim} 
  \end{equation}
see  \cite{ElTovbis}.  Motivated by these properties, $k_j,\o_j$  are called
{\it solitonic  wavenumbers and  frequencies } whereas 
the remaining 
 $\tilde k_j$, $\tilde \o_j$  are called {\it carrier  wavenumbers and  frequencies}. 
 
The standard normalized holomorphic differentials $w_j$ of  $\mathfrak{R}_{N}$ are defined by
\begin{equation} \label{D-P}
w_j=[P_j(z)/R(z)] dz, \quad \oint_{\mathrm{A}_i}w_j = \delta_{ij}, \quad i, j= \pm 1, \dots, \pm N,
\end{equation}
where the polynomials 
\begin{equation} \label{Pj}
P_j(z)=\k_{j,1}z^{2N-1}+\k_{j,2}z^{2 {N}-2}+ \dots +\k_{j,2N}
\end{equation}
 have complex coefficients and the radical $R$ given by \eqref{R_N}
%\begin{equation}
%\label{rsurf}
%R(z)=\prod_{k=1}^{2N+1}(z-\a_k)^\hf(z-\bar\a_k)^\hf.\end{equation}
defines the hyperelliptic surface  $\mathfrak{R}_{N}$.
%, i.e, the product runs over all the endpoints (of the bands) $\a_k\in\C^+$.
Then,  
according to \eqref{gform3}-\eqref{Pj} (see also 
\cite{ElTovbis}),  
the   wavenumbers  and frequencies satisfy the  systems  
\begin{align}
&& \tilde k_j+\sum_{|m|=1}^{N} k_m\oint_{\mathrm{B}_m}{\frac{P_j(\z)d\z}{R(\z)}}=-{2} \oint_{\gt}{\frac{\z P_j(\z) d\z}{R(\z)}}, \nonumber\\
&&\tilde \o_j+\sum_{|m|=1}^{N} \o_m\oint_{\mathrm{B}_m}{\frac{P_j(\z)d\z}{R(\z)}}=- {4}\oint_{\gt}{\frac{\z^2 P_j(\z) d\z}{R(\z)}}, \nonumber \\
&& |j|= 1,\dots, N,  \label{WUPjM}
\end{align}
where $\gt$ is a large clockwise oriented contour containing $\G$. In fact, every equation in \eqref{WUPjM} is the Riemann bilinear identity for the differentials $w_j$ and $dp_N,dq_N$ respectively.

Taking imaginary parts of \eqref{WUPjM} and using the residues in the right hand side,  we obtain the  systems 
\begin{align} 
\sum_{|m|=1}^{N} k_m\Im\oint_{\mathrm{B}_m}{\frac{P_j(\z)d\z}{R(\z)}} &={4}\pi \Re\k_{j,1},\nonumber \\
\sum_{|m|=1}^{N} \o_m\Im\oint_{\mathrm{B}_m}{\frac{P_j(\z)d\z} 
{R(\z)}} & ={8}\pi \Re\left(\k_{j,1}\sum_{k=1}^{2N+1}\Re\a_k +\k_{j,2}\right), \nonumber \\
& \qquad  |j|= 1,\dots, N, \label{WFR}
 \end{align}
where the latter summation is taken over all the endpoints in $\C^+$,
which define the  solitonic  wavenumbers  and frequencies.
We call  \eqref{WFR} the solitonic nonlinear dispersion relations (NDR).
Indeed, the NDR
 indirectly connect (through the Riemann surface $\mathfrak{R}_{N}$) 
the solitonic wavenumbers and frequencies of the finite gap solution  $\psi_{N}$,
i.e., \eqref{WFR} represents  nonlinear dispersion relations. Once the solitonic wavenumbers and frequencies were obtained, the corresponding carrier quantities can be reconstructed by taking the real part of  \eqref{WUPjM}.

Equations \eqref{WFR} together with \eqref{breather_lim}  {are} our starting point  for deriving 
equations \eqref{dr_breather_gas1}-\eqref{dr_breather_gas2}. 
%Before describing the derivation,
We want to point out that the matrix of the systems  \eqref{WFR} is negative-definite and, 
therefore, each of the systems
\eqref{dr_breather_gas1}-\eqref{dr_breather_gas2} has a unique solution.
The negative-definiteness of the matrix of the systems  \eqref{dr_breather_gas1}-\eqref{dr_breather_gas2} 
follows from the properties of the the Riemann period matrix $\t$ of the Riemann surface $\mathfrak{R}_{N}$
( $\Im \t$ is positive definite).

Suppose now that we start shrinking each band to a point. Then we will be taking the finite gap solution
to its multi-soliton solution  limit, where the phases should be transformed into the corresponding norming constants.
The idea of the {\it thermodynamical limit} consists of increasing the number $2N+1$ of bands  simultaneously with shrinking the size $2|\d_j|$
of each  band $\g_j$ (with the exception of the permanent band  $\g_0$ in the breather gas) 
at  some exponential rate with respect to $N$, so that the centers $z_j$ of the bands $\g_j\subset \C^+$, $j=1,\dots,N$, 
will be filling a certain compact set $\G^+\subset \C^+$ with some limiting probability density $\varphi(z)$. 
In particular, we assume
\be\label{exp-scale}
|\d_j|=e^{-N(\nu(z_j)+o(1))},
\ee
uniformly on $\G^+$, 
where $\nu(z)$, called the scaled bandwidth function, is some nonnegative continuous function on $\G^+$. 
In the case when $\nu(z)=0$ on some subset of $\G^+$, we
still assume that the corresponding $\d_j\ra 0$ much faster than $N^{-1}$, perhaps, see Appendix \ref{sec-err} for details.
Moreover, we assume  the distance between any two bands to be much larger than the size of the shrinking bands, that is, it must be of order at least $O(N^{-1})$ uniformly on $\G^+$.

Under these assumptions we derive 
 the leading order behavior of the coefficients of the linear system \eqref{WFR}, see 
 Lemma \ref{lem-est-evenB}
 and \eqref{lead-int-ass-evB}. The  expression for the off-diagonal entries from \eqref{lead-int-ass-evB} provides the 
 kernel of the integral operator in \eqref{dr_breather_gas1}-\eqref{dr_breather_gas2}, whereas the  
 expression for the diagonal entries defines the secular (non-integral) term in the left hand sides of
 \eqref{dr_breather_gas1}-\eqref{dr_breather_gas2}. Here the spectral scaling function
 \begin{equation}\label{dens_states}
 \s(z)=2\frac{\nu(z)}{\varphi(z)}\quad\text{  and }\quad
u(z)= \frac{1}{\pi} \check u(z) \varphi(z),\ \  v(z)= \frac{1}{\pi} \check v(z) \varphi(z),
\end{equation}
where $ \check u(z), \check v(z)$ are some smooth functions on $\G^+$ interpolating the values $\frac{Nk_j}{2}, -\frac{N\o_j}{2}$ at $z=z_{j,N}$,  $j=1,\dots,N$, respectively,
see also \cite{ElTovbis}.

%that is,$\tilde u(z_j)=\tilde k_j$, $\tilde v(z_j)=\tilde \o_j$,

The breather gas is obtained when in the thermodynamic limit all the bands except $\gamma_0$ are shrinking ($\d_j\ra 0$), whereas
the exceptional band $\gamma_0$ approaches some permanent limiting position as $N \to \infty$. In particular, in this paper we assume that the endpoints of $\gamma_0$
approach $\pm \delta_0\in i\R$ respectively. Being considered alone, the permanent spectral band $\gamma_0$ corresponds
to the plane wave solution \eqref{pw} with $q=|\delta_0|$. The band $\gamma_0$ together with  Schwarz symmetrical points
of discrete spectrum $z, \bar z$
correspond to a soliton on the plane wave (carrier) background, also known as 
a breather. It is  remarkable that the kernel in the integral equations \eqref{dr_breather_gas1}-\eqref{dr_breather_gas2}, 
being divided by $\Im R_0(z)$, provides an elegant expression for the phase shift  of two interacting 
breathers; some equivalent
(see \cite{Roberti})
but considerably more  complicated expressions for this phase shift were recently obtained in \cite{LiBiondini},
\cite{Gelash}.
Therefore, equations 
 \eqref{dr_breather_gas1}-\eqref{dr_breather_gas2} 
represent nonlinear dispersive relations for the breather gas.
It is easy to check that equations   \eqref{dr_breather_gas1}-\eqref{dr_breather_gas2}
coincide with  \eqref{dr_soliton_gas1}-\eqref{dr_soliton_gas2}  in  the limit $\delta_0\to 0$.
Thus, soliton gas can be considered as a particular case of the breather gas, see \cite{ElTovbis} for details. 

In the case of subexponential  rate of shrinking of bands $\g_j$ in the thermodynamic limit, the function 
$\s(z)$ turns to be zero and we obtain a breather (or soliton, if $\d_0\ra 0$) {\it condensate} (\cite{ElTovbis}).
As it was mentioned in Remark 1.4 from\cite{KT2021}, the term ``condensate'' reflects the fact that for a 
given $\G^+$  a certain quadratic energy $J_\s(\m^*_\s)$ is minimized in $\s$  when $\s\equiv 0$ on $\G^+$. Here $d\m^*_\s(z)=u(z)d\l(z)$ denotes the minimizing measure for a given $\s$.

Proving the transition  from linear systems \eqref{WFR} to the corresponding integral equations \eqref{dr_breather_gas1}-\eqref{dr_breather_gas2} requires the following additional 
assumption on the probability measure $\varphi(z)$.
Let us divide the compact $\G^+=\cup_{j=1}^N\G_j$ into disjoint ``regions of attraction"  $\G_j$, $z_{j,N}\in\G_j$, in such a way that all points in $\G^+$ that are closer to $z_{j,N}$ than to any other $z_{k,N}$, $k\neq j$, belong to $\G_j$. Then 
\be\label{varphi-def}
\l(\G_j)=\frac{1}{N\varphi(z_{j,N})}+O(\frac{1}{N^{1+\b}}) \quad\text{ as $N\ra\infty$}
\ee
with some $\beta>0$ uniformly in $\G^+$.

\section{The thermodynamic limit of averaged densities and fluxes }
\label{sec-cons-dens}

 It is well known that the  averaged conserved quantities (densities and fluxes) of the finite gap solutions are closely related to certain Abelian differentials of the second kind \cite{NovikovSoliton} of the correspondent hyperelliptic Riemann surface. In particular, for the case of fNLS, expressions for the averaged densities and fluxes can be found in \cite{ForLee},  where they were
 %In this section, we 
 %follow the idea of \cite{TE2016} to 
 %derive formulae of the thermodynamic limit of those averaged densities and fluxes.

expressed  as
%\sout{ invariants  (conservation laws) are just }
the coefficients of the expansion at infinity of  
%\sout{ the relation:$\partial_t 
$dp_N, dq_N.$
Let's first consider the expansion of $dp_N$ at $z=\infty$:
\begin{align*}
    dp_N={\le(1-\sum_{m=1}^\infty\frac{I_{m,N}}{z^{m+1}}\ri)dz.}
\end{align*}
%\sout{Then the $I_{m,N}$ are the averaged invarants (for $dp_N$, the invariants are often called densities in physics).} 
Theorem \ref{theo-ImN}, part (i), gives the 
leading behavior of $  I_{m,N}$ in the thermodynamic limit.
The analytic function $\r_N$, defined in the Appendix \ref{sec-err}, equation  \eqref{R-est}, and  studied in Lemma \ref{lem-Rd}, plays an
important part in the proof.

\subsection{Proof of Theorem \ref{theo-ImN} part (i)}
\begin{proof}

The definition of $\r_N$ (see also in \eqref{R-est}) is
\be\label{rho-d}
R(z)=R_0(z)P(z)(1+\r_N(z)),\quad\text{ where}\quad P(z):= \prod_{|j|=1}^N(z-z_j)
\ee
and $R(z)$ is given by \eqref{R_N}.
%(see the definition \eqref{R-est}), 
Since $\r_N$, $R_0$ and $g_x$ are all analytic at $\infty$, we can write
\begin{align}
    2g_x(z)=\sum_{m=0}^\infty\frac{G_{m,N}}{z^m}\label{g-x-expansion}
\end{align}
and
\begin{align}
    \frac{2g_x(z)}{(1+\r_N(z))R_0(z)}=\sum_{m=0}^\infty\frac{\tilde G_{m,N}}{z^{m+1}}.\label{new-expansion}
\end{align}
Based on the representation \eqref{gform3} of the $g$-function, differentiated with respect to $x$, we have
\begin{align}
    \tilde G_{m,N}=
     \sum_{k=0}^{m}R_kY_{m-k}, \quad m=\{0\}\cup\N, \label{rel-Y-I}
\end{align}
where
\begin{align}
    Y_k:=-\frac{1}{2\pi i}\left(\sum_{|j|=0}^N\oint_{A_j}(\z+{W_{j,N}})+\sum_{|j|=1}^N\oint_{B_j}{U_{j,N}}\right)\frac{\z^{2N+k}d\z}{R(\z)},\quad k=\{0\}\cup\N,\label{Im-def}
\end{align}
{\begin{align}
	W_{j,N}&=\frac{1}{2}\tilde{k}_j:=\frac{1}{2}\oint_{\mathrm{B}_j}dp_N,\quad |j|=0,1,\cdots,N,\nonumber\\ U_{j,N}&=\frac{1}{2}k_j:=-\frac{1}{2}\oint_{\mathrm{A}_j}dp_N,\quad, |j|=1,2,\cdots,N\nonumber
\end{align}}
and {$R_k$ are the coefficients of}
\begin{align}
P(z)=\frac{R(z)}{(1+\r_N(z))R_0(z)}=
    %\prod_{|j|=1}^N(z-z_j)=
    \sum_{k=0}^{2N}R_kz^{2N-k}.\label{R-k}
\end{align}
Let us introduce notations
\begin{align}\label{def-A}
    \mathbb{A}(f)&:=-\frac{1}{2\pi i}\sum_{|j|=0}^N\oint_{A_j}(\z+{W_{j,N}})\frac{f(\z)d\z}{R(\z)},\cr
    \mathbb{B}(f)&:=-\frac{1}{2\pi i}\sum_{|j|=1}^N\oint_{B_j}{U_{j,N}}\frac{f(\z)d\z}{R(\z)}.
\end{align}
Then, using equation  \eqref{rel-Y-I}, \eqref{R-k} and \eqref{Im-def}, we have
\begin{align}
    \tilde G_{m,N}=(\mathbb{A}+\mathbb{B})(\sum_{k=0}^{m}R_k\z^{2N+m-k})=(\mathbb{A}+\mathbb{B})(\z^{m}P(\z)),
    % &=Y_m+(\mathbb{A}+\mathbb{B})\le(\z^m\le(\prod_{|j|=1}^N(\z-z_j)-z^{2N}\ri)\ri)\cr
    % &=(\mathbb{A}+\mathbb{B})\le(\frac{\z^mR(\z)}{R_0(\z)(1+\r_N(\z))}\ri).
\end{align}
where we have used the fact that 
\begin{align*}
    (\mathbb{A}+\mathbb{B})(\z^l)=0,\quad l=0,1,\cdots,2N-1,
\end{align*}
that follows from \eqref{WO}.
% {[Use same variables, either $\xi$ or $\z$]}

Applying Lemma \ref{lem-Rd}  and noting that $W_{0,N}=0$, we have
\begin{align}\label{A-est}
    \mathbb{A}(\z^{m}P(\z))&=\mathbb{A}\le(\frac{\z^mR(\z)}{R_0(\z)(1+\r_N(\z))}\ri)=-\frac{1}{2\pi i}\oint_{A_0}\frac{\z^{m+1}d\z}{R_0(\z)}+O(\r_*(\delta,N)),
\end{align}
where $\r_*(\delta,N)$ is given by \eqref{rho-est}.

Denote $D_{j}=\{z:|z-z_j|\leq \sqrt 2|\d_j|\}$. Again, by Lemma \ref{lem-Rd}, we have
\begin{align}\label{B1-est}
    \mathbb{B}&\le(\frac{\z^mR(\z)}{R_0(\z)(1+\r_N(\z))}\ri)=-\frac{1}{2\pi i}\sum_{|j|=1}^N {U_{j,N}}\oint_{B_j}\frac{\z^m}{R_0(\z)(1+\r_N(\z))}d\z\cr
    &=-\frac{1}{2\pi i}\sum_{|j|=1}^N {U_{j,N}}\le(\int_{B_j\backslash D_j}+\int_{B_j\cap D_j}\ri)\frac{\z^m}{R_0(\z)(1+\r_N(\z))}dz.
 \end{align}   
 
 Consider first the case when $\g_0$ is separated from $\G$, i.e., when $R^{-1}_0(z)$  is bounded on all $B_j$ uniformly  in $N$.
 Then, similarly to  \eqref{A-est}, each integral $\int_{B_j\backslash D_j}$ can be approximated by  the corresponding integral  $\int_{B_j}$ with the accuracy
 $O(\r_*(\delta,N))+O(\d)$, where the second term is an estimate of 
 $\int_{D_j}\frac{\z^m}{R_0(\z)}d\z$.
 
According to  Lemma \ref{lem-Rd}, part b)
\be\label{B2-est}
\int_{B_j\cap D_j}\frac{\z^m}{R_0(\z)(1+\r_N(\z))}d\z=
\int_{B_j\cap D_j}\frac{\z^m(\z-z_j)}{R_0(\z)R_j(\z)}(1+\r_*(\d,N))d\z
 \ee
 Note that $|\frac{z-z_j}{R_j(z)}|\leq 1 $ on the contour $B_j$ that locally is  orthogonal to $\g_j$ and crosses it at $z_j$. Thus, the integrand in   \eqref{B2-est} is (uniformly in $j,N$)  bounded and so the integral in \eqref{B2-est} is of order $O(\d)$. Thus, we obtain

\begin{align}\label{B-fin-est}
    \mathbb{B}&\le(\frac{\z^mR(\z)}{R_0(\z)(1+\r_N(\z))}\ri)=-\frac{1}{2\pi i}\sum_{|j|=1}^N {U_{j,N}}\oint_{B_j}\frac{\z^md\z}{R_0(\z)}(1+O(\r_*(\delta,N))+O(\delta)).
\end{align}
Together with \eqref{A-est} this yields
\begin{align}\label{G-est}
    \tilde G_{m,N}=-\frac{1}{2\pi i}\le[\sum_{|j|=1}^N {U_{j,N}}\oint_{B_j}\frac{\z^md\z}{R_0(\z)}+\oint_{A_0}\frac{\z^{m+1}d\z}{R_0(\z)}\ri]
    (1+O(\r_*(\delta,N))+O(\delta)).
\end{align}
Consider now the case when $z_j$  is on the distance $O(1/N)$ from $\g_0$. Then the integrand in \eqref{B2-est} is not bounded on 
$B_j\cap D_j$. The corresponding calculations show that in this case  the $O(\delta)$ term in \eqref{B-fin-est} and \eqref{G-est} should be replaced by $O(N^\hf\delta)$. 

Going back to expansion
%\eqref{new-expansion} and 
\eqref{g-x-expansion}, we have
\begin{align}
\sum_{m=0}^\infty\frac{ G_{m,N}}{z^m}= (1+\r_N(z))R_0(z)\sum_{m=0}^\infty\frac{\tilde G_{m,N}}{z^{m+1}}.
\end{align}
%where $\tilde G_{m,N}$ can be determined as follows.
% {[I suggest to interchange the tilde: the coefficients of $2g_x$ better be without tilde]}
Note that 
\be
R_0(z)=\sum_{k=-1}^\infty d_k\d_0^{k+1}z^{-k},
\ee
where $d_k$ is defined in \eqref{dm-int} and $d_{-1}=1,$

By the Cauchy's estimates and Lemma \ref{lem-Rd} we conclude that
all  the Taylor coefficients of $\r_N(z)$ at $z=\infty$ are of the order $O(\r_*(\d,N))$. This estimate will be uniform for all the Taylor coefficients provided that the compact $\G$ is contained inside the unit circle $|z|=1$.

Thus, we have
\begin{align*}
     G_{m,N}&=
     \le(\sum_{k=-1}^{m-1}\tilde G_{m-1-k,N}d_{k}\ri)\le(1+O(\r_*(\delta,N)\ri)\\
     &=-\frac{1}{2\pi i}\sum_{k=-1}^{m-1}\le[\sum_{|j|=1}^N {U_{j,N}}\oint_{B_j}\frac{\z^{m-1-k}d\z}{R_0(\z)}+\oint_{A_0}\frac{\z^{m-k}d\z}{R_0(\z)}\ri]d_{k}\d_0^{k+1}(1+O(\delta))\\
     &=-\frac{1}{2\pi i}\le[\sum_{|j|=1}^N {U_{j,N}}\oint_{B_j}\frac{[\z^{m-1} R_0(\z)]_+d\z}{R_0(\z)}+\oint_{A_0}\frac{\z[\z^{m-1}R_0(\z)]_+d\z}{R_0(\z)}\ri](1+O(\delta)),
\end{align*}
where we have taken account that $O(\d)$ is much larger than
$O(\r_*(\delta,N)$. Note that
\begin{align*}
    \frac{1}{2\pi i}\oint_{A_0}\frac{\z[\z^{m-1}R_0(\z)]_+d\z}{R_0(\z)}=d_{m}\d_0^{m+1}, m\in\N.
\end{align*}

Then we use the relation that
\be
I_{m,N}= -mG_{m,N},\label{G-I relation}
\ee
thus, we have
%o compute $I_{m,N}$, note we have
\begin{align}
    I_{m,N}&=\le(\frac{m}{2\pi i}\le[\sum_{|j|=1}^N {U_{j,N}}\oint_{B_j}\frac{[\z^{m-1} R_0(\z)]_+d\z}{R_0(\z)}\ri]+md_m\d_0^{m+1}\ri)(1+O(\d)),
\end{align}
and, taking into account that $O(\d)$ should be replace by $O(N^{1/2}\d)$ when $z_j$ is on the distance $O(1/N)$ from $\g_0$, equation \eqref{br-avg-den} follows.

\end{proof}

\br 
It follows from \eqref{rel-Y-I},\eqref{G-est} that
\be\label{2gx-inf}
2g_{x}(\infty)=
\le[-\frac{1}{2\pi i}\sum_{|j|=1}^N {U_{j,N}}\oint_{B_j}\frac{dz}{R_0(z)}\ri]
    (1+O(N^{1/2}\delta)).
\ee
% \red{[Something wrong with $m$ here since $c_1$ should be zero.]}
\er

\bc In the thermodynamic limit of fNLS soliton gas we have
\begin{align}\label{sol-avg-den}
    I_{m,N}=\le(\frac{m}{2\pi i}\sum_{|j|=1}^N {U_{j,N}}\oint_{B_j}\z^{m-1}d\z\ri)
    %+O(\r_*(\delta,N)
    (1+O(N^{1/2}\delta)),
\end{align}
where we used the notations from Theorem \ref{theo-ImN}.
\ec

\subsection{Approximation of $u$}\label{sect-approx-u} 

The accuracy of approximation of $I_m$ with $I_{m,N}$, $N\ra\infty$, depends on the rate of convergence of measures $\l_N$ to $ud\l$.  In this subsection we study this rate of convergence under some additional assumptions on the thermodynamic limit. 
%We start with the following lemma.

First,  here and henceforth we assume that
the condition   \eqref{varphi-def} (with some $\beta>0$) about the  continuous and positive on $\G^+$ probability density of bands $\varphi(z)$ 
%\be
%\l(\G_j)=\frac{1}{N\varphi(z_j)}+O(\frac{1}{N^{1+\b}})
%\ee
is valid uniformly in $\G^+$.

\br\label{rem-cont}
Even though some results of this subsection can be extended to 2D compact sets $\G^+$, in order to  simplify our exposition in this subsection  we assume that $\G^+$ is a contour.
\er

\begin{lemma}\label{lem-M-pos}
  For any $N\in\N$ and any Schwarz symmetrical hyperelliptic Riemann surface $\mathfrak{R}_N$, the matrix $M_N= -\Im \oint_{\tB_k}\frac{P_j(\z)d\z}{R(\z)}$, {where $\tB_k=B_k\cup B_{-k}$,} is symmetric and positive definite.
  Moreover, the large $N$ asymptotics of $-M_N$ is given by \eqref{lead-int-ass-evB}.
 \end{lemma}
 
 This lemma follows from the well known properties of the Riemann period matrix of $\mathfrak{R}_N$. The last statement is the subject of Lemma \ref{lem-est-evenB}.

Next, we assume that  
%$\varphi(z)$ is a continuous positive density of a  probability measure $\varphi(z)|d z|$ on $\G^+$ and  
$\nu(z)\in C(\G^+)$ and that
\be\label{Fred-2}
\min_{\nu\in\G^+} \nu(\eta)=\nu_0>0
\ee
on the compact $\G^+$.  The results for the rest of this subsection are obtained under this assumption. In particular, it implies that  
 the 1st NDR equation \eqref{dr_soliton_gas1} for the fNLS soliton gas   written, according to \eqref{dens_states}, as
\be\label{u-eq}
 \int_{\G^+}\ln \le|\frac{\m-\bar\eta}{\m-\eta}\ri|
u_x(\m)\varphi(\m)|d\m|+2\nu(\eta)u_x(\eta)= \Im \eta,
\ee
where $u_x:=\frac u\varphi$, is a Fredholm integral equation of the second type;  therefore, the existence (and uniqueness) of the solution $u_x$ is  guaranteed. The same is true for the corresponding breather gas equation \eqref{dr_breather_gas1}). The fact that $u_x\geq 0$ on $\G^+$ in the case of the soliton gas was proven in \cite{KT2021}.

 \begin{lemma}\label{lem-norm-M-inv}
  Under the condition \eqref{Fred-2},  the spectral radius of  $NM_N^{-1}$ in the thermodynamic limit does not 
  exceed $\frac{\pi}{2\nu_0}+O(1/N)$ for all sufficiently large $N\in\N$.
  %$N\ra\infty$.
 \end{lemma}
\begin{proof}
 The first system \eqref{WFR} can be written as 
 \be\label{Ueq1}
 -\sum_{|k|=1}^{N}\mathring{u}_k\Im\frac 1N\oint_{B_k}\frac{P_j(\z)d\z}{R(\z)}=\Im z_j, \qquad j=1,\dots,N,
 \ee
where   $\frac{k_j}{2}=U_{j,N}=\frac{ \mathring u_j}{N}$. According to Lemma \ref{lem-M-pos}, the matrix $\frac 1N M_N=:\hat M_N(\nu)$ is positive definite for any $\nu\geq 0$ and all $N\in\N$, that is the spectrum of $\hat M_N(0)$ is positive for  all $N\in\N$. Next, according to Lemma \ref{lem-est-evenB},    
\be
\rm{diag}\;\hat M_N(\nu)= \rm{diag}\;\hat M_N(\nu-\nu_0)+ \left[\frac{2\nu_0}{\pi}+O(1/N)\right]\1_N,
\ee
 where
$\1_N$ denotes the identity matrix of size $N$. 
Note that both matrices $M_N(\nu)$ and its asymptotic limit $M^a_N(\nu)$, given through \eqref{lead-int-ass-evB}, are symmetric and therefore, both have real spectrum. Arrange the eigenvalues of both matrices in the descending order. Then, 
according to Weyl's Perturbation Theorem, see Theorem VI.2.1, p. 156, \cite{Bhatia}, the distance between the  corresponding eigenvalues is uniformly bounded by $O(1/N)$. That means that any possible negative eigenvalue of  $M^a_N(\nu-\nu_0)$ must be of order  $O(1/N)$,
since the matrix $\hat M_N(\nu-\nu_0)$ is positive definite  all $N\in\N$.
Then the spectrum of $M^a_N(\nu)$ and of $\hat M_N(\nu)$ are bounded from below by $2\nu_0/\pi + O(1/N)$. Thus,
we obtained the 
desired spectral estimate for
 $\hat M^{-1}_N(\nu)$.
\end{proof}

\begin{remark}
It is well known (see, for example, \cite{SaTo}) that the integral operator $G$ in \eqref{u-eq} (applied to $u(\m)=u_x(\m)\varphi(\m)$) expressing the 
Green's potential of $u$ 
is positive definite. Therefore, arguments similar to those used in Lemma \ref{lem-norm-M-inv} 
 show that  the  spectrum of the operator 
$G+\s$ is situated on $[\s_0,\infty)$, where
\be
\s(\eta)\geq \s_0>0 \quad \text{on $\G$}.
\ee
Thus, we conclude that the operator $(G+\s)^{-1}$ has inverse
bounded by $\s_0^{-1}$ in the appropriate functional space.
 
\end{remark}

Let for a fixed large $N\in\N$ the points $z_j:=z_{j,N}\in \G^+$, where $j=1,\dots,N$, denote the centers of the bands that are distributed on $\G^+$ according to $\varphi(z)$. 
We   want to know how well the values of 
\be\label{uU}
\hat u_j=\frac{1}{\pi} \varphi(z_j)NU_{j,N},  
 \ee
approximate $u(z_j)$. Another question is the approximation of 
 the solution $u(z)$ of \eqref{u-eq} by a piecewise constant
function
\be
\hat u(z)=\hat u_N(z):=\sum_{j=1}^N \hat u_j\chi_j(z),
 \ee
where 
 $\chi_j$ is a characteristic function of a simple arc of $\G_j^+$ containing $z_j$ and going half way to the neighboring points $z_{j\pm 1}$.  
 %The vector $U_N=(U_{1,N},\dots, U_{N,N})$ in \eqref{uU} is defined through the solution

 \begin{theorem}\label{th-v>0}
Let the assumptions \eqref{Fred-2} and \eqref{varphi-def} with some $\beta >0$  hold. 
If  $u(z)$ is $\a$-H\"older continuous  on $\G^+$ with some $\a\in(0,1]$. Then: 

(a)  the discrete measure $\l_N$ from Theorem \ref{theo-ImN}, part (ii),  weakly* converges to $u\l$ with  accuracy $O(N^{-\varrho})$, where $\varrho=\min\{\a,\beta\}$; if $\varrho=1$, then the  accuracy $O(N^{-\varrho})$ should be replaced by  $O(\frac{\ln N}{N})$;
%$\hat u(z)$ weakly converges to $u(z)$ as $N\ra\infty$ uniformly on $\G^+$.

(b) %If  $u(z)$ is Lipschitz on $\G^+$ then 
$|u(z)-\hat u_N(z)|=O(N^{-\varrho})$ 
 as $N\ra\infty$ uniformly on $\G^+$.
If $\varrho=1$, then the  accuracy $O(N^{-\varrho})$ should be replaced by  $O(\frac{\ln N}{N})$.
 \end{theorem}

\begin{proof}
 (a) Substitution of 
 %$\tilde u(z):=
 $\sum_{j=1}^N u(z_j)\chi_j(z)$ into \eqref{u-eq} yields
 \be\label{u-eq1}
 \sum_{j=1}^N u(z_j)\int_{\G^+_j}g(w,z)|dw|+\s(z) \sum_{j=1}^N u(z_j)\chi_j(z)=\phi(z)+E_1(z),
 \ee
 where: $g(w,z)$ denotes the kernel of the integral operator and $\phi(z)$ denotes the right hand side in  \eqref{u-eq}, $\G^+_j=\supp \chi_j$ and  $E_1(z)$ is the error term. It is straightforward to check that 
 %If $u(z)$ is Lipschitz on $\G^+$ then  
 $E_1(z)=O(N^{-\a})$ 
 uniformly on $\G^+$.
 
 We now choose $z=z_k$, $k=1,\dots, N$, so that \eqref{u-eq1} becomes
\be\label{u-eq2}
 \sum_{j\neq k} u(z_j)\int_{\G^+_j}g(w,z_k)|dw|+u(z_k)\left[\s(z_k)+\int_{\G^+_k}g(w,z_k)|dw|\right] =\phi(z_k)+E_1(z_k).
 \ee
 
Using the mean value theorem, we can write $\int_{\G^+_j}g(w,z_k)|dw|=g(w_j)|\G_j|$, where $|\G_j|$ is the arclength of $\G_j$ and $w_j\in\G_j$.  Then the error $E_2$ in replacing $ \sum_{j\neq k} u(z_j)\int_{\G^+_j}g(w,z_k)|dw|$ with 
$\sum_{j\neq k} u(z_j)g(z_j,z_k)|\G_j|$ can be estimated as 
\bea\label{E2}
\sum_{j\neq k} u(z_j)\le|\int_{\G^+_j}g(w,z_k)|dw| -g(z_j,z_k)|\G_j|\ri|=
\sum_{j\neq k} u(z_j)\le|g(w_j,z_k) -g(z_j,z_k)\ri||\G_j|\cr
\leq 
\max_{j}(u(z_j)N|\G_j|)\frac 1N \sum_{j\neq k} \max_{w_j\in\G_j}\le|g(w_j,z_k) -g(z_j,z_k)\ri|\quad
\eea

We now want to split the terms in the last into two categories: those that are ``close'' to $z_k$ and those that 
are ``away'' from $z_k$. To arrange such a split 
we notice that for a  given $\G^+$ there exits some $s>0$ such that $g(w,z)$ becomes a monotonic function of $w$
for any fixed $z\in \G^+$ whenever $|w-z|<s$ and all $w$ are ``on the same side'' of $z$.
Denote part of $\G^+$ that is $s$-close to $z_k$ by $D_k$.
We then  split the latter sum into $J_1:=\{j: z_j\in D_k,\,j\neq k\}$ and the remaining part $J_2$. If $j\in J_2$,
then $|g(w_j,z_k)-g(z_j,z_k)|=O(1/N)$ uniformly over $\G^+\setminus D_k$. This part of sum in \eqref{E2} 
can be  estimated by $O(1/N)$ uniformly in $z_k\in\G_+$. For any remaining $j\in J_1$ we get
  \be 
 \max_{w_j\in\G_j}\le|g(w_j,z_k) -g(z_j,z_k)\ri|\leq  \pm[g(\a_j,z_k) -g(\b_j,z_k)] 
 \ee
 depending on which side of $z_k$ the point $z_j$ is.  Here $\a_j,\b_j$ denote the beginning and end points of the subarc $\G_j$ of the oriented curve $\G^+$. 
 Let us show that this part of the sum is of the order $O(\frac{\ln N}{N})$ uniformly in $z_k\in\G_+$.
 Indeed, replacing $|g(w_j,z_k) -g(z_j,z_k)|$ with $g(\a_j,z_k) -g(\b_j,z_k)$ on the proper side of $z_k$
 we see that 
 \be
 \frac 1N\sum_{j\in J_1} \max_{w_j\in\G_j}\le|g(w_j,z_k) -g(z_j,z_k)\ri|\leq \frac 2N \max_{j\in J_1} |g(z_j,z_k)|=O(\frac{\ln N}{N})
 \ee
 uniformly  in $z_k\in\G_+$. Clearly, we also have 
 \be\label{int-diag}
 \int_{\G^+_k}g(w,z_k)|dw| =O(\frac{\ln N}{N}).
\ee

 Since  $\max_{j}(u(z_j)N|\G_j|)$ and $\sum_{j\neq k}\frac{ g(z_j,z_k)}{N}$
 are bounded (in $N$), it follows from \eqref{varphi-def}, \eqref{u-eq2} and previous analysis that 

 \be\label{u-eq3}
 \sum_{j\neq k}\frac{ u(z_j)g(z_j,z_k)}{N\varphi(z_j)}+u(z_k)\s(z_k) =\phi(z_k)+ O(\frac{\ln N}{N})+O(N^{-\min\{\a,\beta\}})
 \ee
  uniformly  in $z_k\in\G_+$.
  %where we also took into account  that $\int_{\G^+_k}g(w,z_k)|dw|=O(\frac{\ln N}{N})$. 
  Then, according to \eqref{dens_states},
 
 \eqref{uU} and \eqref{lead-int-ass-evB}, we convert \eqref{u-eq3} into
  \be\label{u-eq4}
 -\sum_{j=1}^N\frac{ \check u_j}{N}\Im\oint_{\bar B_k}\frac{P_j(\z)d\z}{R(\z)} =\phi(z_k)+ O(\frac{\ln N}{N})+O(N^{-\varrho}).
 \ee
  Since $U_{j,N}=\frac{ \mathring u_j}{N}$ and $\phi(z_k)=\Im z_k$, we have recovered \eqref{Ueq1}  subject to the 
  %$O(\frac{\ln N}{N})$ 
  error terms.
 Rewriting \eqref{u-eq4} as
  \be\label{u-eq5}
 -\sum_{j=1}^N\check u_j \Im\oint_{\bar B_k}\frac{P_j(\z)d\z}{R(\z)} =N\Im z_k+ O({\ln N})
 +O(N^{1-\varrho})
 \ee
 we obtain that, according to Lemma \ref{lem-norm-M-inv},  
 \be\label{ineq-uU}
 |\check u_j - NU_{j,N}|=O(\frac{\ln N}{N})+O(N^{-\varrho}).
 \ee
 Statement (a) of the theorem follows from \eqref{int-diag}. 

 Since \eqref{int-diag} implies $|u(z)-\hat u(z)|=O(\frac{\ln N}{N})+O(N^{-\varrho})$ as $N\ra\infty$ uniformly on $\G^+$, the proof is completed.
\end{proof}

\begin{remark}
Theorem \ref{th-v>0} justifies transition from the first system of linear equations  \eqref{WFR} for the wavenumbers $k_j$ to the integral equation  \eqref{u-eq} in the NDR.  Similar result should be valid for the second system of linear equations \eqref{WFR}.  Assumption 
\eqref{Fred-2} was  used in the proof of Theorem \ref{th-v>0}. 
 \end{remark}

 \begin{remark}\label{rm-br-app}
 %Under   the assumption \eqref{Fred-2},  
 Theorem \ref{th-v>0} should be valid for breather gases as well.
 \end{remark}

\subsection{Proof of Theorem \ref{theo-ImN} part (ii)}
According to Theorem \ref{th-v>0} and Remark \ref{rm-br-app}, we can now prove part (ii) of Theorem \eqref{theo-ImN}.
% \red{[Haven't wealready  proved it right before Sect \ref{sect-approx-u}? ]}

% \red{[It seems that  $F_m(\a_1)=-\frac {\pi i}2 r_m $, where  $r_m$ is the coefficient of $z^{-m}$ of the expansion of $z^{-1}R_0(z)$ at infinity.]}
\begin{proof}{\it of Theorem \ref{theo-ImN}, part (ii).}
By the assumption of Theorem \ref{theo-ImN}, part (ii), we have
\begin{align*}
    \l_N\xrightarrow[]{*} ud\l,\quad N\ra \infty,
\end{align*}
where $\l_N=\sum_{|j|=1}^N\frac{U_{j,N}}{\pi}\d(z-z_{j,N})$. Since $F_m$ is continuous on $\G^+$, we have
\begin{align}
    m\int_{\G^+}  \le(2\Im F_m(z)-\oint_{A_0}dF_m\ri) d\l_N(z)\rightarrow 2m\int_{\G^+}u(\z)\Im{F_m(\z)}d\l(\z), \quad N\ra\infty,
\end{align}
where we used $U_{j,N}=U_{-j,N}$, equation \eqref{br-avg-den}, and the fact that
$ \oint_{A_0}dF_m(z)=0,\quad m\in\N.$ Substitute back to \eqref{br-avg-den}, the result \eqref{inv-breather-t1} follows.
\end{proof}
\br
The assumptions of Theorem \ref{th-v>0} imply the weak convergence of $\l_N$ to $u\l$. One of these assumptions  is  the requirement that $\G^+$ is 1D compact (contour). 
\er

\begin{corollary} 
In the thermodynamic limit of fNLS soliton gas, we have
\begin{align}\label{inv-sol}
    I_m=2\int_{\G^+}u(\z)\Im{\z^m}d\l(\z),~~m\in \N.
\end{align}
\end{corollary}

\br\label{rem-tilde G} 
Let $2g_x$ be the limiting $g$-function defined in Section \ref{sec-dens-of-diff} below.
Repeating arguments of Theorem \ref{theo-ImN} for \eqref{G-est}, we calculate the Taylor coefficients $\tilde G_m$ of 
$2g_x/R_0=\sum_{m=0}^\infty\frac{\tilde G_m}{z^{m+1}}$ for breather gas 
as 
\be\label{tildeGm}
\tilde G_m= -2\int_{\G^+} u(\z)\Im\int_0^\z  \frac{\m^m d\m}{R_0(\m)} d\l(\z) -\frac{1}{2\pi i}\oint_{A_0}\frac{\z^{m+1}d\z}{R_0(\z)}.
\ee
\er

%Now we want to clarify  what we mean by the limiting density $\varphi(z)$. 

{\it Example.}
Consider a special density of states (see equation \eqref{per-u-br}):
\begin{align}
    u(z)=\frac{|z|}{\pi L}\int_{0}^{\a L}\frac{dx}{\sqrt{q^2(x)+z^2}}, \quad z\in \G^+\subset i\R_+,
\end{align}
where {$q(x)=Q\chi_{[0,\a L)}+q\chi_{[\a L,L]}, \quad Q\geq q\geq 0, \a\in (0,1]$, $x\in [0,L]$} and $q(x)=q(-x)$ for $x\in [-L,0]$. Then  $q(x+2L)=q(x)$ and $\G^+=[iq,iQ]$. Such a density of states $u$ is a periodic breather/soliton gas that will be discussed in Section \ref{sec-per-breather} below. In fact, for such $q(x)$, we have 
\begin{align}
    u(z)=\frac{-iz\a}{\pi \sqrt{Q^2+z^2}}, \quad z\in \G^+.
\end{align}
Since $R_0(\z)=\sqrt{\z^2+q^2}=\sum_{k=-1}^\infty d_k(iq)^{k+1}\z^{-k}$  in a neighborhood of $\z=\infty$, we have $R_0(\z)^{-1}=\z^{-1}\frac{dR_0(\z)}{d\z}=-\sum_{k=-1}^\infty kd_k(iq)^{k+1}\z^{-(k+2)}$ in a neighborhood of $\z=\infty$.
Then
\begin{align*}
    [\z^{m-1}R_0(\z)]_+R_0(\z)^{-1}&=-\le(\sum_{k=-1}^{m-1} d_k(iq)^{k+1}\z^{m-1-k}\ri)\le(\sum_{k=-1}^\infty kd_k(iq)^{k+1}\z^{-(k+2)}\ri)\\
    &=\z^{m-1}-d_m(iq)^{m+1}\z^{-2}+O(\z^{-4}),
\end{align*}
where we used the fact that all coefficients of $R_0(\z)R_0^{-1}(\z)-1$ vanish. Plugging into \eqref{def-F-m}, and taking the imaginary part, we have
\begin{align}
    \Im F_m(\z)=\begin{cases}
    0,& \text{$m$ is even,}\\
    -i\le(\frac{1}{m}\z^m+d_m(iq)^{m+1}\z^{-1}+O(\z^{-3})\ri), &\text{$m$ is odd.}
    \end{cases}\label{Im-Fm}
\end{align}
Thus, applying formula \eqref{inv-breather-t1}, we have, for odd $m$,
\begin{align}\label{Iqm}
    I_m&=\frac{m}{2}\oint\frac{(-i)^3z\a}{\pi \sqrt{Q^2+z^2}}\le(\frac{1}{m}z^m+d_m(iq)^{m+1}z^{-1}+O(z^{-3})\ri)dz+md_m(iq)^{m+1}\nonumber\\
    &=\frac{mi\a}{2\pi}\oint(1+(Q/z)^2)^{-1/2}\le(\frac{1}{m}z^m+d_m(iq)^{m+1}z^{-1})\ri) dz+md_m(iq)^{m+1}\nonumber\\
    &=-m \Res_{z=\infty}\le\{(1+(Q/z)^2)^{-1/2}\le(\frac{1}{m}z^m+d_m(iq)^{m+1}z^{-1})\ri)\ri\}+md_m(iq)^{m+1}\nonumber\\
    &=-m\a (-d_m(iQ)^{m+1}+d_m(iq)^{m+1}) +md_m(iq)^{m+1}\nonumber\\
    &=-\frac{i^{m+1}m!!}{(m+1)!!}(\a Q^{m+1}+(1-\a)q^{m+1})=
    -\frac{i^{m+1}m!!}{(m+1)!!}
    \braket{q^{m+1}(x)},
\end{align}
where $d_m$ is defined in \eqref{dm-int} and $\oint$ denotes the integral over a clockwise loop enclose $\G$ and $ \braket{\cdot}$ denotes the average over the period.
For even $m$, $I_m=0$. 

It is well-known that the densities $f_k$ for fNLS conserved quantities satisfy the following recursion relation (see \cite{Wadati}):
\begin{align}
    f_{n+1}&=\frac{1}{2}\left(\sum_{k=1}^{n-1}f_{k}f_{n-k}-q(x)\le(\frac{f_n}{q(x)}\ri)_x\right),\quad n\in \N\cr
   f_{1}&=\frac{1}{2}|q(x)|^2.    
\end{align}
Let $q(x)$ be a piece-wise constant and  periodic function. Computing the limiting average of $f_n$ over $[-T,T]$, where $T\ra\infty$ (see, for example,  equation (3.3a,b) in \cite{FFM1980}),
we obtain
\begin{align}\label{fn}
    \braket{f_m}=a_m\braket{|q(x)|^{m+1}},\quad m\in \N,
\end{align}
where   $   a_1=\frac{1}{2}, $ and for $m\geq 2$:
\begin{align*}
    a_{m}&=\begin{cases}
    0, & \text{$m$ is even},\\
    \frac{1}{2}\sum_{k=1}^{m-2}a_ka_{m-1-k}, & \text{$m$ is odd}.
    \end{cases}
\end{align*}
Note that for a periodic function the limiting average coincides with the average over the period.

It is easy to check that $a_m=-d_m$. Comparing \eqref{Iqm} with \eqref{fn}, we obtain
\begin{align}
    \braket{f_m}=(-1)^{\frac{m+3}{2}}m^{-1}I_m, \quad m\in \N.
\end{align}

\subsection{Averaged conserved fluxes under   the fNLS and higher times  flows}
\label{sect-higher}
{
	It is well known that the complete integrability of the fNLS equation implies there are infinitely many conversed quantities \cite{ZS}. In fact, the fNLS flow is the second flow in the  so-called focusing Zakharov-Shabat(ZS) hierarchy (for a full detailed description  of the hierarchy see \cite{FT2007}). 
	The RHP \ref{RHPY} for finite gap fNLS solution can be extended to include higher  ZS hierarchy flows if the function $f$ in \eqref{RHPY-jump} is replaced by
	%as the one for fNLS in section \ref{sec-background} by generalizing the function $f$ to 
	\begin{align}
		f(z|\bm{t})=f_0(z)+\sum_{l=1}^\infty t_lz^l,
	\end{align}
	where  $f_0(z)$ is  the same as in  \eqref{RHPY-jump} and  $\bm{t}=(t_1,t_2,\cdots)$,  with $t_1=x$, 
	$t_2=2t$ and  $t_3,t_4,\dots$ denoting higher flows times.  Here $x,t$ are space-time variables for the fNLS.  
	 By letting  all $t_l=0,\,l\geq 3$,  one recovers the $f$ for the fNLS equation. In general, for $l$th flow in the ZS hierarchy one takes
	\begin{align}
		t_1=x,~~~t_l=T_jt\delta_{j,l},~~~ l=2,3,\cdots,
	\end{align}
	where $\delta_{j,l}$ is the Kronecker delta,  $t$ is the time variable for the $j$th hierarchy equation  and {$T_j>0$ is a certain constant. For example,  $j=3$ corresponds to the mKdV equation with $t_3=4t$.
}

{
	Similarily to the fNLS quasienergy meromorphic differential $dq_N=dq_{2,N}$ one can (uniquely) define the second kind real normalized meromorphic differentials $dq_{j,N}$, $j=2,3,\dots$,  on the Riemann surface $\mathfrak{R}_N$ assocaited with  averaged conserved fluxes  of the fNLS and higher order flows  with the only singularities at infinity
	\begin{align}
		dq_{j,N}(z)\sim [\pm jT_jz^{j-1}+O(z^{-2})]dz, \label{gen-diffs}
	\end{align}
of the main and the secondary sheet  respectively. Accodring to \cite{ForLee}, the  averaged conserved fluxes of the $j$th flow  are the coefficients $I_{m,j,N}$ of the expansion %at infinity of
	\begin{align}
		dq_{j,N}=\le(jT_jz^{j-1}-\sum_{m=1}^\infty\frac{I_{m,j,N}}{z^{m+1}}\ri)dz.
	\end{align}

One can now follow the approach described in Section \ref{sec-background} to obtain  the breather gas   second NDR equation  
	\begin{align}\label{j-NDR-brea}
	\int_{\Gamma^+}
	\left[\log\left| \frac{w-\bar z}{w-z}\right|+\log\left|\frac{R_0(z)R_0(w)+ z w-\delta_0^2}
	{R_0(\bar z)R_0(w)+ \bar z w-\delta_0^2}\right|\right] v_j(w) d\lambda(w)
	+\sigma(z)v_j(z) \\
	= - \Im[jT_jz^{j-1} R_0(z)],\quad j=2,3,\cdots,
\end{align}
for the $j$th  hierarchy flow, where  $v_j(z)$ denotes the $j$th order analogue of the density of fluxes $v(z)=v_2(z)$. 	By taking the limit $\delta_0\ra 0$, we obtain the 
soliton gas  second NDR equation  
\begin{align*}
	\int_{\Gamma^+}
	\log\left| \frac{w-\bar z}{w-z}\right| v_j(w) d\lambda(w)
	+\sigma(z)v_j(z) 
	= - \Im[jT_jz^{j}],\quad j=2,3,\cdots,
\end{align*}
for the $j$th  hierarchy flow,
It is worth noting that these equations in  the $j=2$ become the NDR equations \eqref{dr_breather_gas2} and \eqref{dr_soliton_gas2} respectively.

As in the case of the fNLS flow (see \eqref{kdp}), the periods $\o_{l,j}=-\oint_{A_l}dq_{j,N}$, $l=1,\dots, N$, represent the solitonic frequencies of the $j$th flow, $j=2,3,\dots.$ Then the average conserved fluxes in the $j$th  flow (of ZS-hierarchy) are summarized at the following theorem
\begin{theorem}
	\label{gen main thm}
	\begin{itemize}
		\item[(i)]	Fix $m\in\N$,  then for a sufficiently large $N$ in the thermodynamic limit of the breather gases for the $j$th fNLS flow the  averaged conserved fluxes are
		\begin{align}\label{Imjn}
			I_{m,j,N}=&\frac{m}{2\pi i}\le(\sum_{|l|=1}^N U_{l,N}^{[j]}\oint_{B_l}\frac{[\z^{m-1}R_0(\z)]_+d\z}{R_0(\z)}+\oint_{A_0}\frac{T_j\z^j[\z^{m-1}R_0(\z)]_+d\z}{R_0(\z)}\ri) \cr
			&\times(1+O(N^\hf\delta)), \qquad \text{	where } \quad  U_{l,N}^{[j]}:=\hf \o_{l,j};
		\end{align}
		%-\frac{1}{2}\oint_{A_l}dq_{j,N}$.

		\item[(ii)] Let $v_j$ solves \eqref{j-NDR-brea}. If the measures $\l_{j,N}:=\sum_{l=1}^N\frac{U_{l,N}^{[j]}}{2\pi}\d(z-z_{l,N})$, where $z_{l,N}\in\G^+$ denotes the center of the $l$th bands of $\mathfrak R_{N}$, $l=1,\dots,N$, and $\d(z)$ denotes the delta-function,
		%centered at $z\in\G^+$,
		are weakly* convergent to the (signed) measure $v_j(z)d\l(z)$ on $\G^+$
		%, where $u(z)$ solves \eqref{dr_breather_gas1}, 
		then the thermodynamic limit of   $I_{m,N}$ is given  by
		%Under the assumption of thermodynamic limit, the limiting  conversed quantities of the higher order fNLS  flow are	
		\begin{align}
			I_{m,j}=2m\int_{\Gamma^+}v_j(\z)\Im{F_{m}(\z)}d\l(\z)+\frac{m}{2\pi i}\oint_{A_0}\frac{T_j\z^j[\z^{m-1}R_0(\z)]_+d\z}{R_0(\z)},\quad m,j\in \N,\label{gen-br-con}
		\end{align}
		where $F_m(\z)$ is defined in \eqref{def-F-m}.
		Moreover, the second integral term in \eqref{gen-br-con} can be explicitily computed in terms of $d_k$, it reads
		\begin{align}
			\frac{1}{2\pi i}\oint_{A_0}\frac{T_j\z^j[\z^{m-1}R_0(\z)]_+d\z}{R_0(\z)}=-T_j\le(\sum_{k=m}^{m+j-1}(m+j-k-2)d_kd_{m+j-k-2}\ri)\delta_{0}^{m+j}.\label{gen-br-second-term}
		\end{align}
	\end{itemize}
\end{theorem}

\begin{proof}
	The proof is essentially the same as the proof of Theorem \ref{theo-ImN} except that  in the current proof we modify the definition of operator $\mathbb{A}$, see \eqref{def-A}, as
	\begin{align} 
		\mathbb{A}_j(f)&:=-\frac{1}{2\pi i}\sum_{|l|=0}^N\oint_{A_l}(T_j\z^j+{U_{l,N}^{[j]}})\frac{f(\z)d\z}{R(\z)}.
	\end{align}
	%Then the same argument in the proof of Theorem \ref{theo-ImN} with this new defined $\mathbb{A}_j$ leads to the current statement. 
	 Applying  the residue theorem to the integral in 
	\eqref{gen-br-second-term}:
	\begin{align*}
		[z^{m-1}R_0(z)]_+R_0(z)^{-1}&=\sum_{l=1}^\infty\le(\sum_{k=-1,k\leq m-1}^{l-1}d_k(l-2-k)d_{l-2-k}\ri)\delta_0^{l}z^{m-l-1},\\
		&=\sum_{l=1}^\infty\le(-\sum_{k=m}^{l-1}d_k(l-2-k)d_{l-2-k}\ri)\delta_0^{l}z^{m-l-1},
	\end{align*}
one obtains
	\begin{align*}
		\frac{1}{2\pi i}\oint_{A_0}\frac{T_j\z^j[\z^{m-1}R_0(\z)]_+d\z}{R_0(\z)}&=T_j\Res\le\{\z^j[\z^{m-1}R_0(\z)]_+R_0(\z)^{-1};\z=\infty\ri\}\\
		&=-T_j\le(\sum_{k=m}^{m+j-1}(m+j-k-2)d_kd_{m+j-k-2}\ri)\delta_{0}^{m+j},
	\end{align*}
which completes the proof.
\end{proof}

Taking the limit $\delta_0\ra 0$, we get the corresponding formulae for the soliton gas. 

\bc  \label{cor-jth}
In the case of the
soliton gas  for the $j$th fNLS flow, i.e. when $\g_0$ is one of the shrinking bands, the equations \eqref{Imjn} and \eqref{gen-br-con} become
%Under the assumption for soliton gas limit, the averaged conserved quantities of the $j$-th flow for large $N$ and for thermodynamic limit  are resepctively given by
\begin{align}
	I_{m,j,N}&=\frac{m}{2\pi i}\le(\sum_{|l|=1}^NU_{l,N}^{[j]}\oint_{B_l}\z^{m-1}d\z\ri)\le(1+O(N^{1/2}\delta)\ri),\\
	I_{m,j}&=2\int_{\G^+}v_j(\z)\Im \z^{m}d\l(\z), \quad j=2,3,\cdots,
\end{align}
respectively.
\ec

\subsection{Thermodynamic limit of the quasimomentum differentials $dp_N$   and related  questions}\label{sec-dens-of-diff} 
%\todo{F: We need to replace the title with a better one, something emphasize dp/dz?}
%quasi momentum and quasi energy differential $dp$,  $dq$
%and higher differentials $dq_j$, $j\in\N$ }\label{sec-dens-of-diff}

In this subsection we 
%use the already obtained  results of Section   \ref{sec-cons-dens} to 
express the { thermodynamic} limit $dp$ of   $dp_N$, as well as its ``antiderivative"  $2g_x(z)$, see \eqref{2gx-lim-int},
for fNLS soliton and breather gases
 in terms of the density of states $u(z)$.  
 We remind that in this subsection we  assume  that $\G^+$ is a contour and $d\l(z)=|dz|$.
We first show that $dp/dz$ is analytic in $\bar\C\setminus \G$ and find its  boundary behavior on $\G$. The obtained results allow us to express the 
%relative density of bandwidth (
spectral scaling function $\s(z)$ in terms of the average value and the jump of $dp/dz$  on $\Gamma$. Similar results are valid 
for quasienergy as well as for the higher fNLS flows.

%and corresponding densities of fluxes.

Given a  compact piece-wise smooth contour $\G^+\subset \C^+$  define the real
valued function $\th(\m)$ on $\G^+$ by $d\m=|d\m|e^{i\th(\m)}$, where $d\m$ is a differential in the positive direction on $\G^+$.  Taking into account the orientation of $\G$, the same relation on $\G^-$
is $d\m=-|d\m|e^{-i\th(\m)}$. Introducing a new function $\breve u=ue^{-i\th}$, we observe that
$u|d\m|=\breve u d\m$ on $\G^+$. This equation can be extended to $\G$ if 
$\breve u$ is Schwarz symmetrically continued on $\G^-$.

In the following Theorem \ref{the-lim-dpdz} we  express the thermodynamic limit  $\frac{dp}{dz}$ of the quasimomentum density $\frac{dp_N}{dz}$, defined as  $\frac{dp}{dz}=1-\sum_{m=1}^\infty\frac{I_m}{z^{m+1}}$
in  \eqref{def-dpdz}, in terms of the  Cauchy transform of $\breve u$.  Also,  the finite Hilbert transform (FHT) $H_\G$ in Theorem \ref{the-lim-dpdz} and in this paper is defined as 
\be H_\G[f](z)=\frac{1}{\pi}p.v.\int_\G\frac{f(w)dw}{w-z}
\ee
 on $\G$.

\bt\label{the-lim-dpdz}
Let $\G\in\C$ be a simple, compact, piece-wise smooth  Schwarz symmetrical contour and   
the density of states $u(z)$ is the solution of  \eqref{dr_soliton_gas1}.
%$\frac{dp}{dz}=1-\sum_{m=1}^\infty\frac{I_m}{z^{m+1}}$, see \eqref{def-dpdz}.
%\red{1-2g_{xz}}$, where $2g_x$ is defined by \eqref{ther-setup}.  
Then: (i) $\frac{dp}{dz}$ is analytic in $\bar\C\setminus\G$ and 
%\begin{enumerate}
	%\item 
	\be  \label{mer-diff0}
	\frac{dp}{dz}=1+2\pi C_{\G}[\breve u] \quad\text{ on $\bar\C\setminus \G$ },
	\ee
	where $C_{\G}$ denotes the Cauchy transform on $\G$, and;
	(ii) the jump $\D \frac{ dp}{dz}$ of $\frac{ dp}{dz}$ over $\G$ is {$2\pi \breve u$} whereas the average 
	$(\frac{ dp}{dz})_{av}:=\hf[(\frac{ dp}{dz})_+ + (\frac{ dp}{dz})_-]={1-\pi iH_\G[}\breve u]$.
 \et
%Note that the quasi-momentum differential $dp=[1+ 2g_{xz}(z)]dz$, so that
\begin{proof}
(i) The compactness of $\G^+$ implies that the series  \eqref{def-dpdz} for $\frac{ dp}{dz}$ has non zero radius of convergence. Then, according to \eqref{def-dpdz},

\bea\label{dp-hil-0}
\frac{ dp}{dz}= 1- 2 \int_{\G^+}u(\m)|d\m| \sum_{m=1}^\infty\frac{\Im\m^m}{z^{m+1}}=
1-\frac 1{i} \int_{\G}u(\m)|d\m| \sum_{m=1}^\infty\frac{\m^m}{z^{m+1}}\cr
= 1+\frac {2\pi}{2\pi i} \int_{\G}\frac{\breve u(\m)d\m}{\m-z} =1+2\pi C_{\G}[\breve u]
\eea
where we use $\m^m-\bar\m^m=2i\Im \m^m$ and  $u(z)$ has anti-Schwarz symmetric extension of in $\C^-$.
Taking into account the orientation of $\G$,
we have  $d\m=-|d\m|e^{-i\th(\m)}$   on $\G^-$ so that 
$\breve u$ is Schwarz symmetrically continued on $\G^-$.
Formula \eqref{dp-hil-0} is valid in $\bar\C\setminus \G$. Thus, we showed \eqref{mer-diff0}.

(ii) According to \eqref{mer-diff0}, $\D \frac{ dp}{dz}={  2\pi\breve u}$ on $\G$ and $(\frac{ dp}{dz})_{av}=1-i\pi H_\G[\breve u]$.

\end{proof}

As an immediate consequence of Theorem \ref{the-lim-dpdz}, we calculate
\begin{align}
\label{2gx-lim}
-2g_x(z)+z=\int_0^zdp=z+i\int_\G[\ln (\m-z)-\ln \m]u(\m)|d\m|=\cr
z+\int_\G u(\m)\arg\m|d\m|+i\int_\G\ln (\m-z)u(\m)|d\m|,
   \end{align}
which is valid for any $z\in \bar\C\setminus\G$.
As it was noted in the Introduction,
$2g_x(z)$  is analytic but, in general, multiple-valued on $\bar\C\setminus\G$; however, it is single valued on $\bar\C\setminus\g$,  where
$\g\in\C$ be a simple Schwarz symmetric curve that consisits of the superbands and  connecting them gaps.

 According to \eqref{2gx-lim}, $2g_x(\infty)= -\int_\G u(\m)\arg\m|d\m|$, so the very last term in \eqref{2gx-lim} represents the part of the Laurent expansion of $2g_x$ at infinity in the negative powers of $z$.

Considering the average $g_{x+}+g_{x-}$ of the boundary values, we obtain 
\begin{align}
\label{2gx-lim-pm}
g_{x+}(z)+g_{x-}(z)=
-2\int_{\G^+} u(\m)\arg\m|d\m|+i\int_{\G^+}\le[\ln \frac{\bar\m-z}{\m-z}+i\pi\chi_z(\m)\ri]u(\m)|d\m|,
   \end{align}
where $\chi_z(\m)$ is the indicator function of the arc $(z_\infty,z)$ of $\G^+$. Here $z_\infty$ denotes the beginning of the oriented curve $\G^+$. 

In the corollary below the ``carrier density of states" function 
$\tilde u(z)$ was defined in  \cite{ElTovbis} as a smooth Schwarz symmetrical interpolation of the carrier wavenumbers $\tilde k_j$ on $\G$, that satisfies equation (25), \cite{ElTovbis}, see also \eqref{dr_breather_gas1-gen}. Comparing \eqref{2gx-lim-pm} with \eqref{dr_breather_gas1-gen}, in which the limit $\d_0\ra 0$ is taken, and observing that $\Im g_x(z)$ is continuous on $\C$, we obtain  Corollary \ref{cor-gxpm-int} for the case of a soliton gas.

\bc
	In the conditions of Theorem \ref{the-lim-dpdz} we have 
	
	\be \label{sig-dp}
	\s(z)=  \frac{-2\pi\Im \int_0^z \le(\frac{dp}{dz}\ri)_{av}dz}{\D \frac{dp}{dz}(z)}\cdot e^{-i\th(z)},
	\ee
	where $\s(z)$ from \eqref{dr_soliton_gas1} is the relative density of bandwidth. 
	\ec
	\begin{proof}
	
	Indeed, 
	WLOG, we assume $0\in\G$. Integrating the latter equation along $\G$, we obtain
\be
\int_0^z(\frac{ dp}{dz})_{av} dz=z + i\int_\G\ln(w-z)u(w)|dw|-i\int_\G\ln(w)u(w)|dw|,
\ee
where we changed the order of integration in $H_\G[\breve u]$. Taking the imaginary part in the latter equation together
with \eqref{dr_soliton_gas1}
yields \eqref{sig-dp}.
	
\end{proof}

To extend the obtained above result from  soliton to  breather fNLS  gases, we observe that, according to 
  \eqref{G-I relation}  and \eqref{inv-breather-t1},
\be\label{ther-setup}
2g_x(z)= \sum_{m=0}^\infty\frac{G_m}{z^m},\qquad \text{where} \quad   
G_m=-2\int_{\G^+}u(\z)\Im{F_m(\z)}d\l(\z)-d_m\d_0^{m+1}
%R\qquad \text{where} \quad   
\ee
in the case of a breather gas and 
 \begin{align}
 	G_m=-\frac{2}{m}\int_{\G^+}u(\z)\Im \z^m d\l(\z),
 \end{align}
 in the case of a soliton gas, see \eqref{G-I relation} and \eqref{inv-sol}.

The following Theorem \ref{the-lim-dpdz-breath} provides the expression for $2g_x(z)$ in the case of an fNLS breather gas. As a consequence, we obtain a proof of Corollary \ref{cor-gxpm-int}.

%\sout{Let us now calculate $2g_x(z)$ for the fNLS breather gas and confirm that an analog of \eqref{su-sol} holds for the breather gas as well.}

\br
Under the assumptions of Theorem \ref{theo-ImN} part (ii),  it follows from \eqref{2gx-inf} that in the thermodynamic limit of fNLS breather gas, we have
\begin{align}\label{2gx-inf-brea}
	2g_x(\infty)=-2\int_{\G^+}u(\z)\arg{\le(\z+\sqrt{\z^2+\delta_0^2}\ri)}d\l(\z),
\end{align}
and in the soliton gas case, we have
\begin{align}\label{2gx-inf-sol}
	2g_x(\infty)=-2\int_{\G^+}u(\z)\arg{\z}d \l(\z).
\end{align}
\er

\bt\label{the-lim-dpdz-breath}
Let $\G\in\C$ be a simple, compact, piece-wise smooth  Schwarz symmetrical contour,  $\g_0=[-i\d_0,i\d_0]$  be the permanent band,
the density of states $u(z)$  solves  \eqref{dr_breather_gas1}
and  $2g_x$ is defined by $2g_x/R_0=\sum_{m=0}^\infty\frac{\tilde G_m}{z^{m+1}}$ with $\tilde G_m$  given by \eqref{tildeGm}. Then
\begin{align}\label{2gx-brea}
2g_x(z)=i\int_\G u(\z)\ln\frac{R_0(\z)R_0(z)+\z z-\d_0^2}{\z-z} |d\z|+z-R_0(z),   
\end{align}
where $z\in \bar\C\setminus (\G\cup \g_0)$.
%including the shores of $\G$. 
Moreover, substituting $z=\infty$ in \eqref{2gx-brea} we obtain $2g_x(\infty)$ given by \eqref{2gx-inf-brea}.
 \et
\begin{proof}
 Using \eqref{tildeGm}, we start obtain
 \begin{multline}\label{2gx/R}
     \frac{2g_x}{R_0}=\sum_{m=0}^\infty\frac{\tilde G_m}{z^{m+1}}=
      \sum_{m=0}^\infty
     \frac {-1}{z^{m+1}}\le[\int_{\G^+} 2u(\z)\Im\int_0^\z  \frac{\m^m d\m}{R_0(\m)} |d\z| +\frac{1}{2\pi i}\oint_{A_0}\frac{\z^{m+1}d\z}{R_0(\z)}\ri]\cr
     =i\int_{\G} u(\z)\int_0^\z  \frac{\sum_{m=0}^\infty\frac{\m^m}{z^{m+1}} d\m}{R_0(\m)} |d\z|-
     \frac{1}{2\pi i}\oint_{A_0}\frac{\sum_{m=0}^\infty\frac{\m^m}{z^{m+1}}d\z}{R_0(\z)}=\cr
     -i\int_{\G} u(\z)\int_0^\z  \frac{d\m}{(\m-z)R_0(\m)}|d\z|+
     \frac{1}{2\pi i}\oint_{A_0}\frac{\z d\z}{(\z-z)R_0(\z)}=\cr
     \frac{z}{R_0(z)}-1 +
     \frac{i}{R_0(z)}\int_{\G} u(\z)\ln\frac{R_0(\z)R_0(z)+\z z-\d_0^2}{\z-z} |d\z|,
\end{multline}
 where we used the anti-derivative \eqref{anti-der} and the anti Schwarz symmetry of $u$ in the latter transformation.  Multiplying \eqref{2gx/R} by $R_0$ yields \eqref{2gx-brea}.  Note that $z-R_0(z)$ is an odd function near $z=\infty$ and therefore does not contribute to $2g_x(\infty)$. Therefore, to recover \eqref{2gx-inf-brea},  one needs to divide both the numerator and the denominator of the logarithm in \eqref{2gx/R} by $z$ and use
 the anti Schwarz symmetry of $u$. 
 \end{proof}
 It is straightforward to check that in the limit $\d_0\ra 0$ equation
 \eqref{2gx-brea} for the breather gas turns into \eqref{2gx-lim} for the  soliton gas. We also note that $2g_x(z)$ is analytic in $\bar\C\setminus(\G\cup\g_0)$ and $\Im g_x(z)$ is continuous in $\C\setminus \g_0$.
 
 As a consequence of Theorem \ref{the-lim-dpdz-breath}, we obtain  Corollary \ref{cor-gxpm-int} for the case of a breather gas.

\br
Define $dp/dz=1-\sum_{m=1}^\infty I_m/z^{m+1}$ with $I_m$ given by equation \eqref{inv-soliton-t1}, use the result of Theorem \ref{the-lim-dpdz-breath} and assume that $I_m=-mG_m$ (which can be regarded as the limit of \eqref{G-I relation}). Then we have
\begin{align}\label{dpdz-brea}
	\frac{dp}{dz}=1-2g_{xz}=1+2\pi C_\G[\breve{u}]+\left(\frac{z}{R_0(z)}-1-i\int_{\G}\frac{u(\z)|d\z|}{z+\le(\frac{R_0(\z)-R_0(z)}{R_0(\z)z-R_0(z)\z}\ri)\delta_0}\right).
\end{align}
Comparing  \eqref{dpdz-brea} with $dp/dz$ for the soliton gas from Theorem \ref{the-lim-dpdz}, one finds that  they differ by a ``breather correction" term in the round brackets in \eqref{dpdz-brea}. 

It is easy to see that  in the limit  $\delta_0\ra 0^+$  this correction term becomes $0$ and $dp/dz$ for the breather gas reduces to that for the soliton gas.
\er

\br \label{remark-high-mer}
Similar to Theorem \ref{the-lim-dpdz} results can be obtained 
 for the meromorphic differentials $dq_j$ (defined as the thermodynamic limit of $dq_{j,N}$ given by equation \eqref{gen-diffs}), $j=2,3,\cdots$, where we replace $u(z)$ by the corresponding $v_j(z)$ and
$1$ by $jT_jz^{j-1}$ .
In particular,
\be \label{mer-diff}
\frac{dq_{j}}{dz}-jT_jz^{j-1}=2\pi C_{\G}[\breve v_j], \quad j=1,2, ...,
\ee
where $j=1$ corresponds to  equation \eqref{mer-diff0}. Obviously, $\frac{dq_j}{dz}$ is analytic on $\C\setminus(\G\cup\g_0)$ and the 
jump of $\frac{dq_j}{dz}$ over $\G$ is ${2\pi \breve v_j}$. The average of boundary values of  $\frac{dq_j}{dz}$ on $\G$ is
$1{-i\pi H_\G[\breve v_j]}$.
\er

Suppose now that $\G\subset i\R$, i.e., we have a bound state gas. Then $\breve v_j=-iv_j$, so  \eqref{mer-diff} becomes
\be \label{mer-diffA}
\frac{dq_{j}}{dz}-jT_jz^{j-1}=-2i\pi C_{\G}[v_j].
\ee
Thus the jump of  $\frac{dq_{j}}{dz}$ on $\G$ is {$-2i \pi v_j$} and its  average value on $\G$ is  $1{-\pi H_{\G}[v_j]}$.

\section{ Periodic gases}\label{sec-per}

By periodic soliton or breather fNLS gases we understand gases whose spectral characteristics $\G^+$, $\varphi(  z)$ and 
$\nu(  z)$ can be generated as the   semiclassical limit  of the direct fNLS spectral problem with periodic 
potentials. Such problem was considered in the recent paper \cite{BioOre2020}, where 
the authors considered the (formal) semiclassical limit of the  direct fNLS spectral problem with the real,  even, continuous
single lobe potential $q(x)$ with a period $2L>0$, that is, $q(x+2L)=q(x)$. Following \cite{BioOre2020}, WLOG, we assume
$\max q(x)=q(0)=M$, $\min q(x)=q(\pm L)=m>0$ and $q(x)$ is monotonically decreasing on $[0,L]$.
Then the semiclassical ($\e\ra 0^+$)
limit of  the Floquet discriminant (trace of the monodromy matrix) is given by
\be\label{trace_Mon}
w(\l)= 2\cos\frac{S_1(\l)}{\e}\cosh \frac{S_2(\l)}{\e}
\ee
where  $  z\in[im,iM]$ is the original (Zakharov-Shabat) spectral variable,  $\l:=-  z^2 \in[m^2,M^2]$,
\be\label{S12}
S_1(\l)=\int_{-p(\l)}^{p(\l)}\sqrt{[q^2(x)-\l]}dx,~~~~~~\qquad S_2(\l)=\int_{p(\l)}^{L}\sqrt{|q^2(x)-\l|}dx,
\ee
and $x=\pm p(\l),\, |x|\leq L   \Leftrightarrow  q^2(x)=\l$. 

The Lax spectrum (the bands) of the spectral problem are defined by the requirements  $\Im \D(  z)=0$,
$|\Re \D(  z)|\leq 2$, where $\D(  z)=w(\l)$. It is well known that $ \D(  z)$ 
 is Schwarz symmetrical and it was argued in \cite{BioOre2020} that, in the semiclassical limit, the Lax 
 spectrum consists of a ``cross'' $\R\cup [-im,im]$, which represents a single  band,  combined with the compact
 \be\label{per-Gam}
 \G^+=[im,iM] \quad\text{ and its Schwarz symmetrical image $\G^-=\bar\G^+$}
\ee
where additional $\e$-scaled bands are accumulating as $\e\ra 0$. 
Equation \eqref{trace_Mon} implies  
that the centers of bands $\l_n\in\G$, $\G=\G^+\cup\G^-$, are given by $\cos\frac{S_1(\l)}{\e}=0$
or
\be
S_1(\l_n)=2\int_{0}^{p(\l_n)}\sqrt{|q^2(x)-\l_n|}dx=\pi \e(n+\hf), ~~~~~~~~~~~{\rm where }~~~ n\in\N.
\ee

Then the total number of bands $N$ is the integer part  of 
\be\label{N-form}
N= {\rm Int~part}\le\{\frac 2{\pi \e}\int_0^L\sqrt{q^2(x)-m^2}dx -\hf\ri\}.
\ee

By definition, the density 
\be\label{density}
 \varphi(\l)=\lim_{\D\ra 0}\lim_{N\ra \infty}\frac {\# \text{of $\l_n$ in $\D$ nbhd of $\l$}}{2N\D}.
\ee
Since
\bea
\lim_{N\ra \infty}\frac {\# \text{of $\l_n$ in $[\l_2,\l_1]$ }}{N(\l_1-\l_2)}=
\frac{\int_{0}^{p(\l_2)}\sqrt{q^2(x)-\l_2}dx      -  \int_{0}^{p(\l_1)}\sqrt{q^2(x)-\l_1}dx}{(\l_1-\l_2)\int_0^L\sqrt{q^2(x)-m^2}dx }\cr
= \frac{\int_{0}^{p(\l_1)}\frac{(\l_1-\l_2)dx} {\sqrt{q^2(x)-\l_2}+ \sqrt{q^2(x)-\l_1}}    +  \int_{p(\l_1)}^{p(\l_2)}\sqrt{q^2(x)-\l_2}dx}
{(\l_1-\l_2)\int_0^L\sqrt{q^2(x)-m^2}dx },
\eea
we obtain     
\be\label{phi_lambda}
 \varphi(\l)=\frac{\int_{0}^{p(\l)}\frac{dx} {\sqrt{q^2(x)-\l}}}{2\int_0^L\sqrt{q^2(x)-m^2}dx },
\ee
provided that $p(\l)$ is differentiable or at least H\"{o}lder class with the exponent $\a>\hf$.
%where $p(\l)$ can be replaced by $q^{-1}(\l^\hf)$.
Transition from $ \varphi(\l)$ to  $ \varphi(  z)$ yields the density of bands function
\be\label{phi_eta}
 \varphi(  z)=\frac{|  z|\int_{0}^{q^{-1}(|  z|)}\frac{dx} {\sqrt{q^2(x)+  z^2}}}{\int_0^L\sqrt{q^2(x)-m^2}dx }.
\ee

Note that the numerator is almost identical to the density (A.6) from \cite{Ven1989}.

Finally, we consider the scaled bandwidth function $\nu(  z)$ (with a slight abuse of notation we sometimes write $\nu(\l)$).
It is asymptotically defined by
\be
\left.\frac{dw}{d\l}\ri|_{\l=\l_n} \D\l=4
\ee
where $\D\l$ is the bandwidth. Then
\be
\D\l= \frac {2\e S_1'(\l_n)}{\cosh  \frac{S_2(\l)}{\e}},
\ee
so that, according to \eqref{N-form}, 
\be\label{per-nu-eta}
\nu(  z)=\frac{\pi S_2(\l)}{2\int_0^L\sqrt{q^2(x)-m^2}dx} =
\frac{\pi \int_{q^{-1}(|  z|)}^{L}\sqrt{|q^2(x)+  z^2|}dx}{2\int_0^L\sqrt{q^2(x)-m^2}dx}.
\ee
This formula is near identical to (A.7) from \cite{Ven1989}.
Now we can easily express the ``relative density of bandwidth'' function  
%in the notations of \cite{ET2019} we have 
\be \label{sig-eta}
\s(  z)=\frac{2\nu(  z)}{\varphi(  z)}=
\frac{\pi \int_{q^{-1}(|  z|)}^{L}\sqrt{|q^2(x)+  z^2|}dx}{|  z|\int_{0}^{q^{-1}(|  z|)}\frac{dx} {\sqrt{q^2(x)+  z^2}}},
\ee
which, together with the compact set $\G^+$, determines the NDR for soliton and breather gases.

Before finishing this subsection we want to emphasize that for us  the asymptotic formula \eqref{trace_Mon} is a motivation 
to introduce a soliton (or a breather) gas with  $\G^+$, $\varphi(  z)$ and 
$\nu(  z)$  ``parametrized'' by $q(x)$ according to \eqref{per-Gam}, \eqref{phi_eta} and \eqref{per-nu-eta} respectively.
Moreover, as it will be shown below, the case of $m=0$ corresponds to the soliton gas on $\G=[-iM,iM]$, whereas 
the case of $m>0$ corresponds to the breather  gas on $\G=\G^+\cup\G^-$ given by \eqref{per-Gam} with additional ``stationary'' band
on $[-im,im]$. In this approach, the requirements on the ``parametrizing'' function $q(x)$ can be relaxed, in particular, $q(x)$
can have finitely many jump discontinuities on the period. We  also want to mention that in the rest of the section we consider the odd continuation of $\varphi(z)$ from $\G^+$ into $\G$, so the $|z|$ in \eqref{phi_eta} can be replaced by $\Im z$.

\subsection{Density of  states $u(  z)$ for the periodic soliton gas}\label{sec-per-sol}

In this subsection we will solve the NDR equation for the density of states $u(  z)$ in the particular case  of $q(L)=m=0$. 
As it turns out, $u(  z)$ is proportional to $\varphi(  z)$ given by \eqref{phi_eta}. 
%\red{Re-writing  
The NDR equation  \eqref{dr_soliton_gas1} for the density of states $u(z)$ can be written as 
\be\label{int-sol-bound}
-\int_{-iM}^{iM} \ln |\m-  z|r(\m)\varphi(\m)|d\m| + 2\nu(  z)r(  z) = -i  z, 
\ee
where $ r(  z)= \frac{ u(  z)}{\varphi(  z)}$ and  $  z\in [-iM,iM]$. we will show that \eqref{int-sol-bound} 
is satisfied by $r=\rm{const}$.

\bt\label{the-per-sol-gas}
If $\varphi(  z)$ and $\nu(  z)$ are given by \eqref{phi_eta} and \eqref{per-nu-eta} where $q(x)$ is monotonically 
decreasing on $[0,L]$ and
$q(L)=0$ then equation \eqref{int-sol-bound} 
%is true on $[-iM,iM]$ and 
has a constant solution 
\be\label{per-wavedens}
r=\frac{1}{\pi L}\int_0^L q(x) dx,
\ee
so that the density of states 
\be\label{per-u}
u(  z)= {\frac{|  z|}{ \pi L}}\int_{0}^{q^{-1}(|  z|)}\frac{dx} {\sqrt{q^2(x)+  z^2}},\quad   z\in\G^+.
\ee
\et

\begin{proof}
As it was proven in \cite{KT2021}, there exists a unique solution for the integral equation \eqref{int-sol-bound}. 
Assume  that $r(  z) = r$, where $r>0$ is a constant.
If we can find $v$ satisfying \eqref{int-sol-bound}, we will prove the theorem.

Changing variables $\m=iy, \,   z=i\x$  in \eqref{int-sol-bound} and then differentiating  in $\x$, we obtain 
%(with a slight abuse of notations)
\be\label{per-diff-int}
\int_{-M}^M\frac{\varphi(iy)dy}{y-\x}+2\frac{d}{d\x}\nu(i\x)=1/r.
 \ee
 Using \eqref{per-nu-eta}, \eqref{phi_eta} we obtain
 \be\label{f-nu}
 2\frac{d}{d\x}\nu(i\x)= \frac{\pi \x \int_{q^{-1}(\x)}^{L}\frac{dx} {\sqrt{\x^2-q^2(x)}}}{\int_0^L q(x) dx}
 \ee
 and
 \bea \label{f-varphi}
 \int_0^L q(x) dx\int_{-M}^M\frac{\varphi(iy)dy}{y-\x}=  \int_{-M}^M dy\int_0^{q^{-1}(y)}\frac{dx} {\sqrt{q^2(x)-y^2}}+\cr
 \x\int_{-M}^M\frac{dy}{y-\x}\int_0^{q^{-1}(y)}\frac{dx} {\sqrt{q^2(x)-y^2}}
 \eea
 Changing the limit of integration in the second integral $\Iscr$, we obtain
 \be
 \Iscr=\int_0^Ldx\int^{q(x)}_{-q(x)}\frac{dy}{(y-\x)\sqrt{q^2(x)-y^2}}.
 \ee
 Since the inner integral is zero when $|\x|< q(x)$ and $\frac{i\pi}{\sqrt{q^2(x)-\x^2}}$ otherwise, we obtain
 \be 
 \Iscr=\int_{q^{-1}(\x)}^{L}\frac{i\pi dx} {\sqrt{q^2(x)-\x^2}}=-\pi\int_{q^{-1}(\x)}^{L}\frac{ dx} {\sqrt{\x^2-q^2(x)}},
 \ee
 where one has to consider the proper branch of $\sqrt{\x^2-q^2(x)}$ to obtain the correct sign. 
 We complete calculating \eqref{f-varphi}  by observing
 \be
\int_{-M}^M dy\int_0^{q^{-1}(y)}\frac{dx} {\sqrt{q^2(x)-y^2}}=\pi L.
 \ee
 Substituting now \eqref{f-nu}, \eqref{f-varphi} into \eqref{per-diff-int}, we obtain
 \be
 r=\frac{1}{\pi L}\int_0^L q(x) dx.\nonumber
 \ee
 \end{proof}

Consider the family of even potentials  $q_k(x)$ 
with the period $kL$,  $k\geq 1$, generated by $q(x)$, where $q_k(x)\equiv q(x)$ on $[0,L]$ and $q_k(x)\equiv 0$ on $[L,kL]$.
We can extend Theorem \ref{the-per-sol-gas} from $q(x)=q_1(x)$ to $q_k(x)$ by considering small deformations $\hat q$
of $q_k$ on $[L-\e,kL]$ so that $\hat q_k$ is monotonically decreasing and  $\hat q_k(kL)=0$. Then  Theorem \ref{the-per-sol-gas}
is valid for $\hat q_k$. Thus, in the small deformation limit (for a fixed $k>0$) we obtain the following result.

\bc\label{cor-per-sol-gas}
For described above periodic  potentials $q_k(x)$ with the period $kL$, $k\geq 1$, formulae \eqref{per-wavedens}, \eqref{per-u}
become
\be\label{per-k}
r_k=\frac{1}{\pi kL}\int_0^L q(x) dx, \qquad
u_k(  z)= \frac{|  z|}{k\pi L}\int_{0}^{q^{-1}(|  z|)}\frac{dx} {\sqrt{q^2(x)+  z^2}}
\ee
respectively.
\ec

Another way to prove Corollary \ref{cor-per-sol-gas} is to repeat the steps of Theorem \ref{the-per-sol-gas}, taking into account that
%in the case  $m=0$ 
the density $\varphi(  z)$ does not depend on $k$ and 
\be\label{per-nuk}
\nu_k(  z)=\nu(  z)+   \frac{\pi(k-1)L|  z|}{2\int_0^Lq(x)dx}.     
\ee

\subsection{ Density of  states  for periodic breather gas}\label{sec-per-breather}

The results of the previous Section \ref{sec-per-sol} do not work in the case when $m>0$ and the bands are located on the interval $\G^+=[im, iM]$
and its Schwarz symmetrical $\G^-$. 
%Technically, there are two ways to extend these results for the case $m>0$: a) to make an infinitesimal ``dip''  to zero at the end of the period (put $q(L)=0$ instead of $q(L)=m$), thus filling the whole
% interval $[-iM,iM]$ with     bands (note that $\nu(  z)$ will be infinitesimally small on $[-im,im]$), or; b) 
Since $[-im,im]$ is a single band of the Lax spectrum of $q(x)$ (\cite{BioOre2020}), it makes sense to assume that the case $m>0$ corresponds to
the breather gas. Indeed, the following Theorem \ref{the-per-breath-gas} shows that  a constant $r(  z)=r$ 
%given by \eqref{per-wavedens}
satisfies the NDR \eqref{dr_breather_gas1} for the breather gas written as
%equations instead  of the soliton gas ones. The second approach seems more natural, especially given that ZS equation has a continuous spectrum on $[-im,im]$, \cite{Biondini},therefore we will pursue it below.
%
\be\label{int-breath-bound}
\Re\int_\G 
\ln \le(\frac{R_0(  z)R_0(\m)+  z\m+m^2}
{\m-  z}\ri)
r(\m)\varphi(\m)|d\m| + 2\nu(  z)r(  z) = -i R_0(  z), 
\ee
where $  z\in\G$,  $\G= [-iM,iM]\setminus [-im,im]$ and  $R_0(  z) =\sqrt{  z^2+m^2}$.
 
 \bt\label{the-per-breath-gas}
If $\varphi(  z)$ and $\nu(  z)$ are given by \eqref{phi_eta} and \eqref{per-nu-eta} where $q(x)$ is monotonically decreasing on $[0,L]$ and  $q(L)=m>0$ then the integral equation \eqref{int-breath-bound} has a constant solution
\be\label{per-wavedens-brea}
r= \frac{1}{\pi L}\int_0^L\sqrt{q^2(x)-m^2}dx,
\ee
so that the corresponding  density of states 
\be\label{per-u-br}
u(  z)= {\frac{|  z|}{\pi  L}}\int_{0}^{q^{-1}(|  z|)}\frac{dx} {\sqrt{q^2(x)+  z^2}}, \quad   z\in\G^+,
\ee
is given by the same expression as in the soliton gas case, see \eqref{per-u}.
%in the case the period is $2kL$ and we continue  $q(x)$ by zero for $x\in[L,kL}$, where $k\geq 1$.
\et
 
\begin{proof}
 Using \eqref{anti-der}, we obtain
 \be\label{breather-int}
 R_0(  z)\int_{\pm im}^\m\frac {d\z}{R_0(\z)(\z-  z)} -\ln(\pm im)=
 - \ln \le(\frac{R_0(  z)R_0(\m)+  z\m+m^2}
{  z-\m}\ri)
 \ee
 %\red{[Some correction by a constant term?]}
 where $\pm \Im \m>0$ respectively. Note that the term $-\ln(\pm im)=\mp\frac{i\pi}2-\ln m$  can be ignored when substituting the left hand side
 in   the integral in \eqref{int-breath-bound}  because: i) we 
need only the real part of this integral and so the $\mp\frac{i\pi}2$ term should be ignored; ii) 
{$u=r\varphi$} is  an odd function and so the integral of  $-(\ln m) u(\m)$ is zero. We now replace 
the logarithmic term  in \eqref{int-breath-bound} by the remaining (first) term  of  the left hand side of \eqref{breather-int}.
Now, integrating by parts the obtained integral in \eqref{int-breath-bound},
 we get
 \be\label{breather-int-2}
-\Re\le[rR_0(  z)\int_\G \frac {S_1(\m)d\m}{R_0(\m)(\m-  z)}\ri]
 + 2\pi r S_2(  z) = -iD R_0(  z),
 \ee
 where $D=\int_0^L\sqrt{q^2(x)-m^2}dx$.
 A few words to explain \eqref{breather-int-2}.
First,   the antiderivative of $2D\varphi$ is $-iS_1$, see \eqref{S12} and \eqref{phi_eta}. 
 %and \eqref{per-wavedens-brea}. 
 Second, 
the secular (not integral) terms that appear in integration by parts  become zero.
Indeed, it is obvious that $S_1(\pm iM)=0$ as well as the integral in \eqref{breather-int}, evaluated at $\m=\pm im$, is zero.
 Finally, $|d\m|=-id\m$ explains the sign of the integral term in \eqref{breather-int-2}.
 %\todo{Up to here, I agree with the argument here.}
 Note that the radical  $\sqrt{q^2(x)+\m^2}$ in $S_1$ must be positive along the contour of integration, i.e., on the left
 (positive) shore of $\G$,  which corresponds to the  branch $\sqrt{q^2(x)+\m^2}\ra -\m$ as $\m\ra\infty$.
 %\todo{F: Is this correct?}
 
 Similarly to Theorem \ref{the-per-sol-gas},  substituting \eqref{S12} into \eqref{breather-int-2} and changing the order of integration, we obtain
 \bea\label{I-per}
 \Iscr=-\int_\G \frac {S_1(\m)d\m}{R_0(\m)(\m-  z)}=2\int_\G
 \frac {\int_0^{q^{-1}(-i\m)}\sqrt{q^2(x)+\m^2} dx}{R_0(\m)(\m-  z)}d\m\cr
= 2\int_0^Ldx\le(\int_{im}^{iq(x)}+ \int_{-iq(x)}^{-im}\ri) \sqrt{\frac{\m^2+q^2(x)}
 {\m^2+m^2}}\frac{d\m}{\m-  z}=\cr
 -2\pi i L+ \frac{2\pi i}{R_0(  z)}\int_0^L\sqrt{q^2(x)+  z^2}\chi_x(  z) dx 
 \eea
 where $\chi_x$ is the characteristic function of the union of the segment $[iq(x), iM]$  with its complex conjugate.
 In \eqref{I-per} we use the standard ($\lim_{\m\ra\infty}\sqrt{q^2(x)+\m^2}= \m$) branch of the radical and thus the
 sign changes after the second equality.

 For $  z\in \C^+$, the latter term in \eqref{I-per} becomes
 \be\label{I-per-breath}
 - \frac{2\pi }{R_0(  z)}\int_{q^{-1}(-i  z)}^L\sqrt{|q^2(x)+  z^2|} dx.
 \ee
 Substituting \eqref{I-per}, \eqref{I-per-breath} into \eqref{breather-int-2} complete the proof of the theorem for $  z\in\G^+$. %\eqref{int-breath-bound} and \eqref{per-wavedens-brea}.
 The case of $  z\in \G^-$ follows from symmetry considerations.
 \end{proof}
 
 \br
 An analog of Corollary \ref{cor-per-sol-gas} is valid for the periodic breather gas with expression for $r_k$
 in \eqref{per-k} being replaced by
 \be\label{per-k-br}
r_k=\frac{1}{\pi kL}\int_0^L \sqrt{q^2(x)-m^2} dx.
\ee
Moreover, the analog of \eqref{per-nuk} for the breather gas is
\be\label{per-nuk-br}
\nu_k(  z)=\nu(  z)+   \frac{\pi(k-1)L|  z|}{2\int_0^L\sqrt{q^2(x)-m^2}dx}.     
\ee
\er

 \subsection{Conserved   densities for periodic gases}

In this subsection, we compute the averaged densities $I_m$ for periodic  gases. We use results from Subsection \eqref{sec-per-sol} and \eqref{sec-per-breather} and from Section \eqref{sec-cons-dens} to derive some formula for computing the averaged densities $I_m$. Based on those formulae, we study the relation between $g$-function with the density of states $u$ in the periodic gases situation.

\bt\label{the-per-I-gx}
If $\varphi(  z)$ and $\nu(  z)$ are given by \eqref{phi_eta} and \eqref{per-nu-eta} where $q(x)$ is monotonically decreasing on $[0,L]$ and  $q(L)=m>0$ , then for any odd $k\in \N$,
\begin{align}
    I_k=\frac{(-1)^{\frac{k+1}{2}}kd_k}{L}\int_0^L q^{k+1}(x)dx,\label{Ik-per}
\end{align}
and $I_k=0$ for any even $k\in \N$,
where $d_k$ is defined in \eqref{dm-int}.

Moreover,
\begin{align}
    2g_x(z)=\frac{1}{L}\int_0^L\le(z-\sqrt{z^2+q^2(x)}\ri)dx+\frac{1}{L}\int_0^L\sqrt{q^2(x)-m^2}dx,\quad z\in \bar\C\backslash \G.\label{2gx-per}
\end{align}
\et

\begin{proof}
 Using the result of Theorem \ref{the-per-breath-gas}, formula \eqref{per-u-br} and equation \eqref{Im-Fm}, we have, for any odd $k\in\N$,
 \begin{align}
      I_k&=\frac{2ik}{\pi L}\int_{im}^{iM} \int_{0}^{q^{-1}(|\z|)}\frac{\z}{\sqrt{q^2(x)+\z^2}}\le(\frac{1}{k}\z^k+d_k(im)^{k+1}\z^{-1}+O(\z^{-3})\ri)dxd\z+kd_k(im)^{k+1}\nonumber \\
      &=\frac{2ik}{\pi L}\int_{0}^{L} \int_{im}^{iq(x)}\frac{\z}{\sqrt{q^2(x)+\z^2}}\le(\frac{1}{k}\z^k+d_k(im)^{k+1}\z^{-1}+O(\z^{-3})\ri)d\z dx+kd_k(im)^{k+1}\nonumber \\
      &=\frac{2i}{\pi L}\int_0^L\frac{2\pi i}{4}\Res_{\z=\infty}\le((\z^{k+1}+kd_k(im)^{k+1}))((q^2(x)+\z^2)^{-1/2})\ri) dx+kd_k(im)^{k+1}\nonumber \\
      &=\frac{i^{k+3}}{L}\int_0^{L}\le(-kd_kq^{k+1}(x)+kd_km^{k+1}\ri)dx+i^{k+1}kd_km^{k+1}=\frac{i^{k+1}kd_k}{L}\int_0^L q^{k+1}(x)dx.\label{Ik-per-proof}
 \end{align}
 While when $k$ is even, due to \eqref{Im-Fm}, $I_k=0$.
 
 Then, by definition, we have
 \begin{align}
   {2g_{xz}=\sum_{k=1}^\infty} \frac{kd_{k}}{L}\int_0^L(iq)^{k+1}(x)z^{-(k+1)}dx={1-\frac{z}{L}}\int_0^L\frac{dx}{\sqrt{z^2+q^2(x)}},\quad z\in \bar\C\backslash \G,
 \end{align}
 and the jump of $2g_{xz}$ on $\G^+$ is
 \begin{align}
     {2g_{xz+}-2g_{xz+}=-\frac{2i |z|}{L}\int_0^{q^{-1}(-iz)}\frac{dx}{\sqrt{z^2+q^2(x)}}=-2i \pi u(z), \quad z\in \G^+.}
 \end{align}
 Applying the anti-symmetric property of $u$, the jump on $\G^-$ can be derived similarly.

Now, by integrating $2g_{xz}$, and taking into account of the boundary behavior \eqref{2gx-inf-brea}, we obtain
\begin{align*}
    2g_x(z)&={z-\frac{1}{L}\int_0^L\sqrt{z^2+q^2(x)}dx}-\pi \int_{\G^+}u(\z)|d\z|\cr 
    &={z-\frac{1}{L}\int_0^L\sqrt{z^2+q^2(x)}dx}+\frac{1}{L}\int_0^L\sqrt{q^2(x)-m^2}dx,\quad z\in \bar\C\backslash \G,
\end{align*}
where we used the fact that $\arg(\z+\sqrt{\z^2+m^2})=\frac{\pi}{2}$ for $\z\in \G^+$.
\end{proof}

\br 
According to Corollary \ref{cor-gxpm-int}, for periodic breather/soliton gas, we have
\begin{align}
    {\sigma(z)u(z)=\frac{1}{L}\int_{q^{-1}(|z|)}^L\sqrt{|z^2+q^2(x)|}dx,\quad \tilde{u}(z)=\frac{1}{L}\int_0^L\sqrt{q^2(x)-m^2}dx.}
\end{align}
These equations are compatible with equations \eqref{sig-eta} and \eqref{per-u-br}.

\er

\subsection{Examples of periodic soliton and breather gases}\label{sec-per-examp}

This subsection consider examples of periodic soliton and breather gases for some choices of $q(x)$. We calculate the corresponding density of bands $\varphi$, scaled bandwidth $\nu$ and the relative scaled bandwidth $\s=\frac{2\nu}{\varphi}$. In some cases we also calculate the corresponding $\nu_k$ and $\s_k$ of the $k$-dilution of the gas. As above, by $k$-dilution we mean the  
%Equation \eqref{phi_eta} implies that if we 
increase the period from $2L$ to $2kL$ with $k\geq 1$ while
keeping $q(x)=m$ for $L\leq x\leq kL$.
%the density $\varphi(z)$ would stay the same.
At the same time, the scaled bandwidth $\nu(z)$ given by \eqref{per-nu-eta} will grow linearly in $k$ for large $k$.
The limit  $k\ra \infty$ corresponds to the semiclassical limit of the decaying potential, provided $m=0$, otherwise, we get 
a potential with a constant background (leading to the breather gas). 
Therefore, equation \eqref{phi_eta} can be used to calculate the semiclassical spectral density for potentials 
with  constant background, whereas \eqref{per-nuk}, \eqref{per-nuk-br} show that in the limit $k\ra \infty$  stationary semiclassical periodic gas is approaching 
super exponential (ideal gas) limit.

Let us calculate some examples.

{\it Example 1.} Consider the box potential with the width $2L$, hight $Q$ and period $2kL$, $k\geq 1$.
Then $M=Q$, $m=0$ and $p(\l)=L$. 
The latter formula can be justified by considering $k$-dilution of the gas.
Thus $S_1(\l)=2L\sqrt{Q^2-\l}$, $S_1(m)=2LQ$ and, by   \eqref{phi_eta},  \eqref{per-nu-eta}
\be\label{per-box}
\varphi(z)=\frac{|z|}{Q\sqrt{Q^2+z^2}},~~~~~~~\quad\nu_k(z)= \frac{\pi(k-1)|z|}{2Q}.
\ee
Note that  $k=1$ corresponds to the condensate $q(x)\equiv Q$ that according to Theorem \ref{the-per-sol-gas},
has DOS ${u(z)=r\varphi(z)=\frac{|z|}{\pi \sqrt{Q^2+z^2}}}$, which is a well known DOS for the soliton condensate on $\G^+=[0,iQ]$.
In the case of $k>1$, we have
\be
\s_k(z)=\pi (k-1)\sqrt{Q^2+z^2},
\ee
i.e., this is exactly the same $\s(z)$ that was obtained in \cite{ElTovbis} {when $r=\frac{Q}{\pi}$ is replaced by $r_k=\frac{Q}{k\pi} <r$, and the $k$-deluted DOS $u_k(z)=u(z)/k.$}

Consider now situation when we fix some $m\in(0,Q)$ and consider the $k$-dilution of the corresponding breather gas. Then
\be\label{box-br-dil}
\nu_k(z)= \frac{\pi(k-1)\sqrt{|z|^2-m^2}}{2\sqrt{Q^2-m^2}},
\quad
\s_k(z)=\pi (k-1)\sqrt{(1-\frac{m^2}{|z|^2})(Q^2+z^2)}.
\ee

In this case, according to Theorem \ref{the-per-breath-gas},
\begin{align}
	r_k=\frac{Q+(k-1)m}{k\pi},\quad u_k(z)=r_k\varphi(z)=\frac{1+(k-1)\frac{m}{Q}}{k\pi}\frac{|z|}{
	\sqrt{Q^2+z^2}}.
\end{align}

To compute the invariants in both cases within the sense of the thermodynamic limit, applying formula \eqref{Ik-per}, we have
\begin{align}\label{I-1}
    I_n=\begin{cases} 0,& n\quad even,\\
    \frac{1}{k} (-1)^{\frac{n+1}{2}}nd_n\le(Q^{n+1}+(k-1)m^{n+1}\ri),& n\quad odd.
    \end{cases}
\end{align}

{\it Example 2.} For the parabolic potential $q(x)=\sqrt{1-x}$ we calculate $p(\l)=1-\l$, so that
\be
S_1(\l)=\frac43(1-\l)^{\frac 32}, \quad \varphi(z)=3|z|\sqrt{1+z^2}, \quad  \nu_k (z)=\frac {\pi |z|}4 \le(2|z|^2+3(k-1)\ri).
\ee
In this case we have
\be
\s_k(z)=\frac\pi {6}\cdot \frac{2|z|^2+3 (k-1)}{\sqrt{1+z^2}}.
\ee

In this example, the limiting averaged invariants are
\begin{align}\label{I-2}
    I_n=\begin{cases} 0,& n\quad even,\\
     (-1)^{\frac{n+1}{2}}\frac{2nd_n}{k(n+3)},& n\quad odd.
    \end{cases}
\end{align}
Note for $k=1$, according to Theorem \ref{the-per-sol-gas}, the DOS $u=r\varphi=\frac{2}{\pi}|z|\sqrt{1+z^2}$.
\br
From the formula \eqref{Ik-per}, it is evident that the limiting averaged invariants $I_n$ of the $k$-dilution of periodic soliton  gases are  simply $I_n/k$, as illustrated by \eqref{I-1} with $m=0$ and \eqref{I-2}.
\er

{\it Example 3.} For the semicircle potential  $q(x)=\sqrt{1-x^2}$ we calculate $p(\l)=\sqrt{1-\l}$, so that
\begin{align*}
S_1(\l)&=2\int_0^{\sqrt{1-\l}}\sqrt{1-\l-x^2}dx\nonumber \\
&=\hf\oint\sqrt{1-\l-x^2}dx =\frac{\pi (1-\l)}{2}=\frac{\pi (1+z^2)}{2}
\end{align*}
Then 
\be
S_1(0)=\frac\pi 2, \quad S_1'(z)=\pi z \quad {\rm and ~so}  \qquad \varphi(z) = 2|z|.
\ee
Finally we obtain
\bea
S_2(\l)=\int_{\sqrt{1-\l}}^1\sqrt{x^2-(1-\l)}dx= 
\frac{\sqrt \l}{2} +\frac{1-\l}{2}\ln\frac{\sqrt{1-\l}}{1+\sqrt\l}
\eea
replacing $\l$ by $-z^2$, we have
\be
\nu(z)=\frac{1}{2}\le(|z|+(1+z^2)\ln{\frac{\sqrt{1+z^2}}{1+|z|}}\ri),\quad~\s(z)=\frac{1}{4}\le(1+\frac{1+z^2}{|z|}\ln{\frac{\sqrt{1+z^2}}{1+|z|}}\ri).\nonumber
\ee
Moreover, according to Theorem \ref{the-per-sol-gas}, we have
\begin{align}
	r=1/4,\quad u(z)=r\varphi=|z|/2.
\end{align}
In this example, the limiting averaged invariants are
\begin{align}
    I_n=\begin{cases} 0,& n\quad even,\\
    \frac{(-2)^{\frac{n+1}{2}}nd_n}{n+2},& n\quad odd.
    \end{cases}
\end{align}

\appendix

\section{Some error estimates}\label{sec-err}

We start with constructing an approximation for the period matrix (the matrix of the system \eqref{WFR}) entries of $\mathfrak R_N$. Let
$\d=\max_{|j|=1}^N |\d_j|$. We also use notation: $R_m(z)=\sqrt{(z-z_m)^2-\d_m^2}$, $m=-N,\dots,N$.

Let us recall a few facts about the thermodynamic limit. First, we assume that for a large $N\in\N$ all the bands (except the stationary band in the breather gas case) shrink around their centers $z_j$, $j=\pm 1,\dots,\pm N$ much faster than $N^{-1}$. In fact, we assume they are shrinking exponentially fast in $N$,
%$N^\beta$, $\beta>0$,  
even though some error estimates stay valid  for algebraically fast shrinking.

Secondly, we assume that in the thermodynamic limit the shrinking bands are segments and all the bands are $O(N^{-1})$ spaced. That is, 
 there exists  a constant $\varphi_0>0$ such 
that  for all $j,k$
\be\label{dist-bound}
\min_{j\neq k}|z_j-z_k|\geq \frac 3{N \varphi_0}.
\ee
By the same argument we can also require that  $ \frac 3{N \varphi_0}$ is the lower bound of the distances between 
 any $z_l$ and the exceptional band $\g_0$.
%in a small neighborhood of $\m_k(0)= z_k$.

The function $\r_N(z)$, defined by
\be\label{R-est}
R(z) =R_0(z)\prod_{|j|=1}^N(z-z_j)(1+\r_N(z)),
%~~~\text{uniformly in}~~z\in \C\setminus V,
\ee
is analytic in $\bar C\setminus \cup_{j=-N}^N \g_j$ and $\r_N(\infty)=0$.
The following lemma shows that in the thermodynamic limit $\r_N(z)$ approaches zero uniformly away from $\G$.

\bl\label{lem-Rd}
%Fix some $M>0$. 
%Let $l(\d)$ be a positive monotonically non decreasing function of $\d\in(0,1)$. Then, 
Under the thermodynamic limit assumptions for the breather gas, including \eqref{dist-bound}, for any sufficiently large $N$ we have: a) 
\be\label{rho-est}
|\r_N(z)|\leq 3\sqrt 2 e\varphi_0^2\d^2 N^2\ln N=:\r_*(\d,N)
\ee
 as long as $z$ is away from the shrinking bands, namely, 
\be\label{away}
|z-z_j|\geq \sqrt 2 |\d_j| \quad \text{ for all $j=\pm 1,\dots,\pm N$};
\ee

b) If $|z-z_j|< \sqrt 2 |\d_j| $ for some $j$, $1\leq |j|\leq N$, then
\be\label{rho-est-in}
(1+\r_N(z))^{-1}=\frac{z-z_j}{R_j(z)}(1+O(\r_*)).
\ee

\el 

\begin{proof} Part a).
 Since
 \be
R(z) =R_0(z)\prod_{|j|=1}^N(z-z_j) \prod_{|j|=1}^N\le(1-\frac{\d_j^2}{(z-z_j)^2}\ri)^\hf,
 \ee
 we need to estimate 
 \be 
 1+\r_N(z)=\prod_{|j|=1}^N\le(1-\frac{\d_j^2}{(z-z_j)^2}\ri)^\hf=e^{\hf\sum_{|j|=1}^N\ln(1-x_j)},
 \ee
  where $x_j=\frac{\d_j^2}{(z-z_j)^2}$.
 Using the obvious inequality
 $|e^x-1|\leq |x|e^{|x|}$, we have
 \be\label{ab-rho}
 |\r_N(z)|\leq \hf|\sum_{|j|=1}^N\ln(1-x_j)|e^{\hf|\sum_{|j|=1}^N\ln(1-x_j)|}.
 \ee
  According to \eqref{away}, all $|x_j|\leq \hf$.
 
 Consider now
 \be
 |\ln(1-x)|^2=\ln^2|1-x|+\arg^2(1-x)\leq |\ln(1-x)|^2+\arcsin^2|x|.
 \ee
 %According to Mean Value Theorem
 One can easily show  that
 \be
 |\ln(1-|x|)|\leq \frac{|x|}{1-|x|},~~~~~~\arcsin|x|\leq\frac{|x|}{\sqrt{1-|x|^2}}\leq \frac{|x|}{1-|x|},
 \ee
so that
\be \label{ln1-x}
|\ln(1-x)|\leq \frac{\sqrt 2|x|}{1-|x|}\leq 2\sqrt 2 |x|
\ee
provided $|x|\leq \hf$.
% \red{Note: if $|x|>1$, then $|\log(1-1/x)|\leq \frac{\sqrt{2}}{|x|-1}$}

Then 
\be
\hf|\sum_{|j|=1}^N\ln(1-x_j)|\leq \sqrt 2\sum_{|j|=1}^N|x_j|\leq \sqrt 2 \d^2 r_N(z),\quad \text{where} \quad r_N(z)=\sum_{|j|=1}^N\frac 1{|z-z_j|^2}.
\ee
Thus, condition \eqref{away} implies
\be\label{ab-rho1}
 |\r_N(z)|\leq \sqrt 2 \d^2 r_N(z)e^{\sqrt 2 \d^2 r_N(z)}.
 \ee

If $d>0$ is the distance between $z$ and $\G$ then $r_N(z)\leq \frac{2N}{d^2}$, so $|\r_N(z)|\ra 0$ very fast
as $N\ra\infty$.  Consider now $|z-z_j|=\sqrt 2 |\d_j|$ and the worse case scenario where the $2N$ centers $z_k$ pack the plane in hexagonal pattern (circles packing pattern) centered at $z_j$. Then we can estimate 
\be\label{hex}
r_n(z)\leq 6N^2\varphi_0^2\le(1+\hf+\frac 13+\dots +\frac 1n\ri )\leq 3N^2\ln N\varphi_0^2,
\ee
where $n$ is the smallest integer satisfying $n\geq\sqrt{\frac{2N}{3}} $. For a large $N$, \eqref{ab-rho1}- \eqref{hex} imply $\r_*\leq 1$ and so we obtain 
\eqref{rho-est}.

To prove part b) we notice that 
\be\label{rho-b} 
1+\r_n(z)=\frac{R_j(z)}{z-z_j}\prod_{|k|=1, k\neq j}^N\le(1-\frac{\d_k^2}{(z-z_k)^2}\ri)^\hf
 \ee
 so that 
 \be
 (1+\r_n)^{-1}=\frac{z-z_j}{R_j(z)}(1+\tilde \r_N(z))^{-1},
 %\prod_{|k|=1, k\neq j}^N\le(1-\frac{\d_k^2}{(z-z_k)^2}\ri)^{-\hf}=O(|z-z_j|)
 \ee
 where  $1+\tilde \r_N$ denotes the  product in \eqref{rho-b}.
 Now  part b) follows from the fact 
that part a) is applicable to $\tilde\r_N$.
 
% and use the arguments of part a) to estimate the product.
%so that $|\r_N(z)|\leq e^{\frac {2\sqrt 2l(\d)}N}-1<\frac {4l(\d)}N$. 
 %That proves the second part of the lemma. The first part follows immediately.
\end{proof}

\br \label{rem-R-sol-est}
In the case of  soliton gas we  introduce $\tilde{\d}=\max_{|j|=0}^N |\d_j|$. Lemma \ref{lem-Rd} is also valid in this case if
we replace $\d$ by $\tilde \d$ and $R_0(z)$ by $z-z_0$.
\er

\br\label{rem-prod-est}
We want to state separately a useful estimate based on \eqref{ab-rho}, \eqref{ln1-x}: 
\be\label{prod-est}
\le|\prod_{j=1}^N(1-x_j)-1 \ri|
\leq \sum_{j=1}^N\ln(1-x_j)|e^{\hf|\sum_{|j|=1}^N\ln(1-x_j)|}\leq 2^{\frac 32}\sum_{j=1}^N|x_j|e^{2^{\frac 32}\sum_{j=1}^N|x_j|}
\ee
provided $|x_j|< \hf$ for all $j$.
\er

\subsection{Approximation of the normalized holomorphic differentials $w_j$}

In the limiting case when all the bands $\g_j$ except, possibly, $\g_0$, shrink to points $z_j$, $j=\pm 1,\dots,\pm N$, that is, $\d=0$,
it is easy to check that 
\be\label{wj0}
w_j(z)=w_j(z,0)=-\frac{R_0(z_j)dz}{2\pi i R_0(z)(z-z_j)}=-\frac{R_0(z_j)\prod_{k\neq j}(z- z_k)dz}{2\pi i R(z)}
\ee
Let us fix some $j=\pm 1,\dots,\pm N$. In general, we have
\be\label{wjd}
w_j(z,\d)=\frac{\m_j(\d)\prod_{k\neq j}(z-\m_k(\d))dz}{ R(z)},
\ee
where ${\m_j}=\k_{j,1}$, see  \eqref{Pj}. As it follows from \eqref{wj0}, $ \m_k (0)=  z_k$ for all  $k\neq j$ 
and $\m_j(0)=-\frac{R_0(z_j)}{2\pi i}$.
Also, it is a well known
fact (see, for example, \cite{FK}) that for any nondegenerate genus $2N$ hyperelliptic Riemann surface $\mathfrak{R}$ there exists a unique collection
of the normalized holomorphic differentials $w_j$, $j=0,\pm 1,\dots,\pm N$. Here we assume that $\d_j^2=\phi_j(\d^2)$,  where all $\phi_j\in C^2[0,\d_*]$
and their norms are uniformly bounded with respect to $N$. 
%for the case of breather gas. 
We also assume  $\phi'_0(0)= 0$.

In the following lemma we estimate the deformation of   $\m_k(\d)$ for sufficiently  small $\d$. Here we assume that the set of all $z_j$ is bounded
and \eqref{dist-bound} holds.

We use notations $\vec \m (\d),\vec \eta$
for $2N$ dimensional vectors $\vec \m (\d)=(\m_{-N}(\d),\dots, \m_{N}(\d))^t$ and    $\vec \eta= (z_{-N},\dots, z_{N})^t$.

\bl\label{lem-muj} Let us fix an arbitrary $j=\pm 1,\dots,\pm N$ and consider all $\m_k(\d)$ defined by \eqref{wjd}. Then for all
 $j=\pm 1,\dots,\pm N$,  all $k=\pm 1,\dots,\pm N$ and  all  $\d$ satisfying 
 \be\label{d-assump}
 \d=o(N^{-6}),
 \ee
  in the thermodynamic limit   we have
\be\label{mk-est}
|\mu_k(\d)- z_k|\leq C N^4 \d^2,~~~k\neq j~{\rm and~}~~|\mu_j(\d)+\frac{R_0(z_j)}{2\pi i}|\leq C N^4 \d^2,
\ee
where the constant $C>0$ does not depend on $N,j,k,\d$.
\el
    \begin{proof}
WLOG, we can assume that all contours $\gt_k$ satisfy the condition \eqref{away} from Lemma \ref{lem-Rd} and, thus, the estimate \eqref{rho-est} from this lemma is valid on  
$\gt_k$.

Fix some $j=\pm 1,\dots,\pm N$.  By definition, 
\be\label{wj-def}
F_{kj}(z;\vec\m,\d):=\oint_{\gt_k}w_j(z,\d)=\d_{k,j},
 \ee
 where $\d_{k,j}$ denotes the Kronecker symbol.
  Then
 %by Implicit Function Theorem,
 \be\label{implF}
 \frac{d\vec\m}{d\d^2}=-\le(\frac{\part \vec F_j}{\part \vec\m}\ri)^{-1} \cdot \frac{\part \vec F_j}{\part \d^2} 
 \ee
 where $\vec F_j$ is the $j$th column of the matrix  $F_{kj}$.
 By Implicit Function Theorem, $\vec \m(\d)$ is uniquely defined and differentiable in some neighborhood of $\vec \m(0)$, provided  that 
 $\frac{\part \vec F_j}{\part \vec\m}$
 is invertible at $\d=0$.
 We start with calculating the latter matrix and its inverse.  
 
 Indeed,
 \be\label{fkjmm0}
 \le. \frac{\part  F_{kj}}{\part \m_m}\ri|_{\d=0} =\frac{R_0(z_j)\d_{km}}{R_0(z_m)(z_j-z_m)}-\frac{\d_{kj}}{z_j-z_m}
 \ee 
 when $m\neq j$ and 
  \be 
 \frac{\part  F_{kj}}{\part \m_j}|_{\d=0}=\oint_{\gt_k} \frac{\part   w_j(z,\d)}{\part \m_j}|_{\d=0}=-\oint_{\gt_k} \frac{dz}{ R_0(z)(z-z_j)}
 =\frac{2\pi i\d_{kj}}{R_0(z_j)}.
 \ee

 So,  the matrix $ \le.\frac{\part \vec F_j}{\part \vec\m}\ri|_{\d=0}$ is the sum of the main diagonal 
 \be\label{dFdm}
 {\rm diag}\le.\frac{\part \vec F_j}{\part \vec\m}\ri|_{\d=0}={\rm diag}\le(\frac{R_0(z_j)}
 {R_0(z_{-N})(z_j-z_{-N})},\dots, \frac{-2\pi i}{R_0(z_j)},\dots, \frac{R_0(z_j)}
 {R_0(z_{N})(z_j-z_{N})}\ri)
 \ee
 and the  $j$th  column $\le( \frac{-1}{z_j-z_{-N}},\dots, \frac{-1}{z_j-z_{N}} \ri)^t$, where the $j$th entry should be taken zero.
 Thus, $\le.\frac{\part \vec F_j}{\part \vec\m}\ri|_{\d=0}$ is an invertible matrix, so that the  Implicit Function Theorem is applicable to \eqref{wj-def} for any fixed $j$.
 
 The inverse $\le.\le(\frac{\part \vec F_j}{\part \vec\m}\ri)^{-1} \ri|_{\d=0}$ has the same structure as $\le.\frac{\part \vec F_j}{\part \vec\m}\ri|_{\d=0}$ with
 \be\label{dFdm-inv_d}
  {\rm diag}\le.\le(\frac{\part \vec F_j}{\part \vec\m}\ri)^{-1} \ri|_{\d=0}={\rm diag}\le(\frac{ R_0(z_{-N})(z_j-z_{-N})}{ R_0(z_j)},\dots, \frac{R_0(z_j)}{-2\pi i},\dots, \frac{ R_0(z_{N})(z_j-z_{N})}{ R_0(z_j)}\ri),
 \ee
 and the $j$th column
 %\be\label{dFdm-inv_j}
 $\le( R_0(z_{-N}),\dots,  R_0(z_{N}) \ri)^t$, where the $j$th entry should be taken zero.  Note that 
 \be\label{inv-mat-est}
 \le.\le(\frac{\part \vec F_j}{\part \vec\m}\ri)^{-1} \ri|_{\d=0}=O(N)
 \ee 
as $N\ra\infty$ uniformly in all the entries.
 
Our goal is estimate the growth of 
 $\vec\e(\d)=\vec\m(\d)-\vec\m(0)$
 in a neighborhood $W_\m(N)$ of $\vec\m(0)$. We define $W_\m(N)$ as a centered at $\vec\m(0)$ ``scaled cube", of size $O(N^{-7})$ in the direction of each component. We also introduce the neighborhood $W(N)=W_\m(N)\times W_\d(N)$, where $W_\d(N)=[0,o(N^{-6}))$. 
 We now estimate the factors in the right hand side of \eqref{implF} for $\vec\m\in W_\m(N)$

 %and    the both side of \eqref{implF} are evaluated at $\d=0$.
 
 First assume that $m\neq j$. Then  we have
 \bea\label{fkjmm1}
 \frac{\part  F_{kj}}{\part \m_m}=\oint_{\gt_k} \frac{\part   w_j(z,\d)}{\part \m_m}={-\m_j(\d)}\oint_{\gt_k} \frac{\prod_{l\neq m,j}(z-\m_l(\d)) dz}{ R(z)}=\cr
 {-\m_j(\d)}\le[\oint_{\gt_k} \frac{dz}{R_0(z)R_j(z)R_m(z)} +\oint_{\gt_k} \frac{dz\le(\prod_{l\neq m,j}\frac{z-\m_l(\d)}{R_l(z)}-1\ri)}{R_0(z)R_j(z)R_m(z)}\ri].
 \eea
Direct  calculation shows that 
%the former integral can be written as
\begin{align}\label{fkjmm2}
-\m_j(\d)\oint_{\gt_k} \frac{dz}{R_0(z)R_j(z)R_m(z)}=
 \le. \frac{\part  F_{kj}}{\part \m_m}\ri|_{\d=0}-\e_j(\d)\oint_{\gt_k} \frac{dz}{R_0(z)R_j(z)R_m(z)}- \cr
{\m_j(0)}\oint_{\gt_k} \frac{\le[\le(1-\frac{\d_j^2}{(z-z_j)^2}\ri)^{-\hf}\le(1-\frac{\d_m^2}{(z-z_m)^2}\ri)^{-\hf} -1\ri]   dz}{R_0(z)(z-z_j)(z-z_m)}.
\end{align}

We roughly estimate the second term in the right hand side of \eqref{fkjmm2} as
$
2\pi\e(\d)N^2 \varphi^2_0,
$
where 
\be\label{e-def}
\e(\d):=\max_k\e_k(\d)= \max_k\{|\m_k(\d) -\m_j(0)|\}
 \ee
 for all $j,k$ and all $\x\in[0,\d]$. 
%where the constant $\varphi_0$ is defined by the requirement that  $\frac 3{N \varphi_0}$ is the lower bound of the distances between any pair of $\eta_l$
%and between any $\eta_l$ and $\a_1,\a_0$.
 %Using the calculations from Lemma \ref{lem-Rd}, 
 We also estimate the last term in  \eqref{fkjmm2} as
 $|R_0(z_j)|\varphi_0^3N^3\d^2$.
 
 Then
 \begin{align}
     \le| {-\m_j(\d)}\oint_{\gt_k} \frac{dz}{R_0(z)R_j(z)R_m(z)} - \le(\frac{R_0(z_j)\d_{km}}{R_0(z_m)(z_j-z_m)}-\frac{\d_{kj}}{z_j-z_m}\ri)\ri|\leq \cr
\varphi^2_0N^2[ 2\pi   \e(\d)+\varphi_0N|R_0(z_j)|\d^2].
  \end{align}
 Using Lemma \ref{lem-Rd} and Remark \ref{rem-prod-est}, the product in the last term from \eqref{fkjmm1} can be estimated as
 \bea \label{2nd-est}
 \le|\prod_{l\neq m,j}\frac{z-\m_l(\d)}{R_l(z)}-1\ri|\leq
\le|\prod_{l\neq m,j}\frac{z-\m_l(\d)}{(z-z_l)}(1+\tilde\rho_N)^{-1}-1\ri|\leq\cr
 %\le|\prod_{l\neq m,j}\frac{z-\m_l(\d)}{(z-z_l)}(1+\rho_N)^{-1}-1\ri| \leq\cr
 \le|\prod_{l\neq m,j}\le(1-\frac{\e_l(\d)}{(z-z_l)}\ri)(1+\tilde\rho_N)^{-1}-1\ri| \leq   \cr 
 2 e^{2^{\frac 32}\e(\d)\varphi_0N^2}[\r_*+\sqrt 2 \e(\d)\varphi_0N^2]
 \eea 
 %does not exceed 
 %\red{\be
 %2\sqrt 2|R_0(z_j)|\varphi_0^2N^2 ( e^{4N^2\varphi_0[\e(\d)+\varphi_0N\d^2] }-1)\leq 16|R_0(z_j)|\varphi_0^3N^4[\e(\d)+\varphi_0N\d^2].\ee}
 provided $\e(\d)N^2\ra 0 $ as $N\ra\infty$.
 Here $\tilde\r$ is the same as in  \eqref{R-est} exception the factors $R_m$ and $R_j$ were removed from $R$. Of course, $\tilde\r_N$ also satisfies the estimate \eqref{rho-est}.

 Thus,
 the last term from \eqref{fkjmm1} can be estimated by
 \be 
 2\varphi^2_0N^2|R_0(z_j)|[\rho_*+2\e(\d)\varphi_0N^2].
 \ee

 So, using \eqref{rho-est},
  \be\label{fkjmm3}
\le| \frac{\part  F_{kj}}{\part \m_m}- 
 \frac{\part  F_{kj}}{\part \m_m}|_{\d=0}\ri|\leq
  4\varphi^2_0N^2[ (\pi +2\varphi_0N^2) \e(\d)+3\sqrt 2\varphi_0^2N^2\ln N|R_0(z_j)|\d^2]
 %
 %\red{17\pi|z_j|\varphi_0^3N^4[\e(\d)+\varphi_0N\d^2]}
 \ee
 for sufficiently large $N$ provided $\e(\d)N^2\ra 0 $ as $N\ra\infty$. 
 Of course, the latter condition holds when $\vec\m\in W_\m(N)$.
 
 Direct calculations show that in the case $m=j$ we have
 \bea\label{fkjmm3a}
 \frac{\part  F_{kj}}{\part \m_j}- 
 \frac{\part  F_{kj}}{\part \m_j}|_{\d=0}=
\oint_{\gt_k}\frac{dz}{R_0(z)}\le[ \frac{\prod_{l\neq j}\frac{z-\m_l(\d)}{R_l(z)}}{R_j(z)} -\frac{1}{z-z_j}\ri]=\cr
\oint_{\gt_k}\frac{dz}{(z-z_j)R_0(z)}\le[\le(1-  \frac{\d_j^2}{(z-z_j)^2}\ri)^{-\hf} \prod_{l\neq j}\frac{z-\m_l(\d)}{R_l(z)} -1\ri]
 \eea
 Applying to \eqref{fkjmm3a} similar  estimates as to \eqref{fkjmm3}, we obtain that the 
 estimate \eqref{fkjmm3} also covers the case $m=j$, which makes it uniform in indices $k,j,m$.

 Now we calculate 
 \be\label{drd_del}
 \frac{\part }{\part \d^2}\frac{1}{R(z)}=  \frac{1}{2R^3(z)}\sum_{m=-N}^N\prod_{n\neq m}R_n^2(z)\phi'_m(\d^2)
= \hf\sum_m\frac{\phi'_m(\d^2)}{R_m^2(z)R(z)},
 \ee
% where the  summation $\sum'_m$ can be taken only over all $m$ such that $\phi'_m(0)\neq 0$\red{???}. 
We remind that by assumption $\phi'_0(0)= 0$. Then 
 \be
  \frac{\part F_{kj}}{\part \d^2}=  \oint_{\gt_k}\frac{\part w_j}{\part \d^2}=
\mu_j(\d)\oint_{\gt_k}   \hf \sum_m\frac{\prod_{l\neq j}\frac{z-\m_l(\d)}{R_l(z)}\phi'_m(\d)dz}{R_m^2(z)R_0(z)R_j(z)}. 
\ee

We also have 
\be\label{drd-kj0}
 \le. \frac{\part F_{kj}}{\part \d^2}\ri|_{\d=0} =\frac{R_0(z_j) \phi'_k(0)}{2R_0( z_k)(z_j- z_k)}\le[\frac{1}{ z_k-z_j}+\frac{ z_k}{R^2_0( z_k)}\ri]
\ee
 in the case  $k\neq j$ and
 \bea\label{k=j}
 \le. \frac{\part F_{jj}}{\part \d^2}\ri|_{\d=0}=  
  \le. -\frac{R_0(z_j)}{2\pi i}\oint_{\gt_j}   \hf \sum_{m}\frac{\phi'_m(\d)dz}{R_m^2(z)R_0(z)R_j(z)}\ri|_{\d=0} 
  = \cr
  \le.\hf \sum_{m\neq j}\frac{\phi'_m(0)}{(z_j-z_m)^2}+ \frac{R_0(z_j)\phi'_j(0)}{4}\le(\frac 1{R_0(z)}\ri)''\ri|_{z=z_j}.
\eea
in the case $k=j$. Similarly to \eqref{fkjmm1}, we have 
\bea\label{drd-est0}
 \frac{\part F_{kj}}{\part \d^2}-  \frac{\part F_{kj}}{\part \d^2}|_{\d=0}=
\hf \sum_{m}\oint_{\gt_k}\le[ \frac{\m_j(\d)\prod_{l\neq j}\le(1-\frac{\e_l(\d)}{z-z_l}\ri)(1+\tilde\r_N)^{-1}\phi'_m(\d)}{R_m^2(z)R_0(z)R_j(z)}\ri. \cr
\le.  -
\frac{\m_j(0)\phi'_m(0)}{R_0(z)(z-z_m)^2(z-z_j)}
\ri]dz,
\eea 
where $\tilde\r$ is defined in the same way as  in \eqref{2nd-est}.

 To estimate \eqref{drd-est0}, we first observe that 
 \be\label{drd-est0-1}
 \le|\frac{\m_j(\d)}{2} \sum_{m}\oint_{\gt_k} \frac{\prod_{l\neq j}\le(1-\frac{\e_l(\d)}{z-z_l}\ri)\tilde\r_N\phi'_m(\d)dz}{R_m^2(z)R_0(z)R_j(z)}\ri|\leq 2 |R_0(z_j)|\phi'\varphi_0^3N^4\r_*,
 \ee
 where 
 \be 
 \phi'=\max_m\sup_{W_\d(N)}|\phi'_m(\d)|.
 \ee 
 The next term in \eqref{drd-est0} to estimate is
 \be \label{drd-est0-2}
 \le|\frac{\e_j(\d)}{2} \sum_{m}\oint_{\gt_k} \frac{\prod_{l\neq j}\le(1-\frac{\e_l(\d)}{z-z_l}\ri)\phi'_m(\d)dz}{R_m^2(z)R_0(z)R_j(z)}\ri|\leq
 2 \pi\e(\d)\phi'\varphi_0^3N^4.
 \ee 
 One more term in \eqref{drd-est0} to estimate is
 \be \label{drd-est0-3}
 \le|\frac{\m_j(0)}{2} \sum_{m}\oint_{\gt_k} \frac{\prod_{l\neq j}\le[\le(1-\frac{\e_l(\d)}{z-z_l}\ri)-1\ri]\phi'_m(\d)dz}{R_m^2(z)R_0(z)R_j(z)}\ri|\leq 4|R_0(z_j)|\e(\d)\phi'\varphi_0^4N^4,
 \ee
 where we used the same considerations as in estimate ]\eqref{2nd-est}. Finally,
 the last term in \eqref{drd-est0} to estimate is
 
 \bea \label{drd-est0-4}
 \le|\frac{\m_j(0)}{2} \sum_{m}\oint_{\gt_k}dz\le[ \frac{\prod_{l\neq j}\le(1-\frac{\e_l(\d)}{z-z_l}\ri)\phi'_m(\d)}{R_m^2(z)R_0(z)R_j(z)}-
 \frac{\phi'_m(0)}{R_0(z)(z-z_m)^2(z-z_j}
 \ri]\ri|\leq \cr
 4|R_0(z_j)|\varphi_0N(1+ \phi'\varphi_0N)\d^2.
 \eea
 
Since the sum of the left hand sides of \eqref{drd-est0-1}-\eqref{drd-est0-4} gives the absolute value of \eqref{drd-est0}, 
  we obtain 
\bea\label{drd-est}
\le| \frac{\part F_{kj}}{\part \d^2}-  \frac{\part F_{kj}}{\part \d^2}|_{\d=0}\ri|\leq
2\phi'\varphi_0^2N^2[|R_0(z_j)|\varphi_0N^2\r_*+ \cr
\varphi_0N^2(\pi+2|R_0(z_j)|\varphi_0)\e(\d)+3 |R_0(z_j)|\d^2].
%\red{16\pi|z_j|\varphi_0^3N^4[\e(\d)+\varphi_0N\d^2]+
%2\phi''\d^2 \le. \frac{\part F_{kj}}{\part \d^2}\ri|_{\d=0}} 
\eea

 Similar estimate is valid for $k=j$.

Equations \eqref{drd-kj0}, \eqref{k=j} imply that  the vector $\le.\frac{\part \vec F}{\part \d^2}\ri|_{\d=0}$ are of the order  $O(N^3)$
uniformly in $k,j$. That will  also be true for  $\frac{\part \vec F}{\part \d^2}$ provided  the error term \eqref{drd-est} is of the same or a smaller order. (Here and henceforth all 
the estimates are  entry wise.)
But the condition $(\vec\m,\d)\in W(N)$
implies the required estimate.
Then the error term \eqref{drd-est} is of the  order $O(N^{-3})$ uniformly in $k,j$.

Let $\frac{\part \vec F_j}{\part \vec\m}=A_0+\D A$, where 
$
A_0=\le.\frac{\part \vec F_j}{\part \vec\m}\ri|_{\d=0}.
$
Then we can rewrite \eqref{implF} as
\be\label{implF-end}
 \frac{d\vec\m}{d\d^2}=- (\1+A_0^{-1}\D A)^{-1}A_0^{-1} \frac{\part \vec F_j}{\part \d^2}. 
 \ee
According to \eqref{inv-mat-est}, $A_0^{-1}=O(N)$.
It also follows then from \eqref{d-assump}, \eqref{dFdm-inv_d} and $\vec\m\in W_\m(N)$ that 
$A_0^{-1} \frac{\part \vec F}{\part \d^2}=O(N^4)$ and 
$\D A=O(N^{-3})$
uniformly in $k,j$. Thus, $A_0^{-1}\D A=O(N^{-2})$ uniformly in $k,j$. Now, by Gershgorin Circle Theorem,
see, for example, \cite{Gant}, $(\1+A_0^{-1}\D A)^{-1}=O(1)$. So, condition $\vec\m\in W_\m(N)$ implies that 
\be\label{dmdd2-est}
\frac{d\vec\m}{d\d^2}= O(N^4)
\ee
uniformly in $k,j$ when $(\vec\m,\d)\in W(N)$.

According to  
 the Mean Value Theorem, 
\be\label{dmdd2-mvt}
\m_k(\d)- \m_k(0)=\frac{d\m_k}{d\d^2}(\xi_k)\d^2
%= O(N^{-7})
\ee 
for all $k$, where $\xi_k\in (0,\d)$.
Let us start deforming $\d$ from $\d=0$ as long as $(\vec\m(\d),\d)\in W(N)$.
Then, according to \eqref{dmdd2-mvt},
 \be \label{est-end}
 |\m_k(\d)- \m_k(0)|=o(N^{-8}),
 \ee 
 that is, $\d\in W_\d(N)$ guarantees that $\vec \m\in W_\m(N)$.

Thus, \eqref{mk-est} follow from \eqref{est-end} and  we proved the lemma.
\end{proof}

    \subsection{Approximation of periods of $\Rscr$ }\label{sec-appr-per}
 
 Approximation of the coefficients of the linear system \eqref{WUPjM} is, in fact, an approximation   
 of the period matrix of the hyperelliptic Riemann surface $\mathfrak R_N$. Since {\bf B}-cycles are crossing small shrinking bands $\g_k$,
 the following formula
 \be\label{log-formula}
\int_\mathcal U \frac{\phi(\z)d\z}{\sqrt{\z^2+\d^2}}=-2\ln|\d|\phi(0)+O(1)
\ee
where   $\d\ra 0$ and for $\phi(\z)$ is continuous and Lipschitz at $\z=0$,   will be used  in the calculations below. Here $\d\in \C$ and 
$\mathcal U$ denotes  a fixed segment in a neighborhood of $\z=0$  passing  through the origin and intersecting the segment $[-i\d,i\d]$ transversely
from left to right. Note that in \eqref{log-formula} we assume that $\sqrt{\z^2+\d^2}$ is positive on $\R$, that is, the branch cut of  $\sqrt{\z^2+\d^2}$
goes from $i\d$ to $-i\d$ through infinity.

Let us denote the cycle $\tilde B_k= B_k\cup B_{-k}$, $k=1,\dots, N$.
To simplify the exposition of the following Lemma \ref{lem-est-evenB}, we assume that $\G^+$ is a 1D compact (contour).

 \bl\label{lem-est-evenB} Under the thermodynamic limit assumptions  for all $k,j=1,\dots, N$ we have 
\bea\label{lead-int-ass-evB}
\oint_{\tB_k}\frac{P_j(\z)d\z}{R(\z)}=&&\frac{1}{i\pi}
\le[ \ln\frac{R_0(z_j)R_0( z_k)+z_jz_k-\d_0^2}
{R_0(z_j)R_0(\bar z_k)+z_j\bar z_k-\d_0^2}
- \ln \frac{z_k-z_j}{\bar z_k-z_j}\ri] +h(k-j)\cr +O\le( N^2\d^{\frac 23}\ri) ~~
&& \text{when}~~~k\neq j~~\text{and}~~ \oint_{\tB_j}\frac{P_j(\z)d\z}{R(\z)}=-\frac{2\ln|\d_j|}{i\pi}+O(1)       \eea
in the leading order as  $N\ra\infty$ provided that  $\d=o(N^{-6})$.  Here $h$ denotes the Heaviside function 
$h(\xi):=\hf(1+ {\rm sign}\,\x)$.
Equations \eqref{lead-int-ass-evB} also stay true  when $\d_0\ra 0$ provided $\d \ll |\d_0|$.
\el
%\red{[The variables in the ``small'' log term where revised. It does not affect the situation when we look for the real part of the log,but affects the imaginary part. Adjust in the rest of the paper as necessary.]}

\begin{proof}
Consider first the case 
$j\neq k$.
Deforming the contours $B_k\cup B_{-k}$ and using the fact that the values of the integral over each sheet of $\mathfrak R_N$ are equal,  we obtain 
\bea\label{int-not-own-even}
\oint_{\tB_k}w_j =\oint_{\tB_k}\frac{P_j(\z)d\z}{R(\z)}
=2\int_{\bar z_{k}}^{z_{k}}\frac{P_j(\z)d\z}{R(\z)} +h(k-j)= \cr
2\m_j\int_{\bar z_{k}}^{z_{k}}\frac{dz}{R_0(z)R_j(z)}\prod_{m\neq j} \frac{z-\m_m}{R_m(z)}+h(k-j),
\eea
where the contour connecting $\bar z_{k}$ and $z_{k}$ in the latter integral is bent, if necessary, to be at least $\frac 3{N \varphi_0}$ away from 
any band $\g_j$ with $j\neq k$, see \eqref{dist-bound}. We can also assume that the lengths  of these contours for all $N$ are uniformly bounded.
%
%We can always assume (for any shape of $\g$) that for any $z$ on $[\bar z_k, z_k]$ we have $|z-z_m|>\frac cN$ for some $c>0$, $m\neq \pm k$
The requirement $\d=o(N^{-6})$ is needed to use Lemma \ref{lem-muj}.
Now, it follows from \eqref{2nd-est}, \eqref{rho-est} and %the arguments from   Lemma \ref{lem-muj} and using 
\eqref{mk-est} that
\be\label{est-t2}
%\le|
\le|\prod_{{m\neq j}}\frac{z-\m_m}{R_m(z)}-1\ri|\leq C_1 N^6\d^2, \quad N\ra\infty,
\ee
for some $C_1>0$ uniformly in $k,j$ as long as $z$ is  $\frac 3{N \varphi_0}$ away from and band $\g_j$. The latter requirement is violated for $\frac{z-\m_k}{R_k(z)}$ near $z= z_k$
and for  $\frac{z-\m_{-k}}{R_{-k}(z)}$ near $z=z_{-k}$. 
Therefore,
we  split the latter  integral in \eqref{int-not-own-even} into 3 parts: small $\e>0$ neighborhoods of each of the endpoints $ z_k,z_{-k}$ and 
the rest of the contour.
The value  of $\e$ should satisfy $\e=o(N^{-1})$ and $\d= o(\e)$, but its exact order will be determined below.
Using the decomposition 
\be\label{rep-1}
\frac{z-\m_{\pm k}}{R_{\pm k}(z)}=\frac{z-z_{\pm k}}{R_{\pm k}(z)}+\frac{z_{\pm k}-\m_{\pm k}}{R_{\pm k}(z)}, 
\ee
Lemma \ref{lem-muj}, \eqref{log-formula}, \eqref{est-t2} and
the rough estimate 
\be\label{est-t1}
\frac 1 {|R_0(z)R_j(z)|}\leq \varphi_0^2N^2
\ee
on the contour of integration, we estimate the integrals over the neighborhoods of $ z_k,z_{-k}$ as 
\be\label{band-est}
O(N^2\e)
%\red{[O(N^8\d^2\e)]}
+O(N^{12}\d^4|\ln\d|), \quad N\ra\infty,
\ee
%where $O(N^2\e)$ and $O(N^{12}\d^4|\ln\d|)$ 
that correspond to the first and second  terms of \eqref{rep-1} respectively.

This estimate
is uniform in $\e,k,j$.
In  the first term of \eqref{band-est}
we used the fact
that both $\frac{z- z_k}{R_k(z)}$,  $\frac{z-z_{-k}}{R_{-k}(z)}$ are bounded near the points $\bar z_k, z_k$.

According to \eqref{est-t2}, \eqref{est-t1}, the integral over the remaining part of the contour can be represented as
\be
\label{red-int}
2\mu_j(\d)\int\frac{dz}{R_0(z)R_j(z)}+O(N^{8}\d^2).
\ee
Because of
\be
 \frac{1}{R_j(z)}=\frac 1{z-z_j}\le[1+O\le(\frac{\d^2}{\e^2}\ri)\ri] %\blue{\le[1+O(N^2\d^2)\ri]}
\ee
and \eqref{mk-est}, we can rewrite \eqref{red-int} as
\be \label{est-t5}
\frac{-2R_0(z_j)}{2\pi i}\int\frac{\le[1+O\le(\frac{\d^2}{\e^2}\ri)\ri]dz}{R_0(z)(z-z_j)} +O(N^8\d^2)
\ee

In view of
 the anti-derivative
{\be\label{anti-der}
\int\frac{d\z}{R_0(\z)(\z-\eta)}=-\frac{1}{R_0(\eta)}\ln\frac{R_0(\eta)R_0(\z)+\z\eta-\d_0^2}{\eta-\z}~~~
\ee}
% where $R_0(\z)=\sqrt{(\z-\eta_0)^2-\d_0^2}$, 
we obtain the leading order term of the first equation \eqref{lead-int-ass-evB} if we substitute the limiting values $ z_k,z_{-k}$ in the anti derivative 
\eqref{anti-der}  in \eqref{est-t5}. Since these limits of integration are distance $O(\e)$ away from the endpoints of the integral in \eqref{red-int} 
 we have introduced  an error of $O(N^2\e)$, see estimate \eqref{est-t1}. The error coming from the $O\le(\frac{\d^2}{\e^2}\ri)$ term of
 \eqref{est-t5} can be estimated as $O\le(\frac{N^2\d^2}{\e^2}\ri)$.

Now, to find the best value of $\e$ in the partition of the contour $[\bar z_k, z_k]$, we equate the errors  $O(N^2\e)$ and $ O(\frac{N^2\d^2}{\e^2})$ from
\eqref{est-t5}. That yields $\e=\d^{\frac 23}$. Thus, the error in \eqref{lead-int-ass-evB} is the maximum of  $O(N^2\d^{\frac 23})$, $O(N^8\d^2)$.
In view of \eqref{d-assump}, the first term is larger. Thus,
we have completed the proof of the first equation  \eqref{lead-int-ass-evB}.

Consider now 
\be\label{int-not-own-even_k=j}
\oint_{\tB_j}w_j=
%=\oint_{\tB_k}\frac{P_j(\z)d\z}{R(\z)}=2\int_{\bar\eta_{k}}^{\eta_{k}}\frac{P_j(\z)d\z}{R(\z)}= 
2\m_j\int_{\bar z_j}^{ z_j}\frac{dz}{R_0(z)R_j(z)}\prod_{m\neq j} \frac{z-\m_m}{R_m(z)}=
\frac{-2R_0(z_j)}{2\pi i}\int_{\bar z_j}^{ z_j}\frac{dz}{R_0(z)R_j(z)} (1+O(N^6\d^2)), 
\ee
where we have used Lemma \ref{lem-muj} and the estimate from \eqref{est-t2}. Now the second equation \eqref{lead-int-ass-evB} follows from 
\eqref{log-formula} taken with the opposite sign. This choice of the sign comes from the fact that 
%the branchcut of $R_j(z)$ is $\g_j$ (which does not go through infinity as in \eqref{log-formula}) and that 
the branch of  $R_j(z)$ in  \eqref{int-not-own-even_k=j} corresponds to $-\sqrt{\z^2+\d^2}$ in
\eqref{log-formula}.
\end{proof}

\br \label{rem-higher-ac} Higher accuracy in the second equation of \eqref{lead-int-ass-evB} can be achieved if we consider
higher order terms in the small $\d_j$ expansion of the elliptic integral $\int_{\bar z_j}^{ z_j}\frac{dz}{R_0(z)R_j(z)}$
from \eqref{int-not-own-even_k=j}.
\er

\end{document}